\documentclass[a4paper,onecolumn,10pt]{quantumarticle}
\pdfoutput=1

\usepackage[utf8]{inputenc}
\usepackage{graphicx}
\usepackage{dcolumn}
\usepackage{bm}
\usepackage{bbm}
\usepackage{dsfont}
\usepackage{amsmath}
\usepackage{amssymb}
\usepackage{physics}
\usepackage{orcidlink}
\usepackage{amsthm}
\usepackage{comment}
\usepackage{tikz}
\usetikzlibrary{arrows.meta, positioning}
\usepackage[capitalise]{cleveref}
\newtheorem{theorem}{Theorem}

\newtheorem{lemma}{Lemma}
\newtheorem{example}{Example}

\newtheorem{proposition}{Proposition}
\newtheorem{definition}{Definition}

\newcommand{\id}{\mathbbm{1}}

\usepackage[capitalise]{cleveref}

\begin{document}

\title{A journey through Flatland: What does the antiflatness of a spectrum teach us?}

\author{Barbara Jasser\orcidlink{0009-0004-5657-2806}}
\email{b.jasser@ssmeridionale.it}
\affiliation{Scuola Superiore Meridionale, Largo S. Marcellino 10, 80138 Napoli, Italy}
\affiliation{Istituto Nazionale di Fisica Nucleare (INFN), Sezione di Napoli, Italy}
\author{Daniele Iannotti\orcidlink{0009-0009-0738-5998}}
\email{d.iannotti@ssmeridionale.it}
\affiliation{Scuola Superiore Meridionale, Largo S. Marcellino 10, 80138 Napoli, Italy}
\affiliation{Istituto Nazionale di Fisica Nucleare (INFN), Sezione di Napoli, Italy}
\author{Alioscia Hamma\orcidlink{0000-0003-0662-719X}}
\email{alioscia.hamma@unina.it }
\affiliation{Scuola Superiore Meridionale, Largo S. Marcellino 10, 80138 Napoli, Italy}
\affiliation{Istituto Nazionale di Fisica Nucleare (INFN), Sezione di Napoli, Italy}
\affiliation{Dipartimento di Fisica `Ettore Pancini', Universit\`a degli Studi di Napoli Federico II, Via Cintia 80126, Napoli, Italy}

\maketitle

\begin{abstract}
We explore the concept of \textit{antiflatness} to characterize the structural fluctuations within the entanglement spectrum of a quantum state (i.e., the spectrum of its reduced density operator). 
As a measure of the interplay between entanglement and magic, two fundamental quantum resources, antiflatness provides second-order information about quantum correlations that standard average measures fail to capture. 
Recognizing that standard majorization theory fundamentally orders states by purity and is structurally blind to spectral fluctuations, we introduce a novel partial ordering known as antiflat majorization, based on the Rényi entropy spread. We define Flatness-Preserving Operations (FPOs), establishing new necessary conditions for state convertibility. 
Furthermore, we unify different measures of antiflatness—such as Capacity of Entanglement, Linear Rényi spread, and Logarithmic antiflatness—using the frameworks of escort distributions and Bregman divergences. 
We prove that the Capacity of Entanglement can be expressed as a second derivative of the Kullback-Leibler divergence along the escort trajectory, connecting it with the Quantum Fisher Information. Finally, we demonstrate that absolute maximal antiflatness is not achieved by a single universal state, but rather by a continuous Pareto frontier of extremal states with jump spectra, and we analyze the typicality of these spectral fluctuations using Haar, Bures-Hall and t-doped Clifford random state ensembles.
\end{abstract}

\section{Introduction}

\begin{figure*}[t]  
    \centering
    \includegraphics[width=\textwidth]{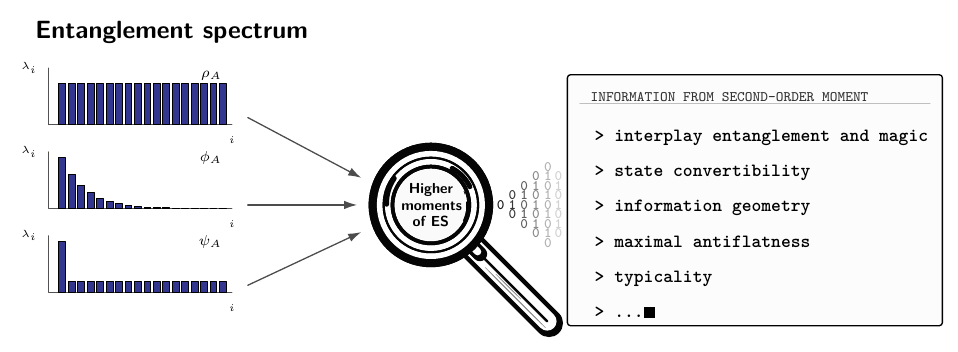}
    \caption{On the left different profiles of the eigenvalues $\{\lambda_i\}_{i=1}^r$ of a given reduced density matrix, with $\lambda_k \geq \lambda_j$ with $k<j$, showing (from top to bottom) the constant uniform, exponentially decaying and maximum antiflat spectrum respectively.
    On the right, a non-exhaustive list of information that can be extracted from the second-order moment of the ES.}
    \label{fig:spectra}
\end{figure*}
Correlations between systems are a cornerstone of quantum mechanics. Entanglement, in particular, provides a fundamental characterization of non-classical correlations among subsystems~\cite{Horodecki_2009, eisert2010colloquium, amico2008entanglement}. A key aspect of bipartite entanglement is encoded in the spectrum of the reduced density matrix (or of the associated modular Hamiltonian), namely the \textit{entanglement spectrum} (ES) \cite{PhysRevLett.101.010504}, whose eigenvalues contain information about the structure of quantum correlations. The ES has proved to be a powerful diagnostic tool: it allows one to identify topological order through edge-state correspondence~\cite{li2008entanglement, fidkowski2010entanglement, pollmann2010entanglement}, extract scaling dimensions in conformal field theories~\cite{calabrese2008entanglement, lauchli2013operator}, and distinguish ergodic thermalization from many-body localization~\cite{yang2015two, geraedts2016many}. The most widely used entanglement quantifiers, such as the von Neumann entropy or the subsystem purity between two parties, provide scalar summaries of the ES.
These correspond to the first moment of the ES and therefore provide a way to measure the average amount of quantum correlations between subsystems~\cite{de2019aspects}. Evaluating only the average of the ES, entanglement is blind to its internal fluctuations, it does not fully characterize the spectrum structure. Distinct spectra may share the same entropy or purity while displaying different distributions of spectral weights.

Although entanglement is a fundamental signature of quantumness and, in particular, a necessary ingredient for the violation of Bell inequalities~\cite{bell1964on, bell1966problem}, it captures only one aspect of quantum complexity. A complementary layer is associated with the failure of a quantum state to admit an efficient description within the stabilizer formalism, or equivalently, with its deviation from Clifford dynamics. This notion is formalized in the resource theory of non-stabilizerness, commonly referred to as magic~\cite{bravyi2007upper, leone2022StabilizerRenyiEntropy, zhu2016CliffordGroupFails, howard2017application, Turkeshi_2025, Magni_2025, Magni_2025_v1}. Stabilizer states and Clifford operations, while capable of generating highly entangled states, remain efficiently classically simulable according to the Gottesman--Knill theorem. Universal quantum computation therefore requires resources outside this framework, supplied for instance by magic states~\cite{howard2014contextuality, howard2017application}. In this sense, magic quantifies a form of quantum computational complexity that is not detected by entanglement alone. Indeed, highly entangled stabilizer states may have zero magic, while conversely, a product state can possess nonzero magic~\cite{gu2024magic, Iannotti_2025_interplay}. Thus, entanglement and non-stabilizerness characterize distinct, though intertwined, aspects of quantum many-body states. Consequently, a characterization based solely on standard entanglement measures is physically incomplete~\cite{cusumano2025nonstabilizerness} whenever the goal is to resolve the computational structure of quantum correlations. To fully capture this interplay, it is therefore useful to introduce diagnostics of the ES that are simultaneously sensitive to both entanglement and non-stabilizerness.

The observation that a stabilizer state has a flat ES~\cite{fattal2004entanglement, HAMMA200522} motivates the introduction of a general notion of \emph{antiflatness}---a measure of how much a spectrum deviates from uniformity. In this context, a \emph{flat spectrum}~\cite{watrous2018theory} corresponds to maximal uniformity: all non-zero eigenvalues are equal, e.g. pure or maximally mixed states. Spectra that are uneven, clustered, or hierarchical indicate that some subspaces dominate over others, leading to reduced effective dimensionality and nontrivial organization of quantum information. Conceptually, antiflatness measures provide second-order information about quantum correlations between subsystems when studying the reduced density matrix spectrum. While entropic quantities characterize the average spread or volume of uncertainty, antiflatness quantifies its fluctuations, the extent to which the spectral weights differ from one another. This makes antiflatness a sensitive probe of structure: it distinguishes states that share the same entropy but differ in how that entropy is distributed among eigenmodes. Specifically, it has been shown that measures of magic can be expressed in terms of the antiflatness of the ES of a pure quantum state~\cite{tirrito2024quantifying,Turkeshi_2023}, i.e. stabilizer states are a subset of flat states~\cite{gu2024magic}.
Moreover, when one restricts attention to the case of genuinely correlated (non-local) magic, namely the irreducible nonstabilizerness that cannot be removed by local basis changes, a tighter connection with the ES probed by antiflatness has been found~\cite{cao2024gravitational}. Such a relation has application in quantum many-body systems~\cite{jasser2025stabilizerSYK, viscardi2025interplayentanglementstructuresstabilizer, IannottiFreeferiomns, collura2026nonlocalnonstabilizernessfreefermion}, high energy physics~\cite{robin2026quantumcomplexitynewdirections, robin2025antiflatnessnonlocalmagictwoparticle, grieninger2026quantumcomplexitystringbreaking}, as well as CFTs, which, if they admit a holographic dual, imply connections with gravitational backreaction in the bulk~\cite{cao2024gravitational, grieninger2026nonlocalmagicholographicschwinger}. 

In this work, we study the properties of the ES through the lens of antiflatness. Our fundamental motivation is to use antiflatness as a bridge to explore the interplay between entanglement and magic~\cite{leone2025noncliffordcostrandomunitaries}. For a bipartition of a pure state, a non-trivial ES is a direct signature of entanglement.  At the same time, stabilizer states are characterized by a flat ES (yielding zero antiflatness)~\cite{fattal2004entanglement}. Deforming this spectrum strictly requires the presence of both entanglement and magic, even though their mere coexistence is not always sufficient~\cite{tirrito2024quantifying}. Antiflatness thus provides a natural framework to identify spectral features that remain invisible to standard entanglement measures but are sensitive to the computational complexity of the underlying quantum state. Beyond the specific connection to entanglement and magic, antiflatness offers a complementary language for describing quantum states through the geometry of their entanglement spectra and the fluctuations encoded therein.

The paper is organized as follows. 
In~\cref{Sec : Ordering}, we introduce a partial ordering of quantum states based on the spread of their Rényi entropies, shifting the focus from probability concentration to spectral inhomogeneity. 
Within the same section, we define Flatness-Preserving Operations (FPOs), which provide information-theoretic constraints on state convertibility~\cite{chitambar2019quantum, gour2024resources, brandao2015second, faist2018fundamental}. 
We then investigate the typicality of this convertibility framework by computing the corresponding accessible volume in illustrative low-dimensional cases, while emphasizing the difficulty of evaluating it in general.
In~\cref{Sec : Quantifiers}, we introduce and analyze several faithful quantifiers of antiflatness: the Capacity of Entanglement, the Linear Rényi Spread, and the Logarithmic Antiflatness. 
We compute their absolute maxima and show that the extremal spectra form a continuous Pareto frontier rather than a single universal maximizer. 
We also prove that these quantifiers are compatible with the antiflatness partial order introduced earlier.
In~\cref{Sec : Escort}, we reformulate these measures in terms of escort distributions. 
Using Bregman divergences between a state and its escort deformation, we show that the Capacity of Entanglement can be expressed as the second derivative of the Kullback--Leibler divergence along the escort trajectory in the limit $\alpha\to1$. 
This strengthens its connection with Quantum Fisher Information, identifying it as a susceptibility of the entanglement spectrum to thermal-like escort deformations.
In~\cref{Sec : Averages}, we analyze the typical behavior of entanglement-spectrum fluctuations under Haar, Bures--Hall, and $t$-doped Clifford random-state ensembles. 
We find that these spectral fluctuations vanish rapidly in the large-dimensional bipartite limit, refining the standard concentration-of-measure picture according to which the reduced state of a sufficiently small subsystem is overwhelmingly close to maximally mixed~\cite{Popescu_2006, Hayden_2007}.
In~\cref{sec:Rate_entanglement}, we derive an upper bound on the rate at which linear entanglement can be generated during dynamical evolution, expressed in terms of the Linear Rényi Spread. 
Conclusions and future directions are discussed in~\cref{Sec : Conclusions}.

\section{Ordering via Rényi spread}
\label{Sec : Ordering}
Let $\mathcal{D}(\mathcal{H})$ be the set of quantum states acting on a finite-dimensional Hilbert space $ \mathcal{H} \cong \mathbb{C}^d$, endowed with a tensor product structure such that $\mathcal{H} \simeq \mathcal{H}_A \otimes \mathcal{H}_B$ with  $d_A=\text{dim}(\mathcal{H}_A)$ and $d_B=\text{dim}(\mathcal{H}_B)$. Given a state $\rho \in \mathcal{D}(\mathcal{H})$, let $\Tr_B(\rho) \equiv \rho_A$ be its reduced density matrix. We are interested in characterizing the structure of $\rho_A$ in terms of the distribution of its eigenvalues.

The set of \emph{flat$_A$ states} is defined as the set of those states whose reduced density matrix is maximally uniform on its support, i.e., proportional to a projector
\begin{equation}
    \text{flat}_A:= \Big\{ \rho \in  \mathcal{D}(\mathcal{H}_A \otimes \mathcal{H}_{B}) \Big| \rho_A^2=\frac{\rho_A}{\text{rank}(\rho_A)} \Big\} \, .
\end{equation}
Equivalently, a state belongs to $\text{flat}_A$ if and only if can be written as $\rho_A = \frac{\Pi_A}{\text{rank}(\Pi_A)}$, where $\Pi_A$ is the projector onto the support of $\rho_A$. In this case, $\rho_A$ has a completely flat spectrum on its support, i.e., all its nonzero eigenvalues are equal. This condition implies that all Rényi entropies of $\rho_A$ coincide, providing an information-theoretic characterization of antiflatness. States that do not satisfy this condition exhibit non-uniform eigenvalue distributions in their reduced density matrix and are therefore referred to as \emph{antiflat$_A$ states}. These states encode nontrivial spectral fluctuations, which will be the main objective of this study.
Physical examples of flat$_A$ states are provided by stabilizer states~\cite{fattal2004entanglement} and, more generally, by quantum states whose reduced density matrix is proportional to a projector on their support. This situation appears in systems with topological order, where the ES may display exact or approximate degeneracies associated with long-range entanglement and edge-state correspondence~\cite{li2008entanglement, pollmann2010entanglement, HAMMA200522,flammia2009topological, oliviero2022stability}. In these cases, spectral flatness reflects the absence of additional structure within the support of $\rho_A$, while deviations from a flat$_A$ spectrum quantify the emergence of nontrivial spectral organization that can be used to diagnose topological order \cite{PhysRevB.88.125117, PhysRevLett.110.210602}.

The first result of this work is the introduction of a physically meaningful ordering of antiflat$_A$ states based on the spread of their Rényi entropies. Unlike standard approaches rooted in majorization theory, which organize states according to their purity, we instead focus on the fluctuations of the ES determined by antiflatness,~\cref{App : Failing}. To motivate this construction, we recall that for any density operator $\rho\in D(\mathcal{H})$, the Rényi entropy 
\begin{equation}
    S_{\alpha}(\rho)= \frac{1}{1-\alpha}\log(\Tr[\rho^\alpha])\,,
\end{equation}
is a monotonically decreasing function with respect to the parameter $\alpha$. Every flat$_A$ state has all Rényi entropies of the reduced density operator that evaluate to the same constant value $\log(\text{rank}(\Pi_A))$. Let the Rényi curve of a state be the continuum set of $S_{\alpha}(\rho)$ values given by $\alpha$, we have that the more antiflat a state is, the steeper its corresponding Rényi curve becomes. As showed in Fig.~\ref{fig: REcurves}, antiflatness is encoded in the steepness of the Rényi curves allowing to introduce an  ordering relation based on the Rényi Entropy Spread~\cite{HaydenWinter2003, vandam2002renyientropicboundsquantumcommunication}, defined as
\begin{equation}
    \Delta_{\alpha\beta}(\rho) \equiv S_{\alpha}(\rho) - S_{\beta}(\rho) ,
\end{equation}
with $\alpha < \beta $. This quantity is faithful with respect to the set of flat$_A$ states; for any reduced state $\rho_A$, one has
\begin{equation}
    \Delta_{\alpha\beta}(\rho_A)=0 \quad \forall\,\alpha<\beta
    \qquad \Longleftrightarrow \qquad
    \rho\in \mathrm{FLAT}_A .
\end{equation}
Indeed, all Rényi entropies coincide if and only if the nonzero eigenvalues of $\rho_A$ are all equal, i.e. if and only if $\rho_A$ is proportional to a projector onto its support.
\begin{definition}[Antiflatness Ordering]
\label{def:Antiflatness Ordering}
    Let $\rho, \sigma \in \mathcal{D}(\mathcal{H})$ be two density operators. We say that $\rho$ is \textit{antiflat majorized} by $\sigma$, denoted by $\rho \prec_{AF} \sigma$, if and only if the Rényi entropy spread of $\rho_A$ is bounded above by that of $\sigma_A$ for all valid parameters. Formally,
    \begin{equation}
        \rho \prec_{AF} \sigma \iff  \Delta_{\alpha\beta}(\rho_A) \le \Delta_{\alpha\beta}(\sigma_A) \quad \forall \,  \alpha < \beta \; \ .
    \end{equation}
\end{definition}

This ordering compares states according to the spread of their entanglement spectra rather than their concentration. In particular, it distinguishes states that are indistinguishable from the perspective of standard entropy measures but differ in higher-order spectral fluctuations, see ~\cref{App : Discrimination via Antiflatness}. Under this ordering, all $\mathrm{flat}_A$ states are minimal elements, as they exhibit zero spread for every pair $(\alpha,\beta)$.

The relation $\prec_{AF}$ establishes a \textit{partial} order rather than a \textit{total} order on the space of density matrices. While the Rényi spread captures the antiflatness hierarchy~\footnote{The ``antiflatness hierarchy'' refers to the partially ordered set (poset) induced on the space of density matrices $\mathcal{D}(\mathcal{H})$ by the relation $\prec_{AF}$. A state $\rho$ occupies a strictly lower hierarchical tier than $\sigma$ if and only if its resource content is globally bounded across all statistical scales, satisfying $\Delta_{\alpha\beta}(\rho) \le \Delta_{\alpha\beta}(\sigma)$ for all $\alpha < \beta$. Because this hierarchy constitutes a partial order rather than a total order, the resource structure inevitably branches.} for a broad class of state transitions, there exist pairs of quantum states that remain strictly incomparable under this framework. To understand why, let us analyze the Rényi entropy as a continuous function of the parameter $\alpha$. For two arbitrary states $\rho$ and $\sigma$, their respective Rényi curves may intersect or exhibit different concavities, see Fig.~\ref{fig: REcurves} and~\cref{App : Discrimination via Antiflatness} for more details. Because the ordering $\rho \prec_{AF} \sigma$ requires the inequality $\Delta_{\alpha\beta}(\rho_A) \le \Delta_{\alpha\beta}(\sigma_A)$ holds strictly for \textit{all} valid pairs $\alpha < \beta$, finding even a single pair where this condition reverses is sufficient to break the comparability. Specifically, there could exist specific parameters $(\alpha_1, \beta_1)$ and $(\alpha_2, \beta_2)$ such that $\Delta_{\alpha_1\beta_1}(\rho_A) < \Delta_{\alpha_1\beta_1}(\sigma_A) \quad \text{but} \quad \Delta_{\alpha_2\beta_2}(\rho_A) > \Delta_{\alpha_2\beta_2}(\sigma_A) \ .$ When this crossing occurs, the Rényi spread of $\rho_A$ is not globally bounded by that of $\sigma_A$, nor vice versa. Therefore, neither $\rho \prec_{AF} \sigma$ nor $\sigma \prec_{AF} \rho$ holds and they are deemed incomparable. It is crucial to distinguish this incomparability, which arises from crossing spreads, from the necessary crossing of absolute entropy curves that occurs during iso-purity transformations, as detailed in~\cref{App : Maj freezed}. The antiflatness relation ($\prec_{AF}$) imposes a family of simultaneous constraints across all Rényi intervals $(\alpha,\beta)$ with $\alpha<\beta$, rather than a single scalar ordering criterion. Therefore, the state space $\mathcal{D}(\mathcal{H})$ does not form a straight line in ordering states according to their antiflatness, but rather a branching structure. This induces a hierarchy in which flat$_A$ states—antiflat-majorized by all others—form the minimal elements, while the maximal elements lie on a Pareto frontier (see Sec.~\ref{Sec : Pareto}).

\begin{figure}[t]
    \centering
    \includegraphics[width=0.9\linewidth]{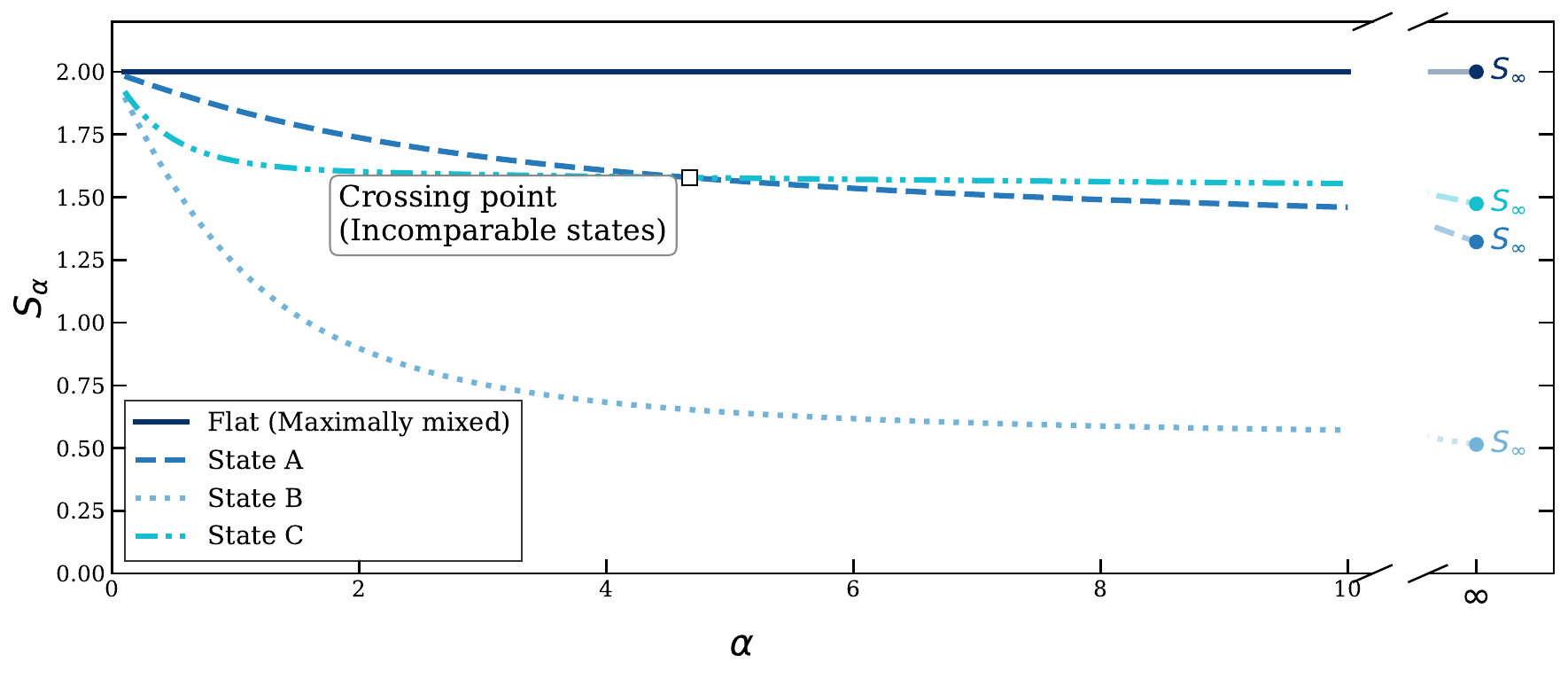}
    \caption{Rényi entropy curves $S_\alpha(\rho_A)$ for representative states in $d_A=4$ with spectra:
    Flat $(0.25,0.25,0.25,0.25)$; State A $(0.4,0.3,0.2,0.1)$; State B $(0.7,0.2,0.08,0.02)$; and State C $(0.36,0.34,0.29,0.01)$. 
    Flat$_A$ states yield constant Rényi curves, corresponding to vanishing Rényi spread. 
    The crossing between States A and C provides a visual indication that pointwise comparison of Rényi entropies is insufficient; their incomparability under $\prec_{AF}$ is confirmed by the spread inequalities, since $\Delta_{1,2}(A)>\Delta_{1,2}(C)$ whereas $\Delta_{0.5,1}(C)>\Delta_{0.5,1}(A)$. This illustrates the partial nature of the antiflat ordering and shows that spectral steepness over a restricted range of $\alpha$ is not, by itself, a criterion for antiflat comparability.
    }
    \label{fig: REcurves}
\end{figure}

\subsection{State Convertibility and Flatness-Preserving Operations}
\label{sec:convertibility}
In several quantum resource theories, state convertibility is characterized by a partial order on spectra. 
An example is bipartite entanglement; where Nielsen's theorem~\cite{nielsen1999conditions} states that a pure bipartite state $|\psi\rangle$ can be deterministically transformed into another pure quantum state $|\phi\rangle$ via Local Operations and Classical Communication (LOCC) if and only if their reduced spectra satisfy $\rho_\psi \prec \rho_\phi$. One can ask whether antiflat majorization $\rho \prec_{AF} \sigma$ similarly dictates state transformations under a specific class of free operations.

We define \emph{Flatness-Preserving Operations} (FPOs) as the maximal set of completely positive trace-preserving (CPTP) maps that do not increase the Rényi entropy spread of the reduced state. More precisely, a CPTP map $\Phi$ belongs to the class of FPOs if
\begin{equation}
    \Delta_{\alpha\beta}\big((\Phi(\rho))_A\big)
    \leq
    \Delta_{\alpha\beta}(\rho_A),
    \qquad
    \forall \rho \in \mathcal{D(H)},\quad \forall \alpha<\beta .
\end{equation}
Thus, FPOs are the operations that cannot transform a state with few ES antiflatness into one with more antiflatness. With this definition, antiflat majorization provides an immediate \textit{necessary condition} for pure-state convertibility under FPO.
Let $|\psi\rangle,|\phi\rangle \in \mathcal{H}_A\otimes\mathcal{H}_B$ be bipartite pure states, and let
$\rho^\psi,\rho^\phi$ denote the associated density matrices. Then
\begin{equation}
|\psi\rangle \xrightarrow{\mathrm{FPO}} |\phi\rangle
\quad \Longrightarrow \quad
\rho^\phi \prec_{AF} \rho^\psi ,
\end{equation}
see~\cref{App: prop necessary condition}. 
Consequently, the relation $\prec_{AF}$ acts as a selection rule for quantum dynamics. If $\rho^\phi \not\prec_{AF} \rho^\psi$, either because the spread strictly increases or because their Rényi spreads intersect, making them incomparable, the transformation $|\psi\rangle \to |\phi\rangle$ is forbidden under free operations. 
Establishing \textit{sufficiency}, proving that $\rho^\phi \prec_{AF} \rho^\psi$ implies the existence of an FPO channel from $|\psi\rangle$ to $|\phi\rangle$, remains a highly non-trivial problem, as pointed out in ~\cref{App: Hardeness of finding FPOs}. As an example, local unitaries provide FPOs as they simply do not change the Schmidt coefficients and hence the ES; since they cannot connect flat$_A$ states with different Schmidt rank. This rank obstruction can be bypassed only by enlarging the system with an additional flat$_A$ auxiliary resource. For instance, tensoring with a maximally entangled EPR pair multiplies the Schmidt rank while preserving spectral flatness, allowing one to match flat spectra of multiplicatively related ranks in an enlarged Hilbert space.

A distinction must be made between FPOs and standard majorization-preserving or unital dynamics. Unital CPTP maps are characterized by standard majorization of spectra: if $\sigma=\Phi(\rho)$ for a unital channel, then $\lambda(\sigma) \prec \lambda(\rho)$. This implies monotonicity of Schur-concave entropies, but it does not imply monotonicity of the Rényi spreads $\Delta_{\alpha\beta}$. Therefore, unital CPTP maps are not, in general, a subclass of FPOs. This already indicates that antiflat convertibility defines a notion of spectral dynamics distinct from standard majorization.

To isolate this difference from trivial changes in the overall mixedness of the state, it is useful to fix the purity of the reduced density matrix. We define the iso-purity manifold
\begin{equation}
    \mathcal{M}_P
    =
    \left\{
    \rho\in\mathcal{D}(\mathcal{H})
    \ \middle|\
    \Tr(\rho_A^2)=P
    \right\},
\end{equation}
with $1/d_A\leq P\leq 1$. For $P>1/d_A$, the full-rank stratum of $\mathcal{M}_P$ is a smooth manifold, while the complete physical set has the structure of a compact stratified manifold; see~\cref{App : Iso-Purity Manifold}. Any two states $\rho,\sigma\in\mathcal{M}_P$ satisfy
\begin{equation}
    S_2(\rho_A)=S_2(\sigma_A).
\end{equation}
Fixing the purity does not fix the rank. The purity and the max-entropy
\begin{equation}
    S_0(\rho_A)=\log\rank(\rho_A)
\end{equation}
impose different constraints, and states with different support sizes may have the same purity. This distinction is crucial: standard majorization becomes frozen on $\mathcal{M}_P$, while antiflat ordering can still impose nontrivial spectral constraints.
\begin{proposition}[Freezing of Standard Majorization]
\label{thm:freezing}
If a state transformation on $\mathcal{M}_P$ obeys standard majorization ($\rho \prec \sigma$), the transformation is trivially restricted to Local Isometries.   
\end{proposition}
\begin{proof}
Let $\sigma, \rho \in \mathcal{M}_P$, implying $S_2(\rho_A) = S_2(\sigma_A)$. If $\sigma$ transforms to $\rho$ via LOCC or Unital maps, standard majorization requires $\rho \prec \sigma$. By the strict Schur-convexity of purity, $\rho \prec \sigma$ and $S_2(\rho_A) = S_2(\sigma_A)$ requires that the ordered spectra are identical; $\lambda(\rho_A)^\downarrow=\lambda(\sigma_A)^\downarrow$. Thus, no non-trivial standard majorization flow can occur.
\end{proof}
Proposition~\ref{thm:freezing} shows that LOCC-type or unital majorization-based transformations cannot modify the ES nontrivially once the purity is fixed. Antiflat ordering behaves differently, because it does not compare spectra through entropy monotonicity alone, but through the spread of the whole Rényi curve.
To make this idea even clearer, let $\sigma \xrightarrow{\mathrm{FPO}} \rho$
be a transformation inside $\mathcal{M}_P$. A necessary condition is $\rho_A \prec_{AF} \sigma_A$, namely
\begin{equation}
    \Delta_{\alpha\beta}(\rho_A)
    \leq
    \Delta_{\alpha\beta}(\sigma_A),
    \qquad
    \forall \alpha<\beta .
\end{equation}
Because $S_2(\rho_A)=S_2(\sigma_A)$ on $\mathcal{M}_P$, the point $\alpha=2$ acts as a natural anchor for the Rényi curves. For the pair $(\alpha,\beta)=(1,2)$, one obtains
\begin{align}
    S_1(\rho_A)-S_2(\rho_A)
    \leq
    S_1(\sigma_A)-S_2(\sigma_A)
    \Longrightarrow
    S_1(\rho_A) \leq
    S_1(\sigma_A).
\end{align}
Similarly, for $(\alpha,\beta)=(2,\infty)$,
\begin{align}
    S_2(\rho_A)-S_\infty(\rho_A)
    \leq
    S_2(\sigma_A)-S_\infty(\sigma_A)
    \Longrightarrow
    S_\infty(\rho_A)\geq
    S_\infty(\sigma_A).
\end{align}
Therefore, at fixed purity, decreasing antiflatness does not correspond to an entropy-increasing process. 
The von Neumann entropy of the final state is forced to decrease, while the min-entropy increases. 
Since $S_\infty(\omega_A)=-\log \lambda_{\max}(\omega_A)$, the latter condition is equivalent to
\begin{equation}
    \lambda_{\max}(\rho_A)\leq \lambda_{\max}(\sigma_A).
\end{equation}
Thus, an iso-purity FPO, if it exists, must reshape the spectrum in a correlated way: it cannot just make the state more mixed in the standard majorization sense.
This endpoint analysis is the manifestation of a more general structural simplification. 
Although antiflat majorization is defined through a two-parameter family of inequalities, it can be reformulated as a one-parameter monotonicity condition.

\begin{proposition}[One-parameter reformulation of antiflat majorization]
\label{prop:one_parameter_AF}
Let $\rho_A$ and $\sigma_A$ be two reduced density matrices, and define
\begin{equation}
    G_{\sigma|\rho}(\alpha)
    :=
    S_\alpha(\sigma_A)-S_\alpha(\rho_A),
    \qquad
    \alpha\in(0,\infty).
\end{equation}
Then we have 
\begin{equation}
    \rho_A\prec_{AF}\sigma_A \iff G_{\sigma|\rho}(\alpha)
    \geq
    G_{\sigma|\rho}(\beta),
    \qquad
    \forall \alpha<\beta .
\end{equation}
Equivalently, $\rho_A\prec_{AF}\sigma_A$ if and only if the function $\alpha\mapsto G_{\sigma|\rho}(\alpha)$ is non-increasing on $(0,\infty)$.
\end{proposition}

\begin{proof}
By definition, $\rho_A\prec_{AF}\sigma_A$ means that
$\Delta_{\alpha\beta}(\rho_A) \leq \Delta_{\alpha\beta}(\sigma_A), \forall \alpha<\beta $.
By definition the condition on the Rényi spreads becomes
\begin{equation}
    S_\alpha(\rho_A)-S_\beta(\rho_A)
    \leq
    S_\alpha(\sigma_A)-S_\beta(\sigma_A).
\end{equation}
Rearranging terms gives
\begin{equation}
    S_\alpha(\sigma_A)-S_\alpha(\rho_A)
    \geq
    S_\beta(\sigma_A)-S_\beta(\rho_A).
\end{equation}
Which is
\begin{equation}
    G_{\sigma|\rho}(\alpha)
    \geq
    G_{\sigma|\rho}(\beta),
    \qquad
    \forall \alpha<\beta .
\end{equation}
Hence $G_{\sigma|\rho}$ is non-increasing. 
The converse follows by reversing the same steps.
\end{proof}

Proposition~\ref{prop:one_parameter_AF} shows that antiflat comparability is not determined by a pointwise comparison of the Rényi entropies themselves. 
Rather, it is determined by the monotonicity of the relative Rényi profile $G_{\sigma|\rho}(\alpha)$. 
Equivalently,
\begin{equation}
    \Delta_{\alpha\beta}(\sigma_A)
    -
    \Delta_{\alpha\beta}(\rho_A)
    =
    G_{\sigma|\rho}(\alpha)-G_{\sigma|\rho}(\beta).
\end{equation}
Therefore, the spread inequalities are satisfied for all $\alpha<\beta$ precisely when $G_{\sigma|\rho}$ is non-increasing.

On the iso-purity manifold, this reformulation becomes especially transparent. 
Indeed, for $\rho,\sigma\in\mathcal{M}_P$ one has
\begin{equation}
    G_{\sigma|\rho}(2)
    =
    S_2(\sigma_A)-S_2(\rho_A)
    =
    0.
\end{equation}
If $\rho_A\prec_{AF}\sigma_A$, then $G_{\sigma|\rho}$ is non-increasing. 
Since it vanishes at $\alpha=2$, one obtains $G_{\sigma|\rho}(\alpha)\geq 0$ and $S_\alpha(\sigma_A)\geq S_\alpha(\rho_A)$ when $ \alpha < 2$, otherwise $G_{\sigma|\rho}(\alpha)\leq 0$ and $S_\alpha(\sigma_A)\leq S_\alpha(\rho_A)$.
Thus, on $\mathcal{M}_P$, antiflat ordering enforces an ordered crossing of the two Rényi curves at $\alpha=2$. 
The crossing itself is not a signal of incomparability; incomparability occurs only when the corresponding function $G_{\sigma|\rho}$ fails to be monotone.

A further consequence appears when the support size is fixed. 
If $\rho,\sigma\in\mathcal{M}_P$ have the same rank, then
\begin{equation}
    S_0(\rho_A)=S_0(\sigma_A),
    \qquad
    S_2(\rho_A)=S_2(\sigma_A).
\end{equation}
Hence $G_{\sigma|\rho}(0)=G_{\sigma|\rho}(2)=0$
where the value at $\alpha=0$ is understood as the continuous limit. 
If, in addition, $\rho_A\prec_{AF}\sigma_A$, then $G_{\sigma|\rho}$ is non-increasing. 
A non-increasing function that takes the same value at two points must be constant between them. 
Therefore,
\begin{equation}
    G_{\sigma|\rho}(\alpha)=0,
    \quad 
    S_\alpha(\rho_A)=S_\alpha(\sigma_A),
    \qquad
    0<\alpha<2.
\end{equation}
Since the Rényi entropies on an open interval determine the spectrum of a finite-dimensional density matrix, it follows that
\begin{equation}
    \lambda(\rho_A)^\downarrow
    =
    \lambda(\sigma_A)^\downarrow .
\end{equation}
Consequently, any nontrivial antiflat-ordered pair inside $\mathcal{M}_P$ must necessarily involve a change of rank:
\begin{equation}
    \rho_A \prec_{AF}\sigma_A \quad\land \quad \lambda(\rho_A)^\downarrow\neq \lambda(\sigma_A)^\downarrow
\Longrightarrow \quad
    \rank(\rho_A)\neq\rank(\sigma_A).
\end{equation}

This gives a clear interpretation of FPO dynamics at fixed purity. 
Standard majorization is frozen on $\mathcal{M}_P$, because equality of purity together with majorization forces equality of the ordered spectra. 
By contrast, antiflat ordering is not governed by Schur-concavity alone. 
If a nontrivial FPO transformation exists on $\mathcal{M}_P$, it must act by changing the support structure of the reduced spectrum. 
In this sense, iso-purity FPOs behave as spectral compressors: they reduce the Rényi spread while preserving purity, forcing a nontrivial reshaping of the spectrum rather than a standard entropy-increasing flow. The proof of the rank obstruction and the computational advantage of the one-parameter formulation are discussed in~\cref{app:one_parameter_af_proofs}.

\subsection{Statistical Geometry and the Accessible Antiflatness Volume}
\label{sec: accessible_volume}

Having defined such state convertibility, one can determine the typicality of such a framework. Given a fixed target pure state $\ket{\phi}$, one can ask: what is the probability that a pure state $\ket{\psi}$ drawn uniformly at random from the bipartite Hilbert space $\mathcal{H}_A \otimes \mathcal{H}_B$ can deterministically be converted into $\ket{\phi}$? For standard transformations via LOCC, Nielsen's theorem dictates that this probability is exactly the probability that the Haar-random state $\sigma_A \equiv \Tr_B(\ket{\phi} \bra{\phi})$ majorizes $\rho_A \equiv \Tr_B(\ket{\psi} \bra{\psi})$. In the literature, this is known as the \textit{Accessible Volume}~\cite{sauerwein2015source}. If we sample global pure states according to the Haar measure of the unitary group, the eigenvalues of the reduced state follow the Lloyd-Pagels (Wishart-Laguerre) probability density function~\cite{lloyd_pagels_1988}, $P_{\text{Haar}}(\vec{x})$. The probability is obtained by integrating this measure over the majorization set
\begin{equation}
    P(\sigma \succ \rho) = \int_{M_a(\rho)} P_{\text{Haar}}(\vec{x}) \, d\vec{x} \ .
\end{equation}
By the Birkhoff-von Neumann theorem, the accessible set $M_a(\rho) = \{ \vec{x} \in \mathbf{\Delta_{d_A-1}} \mid \vec{x} \succ \text{spec}(\rho_A) \}$ forms a strictly convex, flat-faced polytope within the probability simplex $\mathbf{\Delta_{d_A-1}}$. Despite this well-defined flat geometry, computing this integral analytically for arbitrary dimensions $d_A$ remains a highly non-trivial open problem in random matrix theory, typically requiring advanced Monte Carlo methods.

We now extend this statistical inquiry to our resource theory. We seek the probability that a Haar-random state $\sigma$ \textit{antiflat-majorizes} a fixed state $\rho$ ($\sigma \succ_{AF} \rho$), thereby satisfying the necessary antiflat condition for an FPO transformation $\sigma \xrightarrow{\text{FPO}} \rho$
\begin{equation}
\label{eq: iso_purity_integral}
    P(\sigma \succ_{AF} \rho) = \int_{\mathcal{A}_{AF}(\rho)} P_{\text{Haar}}\big(\vec{x}\big) \, d\vec{x} \ ,
\end{equation}
where we define the accessible antiflatness set
\begin{equation}
    \mathcal{A}_{AF}(\rho) = \left\{ \vec{x} \in \mathbf{\Delta_{d_A-1}} \;\middle|\; \Delta_{\alpha \beta}(\vec{x}) \ge \Delta_{\alpha\beta}(\rho_A),\, \forall \, \alpha < \beta \right\} \ .
\end{equation}
Unlike standard majorization, the region defined by antiflatness majorization is \textit{not} a polytope. Because the constraints are defined by Rényi entropies involving logarithms evaluated over a continuous and infinite parameter space ($\alpha < \beta$), the set $\mathcal{A}_{AF}(\rho)$ constitutes a highly non-linear, continuous geometric volume. More comments on its non-polytope nature are detailed in~\cref{App: Probability of Antiflatness}. Consequently, integrating the Haar measure directly over this unconstrained, non-convex hyper-volume is analytically intractable.

We now make the accessible antiflatness probability explicit in low dimensions. 
Let the target reduced spectrum be $\vec r=\lambda(\rho_A)$, and let $\vec\lambda=\lambda(\sigma_A)$ denote the spectrum of the reduced density matrix of a Haar-random bipartite pure state, see \ref{app:Haar average}. 
For generic dimensions $d_A\leq d_B$, the induced Lloyd--Pagels measure~\cite{lloyd_pagels_1988} on the eigenvalues is
\begin{equation}
d\mu(\vec\lambda)
\propto
\delta\!\left(1-\sum_{i=1}^{d_A}\lambda_i\right)
\prod_{1\leq i<j\leq d_A}(\lambda_j-\lambda_i)^2 
\prod_{i=1}^{d_A}
\lambda_i^{d_B-d_A}\,
d\lambda_i .
\label{eq:Haar measure}
\end{equation} 
The generic accessible antiflatness probability is therefore
\begin{equation}
\label{eq:generic_AF_probability}
    \mathbb{P}(\sigma\succ_{AF}\rho)
    =
    \int_{\mathcal{A}_{AF}(\vec r)}
    d\mu(\vec\lambda) \; .
\end{equation}
For $d_A=2$, an ordered spectrum can be written as $\vec\lambda=(x,1-x), x\in[1/2,1]$, and the target as $\vec r=(r,1-r), r\in[1/2,1]$.
For $d_A=d_B=2$, the induced density reduces to $p_2(x)=6(2x-1)^2$.
If the target is flat on its support, namely $r=1/2$ or $r=1$, then all its Rényi spreads vanish $\Delta_{\alpha\beta}(\vec r)=0, \; \; \forall \alpha<\beta $.
Since Rényi entropies are non-increasing in $\alpha$, one has
$\Delta_{\alpha\beta}(\vec\lambda)\geq0$ for every spectrum. 
Therefore every Haar-random state antiflat-majorizes a flat target
\begin{equation}
    \mathbb{P}(\sigma_A\succ_{AF}\rho_A)
    =
    \int_{1/2}^{1}6(2x-1)^2\,dx
    =
    1.
\end{equation}
By contrast, if the target is non-flat, $r\in(1/2,1)$, the full antiflat order is rigid in the binary case: the accessible set collapses to the single point $x=r$. 
Since this set has zero measure with respect to the continuous Lloyd--Pagels density,
\begin{equation}
    \mathbb{P}(\sigma_A\succ_{AF}\rho_A)
    =
    \int_{\{r\}}6(2x-1)^2\,dx
    =
    0.
\end{equation}
The same conclusion holds for $d_A=2$ and arbitrary $d_B=K>2$ in~\cref{App: Probability of Antiflatness}. 
The result above shows that, for any non-flat target state, the Haar probability of sampling a state that antiflat-majorizes it is zero. This reflects the strong rigidity of the full antiflatness order: generic pairs of states are typically incomparable under this criterion. In particular, Haar-random reduced states concentrate around the typical spectral-fluctuation regime described in~\cref{app:Haar average}, rather than forming a nested hierarchy under antiflat majorization.

\section{Quantifiers}
\label{Sec : Quantifiers}
Typically, to quantify how resourceful a state is, it is necessary to define a notion of distance from the set of free states. However, the precise way in which this notion is formalized depends on the chosen framework. In the following we will not rely explicitly on a distance-based construction, but rather introduce suitable quantifiers capturing anti-flatness. By establishing an axiomatic approach, we outline the essential properties that any valid quantifier must satisfy. We are aware that the axiomatic and operational approaches to quantum resource theories do not necessarily coincide, as recently outlined in the context of magic~\cite{heimendahl2021axiomatic}. However, because the axiomatic framework rigorously establishes the ultimate information-theoretic bounds permitted by the geometry of the free states, we adopt it as our starting point. Based on these criteria, we identify a range of candidate measures appropriate for analyzing anti-flatness.

A well-known measure of the antiflatness of a spectrum is the so-called Capacity of Entanglement, previously introduced in~\cite{yao2010entanglement, schliemann2011entanglement}, and further developed in~\cite{de2019aspects}.
\begin{definition}
    The capacity of entanglement of $\rho \in \mathcal{D}(\mathcal{H}_A \otimes \mathcal{H}_B) $ is defined as
    \begin{equation}
    \begin{split}
        \mathcal{V}_A(\rho)&:=tr(\rho_A \log^2 \rho_A)-tr^2(\rho_A \log \rho_A)\equiv \text{Var}_{\rho_A} (-\log \rho_A) \, .
    \end{split}
        \label{eq:capacity}
    \end{equation}
\end{definition}
By treating the reduced density matrix as a canonical thermal state with modular Hamiltonian $H_A = -\log \rho_A$, the capacity of entanglement acts as the statistical heat capacity of the ES. Following standard conventions in the literature~\cite{de2019aspects, cao2024gravitational}, it is formally expressed through the analytic continuation of the partition function $Z_\alpha = \text{Tr}(\rho_A^\alpha)$. Specifically, it corresponds to the second derivative of the cumulant generating function evaluated at $\alpha=1$
\begin{equation}
    \mathcal{V}_A(\rho) = \left. \frac{\partial^2}{\partial \alpha^2} \log \text{Tr}(\rho_A^\alpha) \right|_{\alpha=1} \ .
\end{equation}
Because the Rényi entropy is defined as $S_\alpha(\rho_A) = \frac{1}{1-\alpha} \log \text{Tr}(\rho_A^\alpha)$, substituting the identity $\log \text{Tr}(\rho_A^\alpha) = (1-\alpha)S_\alpha(\rho_A)$ into the expression, we have
\begin{equation}
    \mathcal{V}_A(\rho) = \lim_{\alpha \to 1} \frac{\partial^2}{\partial \alpha^2} \Big[ (1-\alpha) S_\alpha(\rho_A) \Big] = -2 \left. \frac{\partial S_\alpha(\rho_A)}{\partial \alpha} \right|_{\alpha=1} \ .
\end{equation}
This differential formulation proves that while the standard Von Neumann entropy (the $\alpha \to 1$ limit) quantifies the mean of the modular energy, the capacity of entanglement extracts the first-order response (the derivative) of the Rényi entropy, thereby probing the spectrum's variance and the structural fluctuations of the quantum correlations.
Such a quantity has a lot of applications spanning from condensed matter to quantum gravity, with a non-exhaustive list given by~\cite{de2019aspects, schliemann2011entanglement, yao2010entanglement, wei2023average, okuyama2021capacity, cao2024gravitational,  Huang_2023, Shrimali_2022, Bhattacharjee_2021, verlinde2020spacetime}.
It has a variety of mathematical properties that can be found summarised in~\cite{boes2022variance} for $\rho \in  \mathcal{D}(\mathcal{H})$, i.e. (Faithfulness) $\mathcal{V}_A(\rho)=0 $ if and only if $\rho$ is proportional to a projector; (Positivity) $\mathcal{V}_A(\rho)\geq 0$; (Additivity) $\mathcal{V}_A(\rho_1 \otimes \rho_2)=\mathcal{V}_A(\rho_1)+\mathcal{V}_A(\rho_2)$; and (Uniform continuity) $|\mathcal{V}_A(\rho) - \mathcal{V}_A(\rho')|^2 \leq K \log^2(d_A) \epsilon $, with $\epsilon:= \frac{1}{2} \norm{\rho_A - \rho'_A}_1$ and $K$ a constant. Notice that, the Capacity of Entanglement was first introduced to distinguish the states with and without topologically protected gapless ES~\cite{yao2010entanglement} and it is defined in the same way as one defines heat capacity for thermal systems. It turns out that this quantity has an holographic interpretation. Via the Ryu–Takayanagi prescription~\cite{Ryu_2006}, entanglement entropy is given by the area of a minimal surface in Anti-de Sitter. Including quantum fluctuations of this surface leads to the capacity of entanglement~\cite{Nakaguchi_2016}, which captures fluctuations of the entropy itself and thus probes quantum gravitational effects. More generally, a backreacting cosmic brane construction relates this to modular entropy, closely connected to Rényi entropies, providing an alternative definition of the capacity of entanglement~\cite{Dong_2016}.

Recently in~\cite{tirrito2024quantifying}, a new measure which quantifies how much a state is far from being flat has been introduced. For analogy to our ordering, we will call it \textit{Linear Rényi spread} (LRS).
\begin{definition}
    The Linear Rényi spread of $\rho$ is defined as
    \begin{equation}
    \begin{split}
        \mathcal{F}_A(\rho)&:=tr(\rho_A^3)-tr^2(\rho_A^2) \equiv \text{Var}_{\rho_A} (\rho_A)= 1 - 2 T_3(\rho_A)- (1-T_2(\rho_A))^2\, \
    \end{split}
    \end{equation}
    with $T_\alpha (\rho)= \frac{tr(\rho^\alpha)-1}{1-\alpha}$ the Tsallis $\alpha$ entropy.
\end{definition}
This quantity, when evaluated on the reduced density operator has an interesting connection with the non-stabilizerness of a state~\cite{tirrito2024quantifying}. 
From the resource theory point of view, also this quantity is a good measure, since it features the following properties for $\rho \in  \mathcal{D}(\mathcal{H})$, (Faithfulness) $\mathcal{F}_A(\rho)=0 $ if and only if $\rho$ is proportional to a projector ; (Positivity) $\mathcal{F}_A(\rho)\geq 0$; (Lipschitz continuity) $|\mathcal{F}_A(\rho) - \mathcal{F}_A(\rho')| \leq 7 \epsilon $, with $\epsilon:= \frac{1}{2} \norm{\rho_A - \rho'_A}_1$ \cite{cao2024gravitational}. 
Additionally, one can show that it is a subadditive measure $\mathcal{F}_A(\rho) \mathcal{F}_A(\psi)\leq \mathcal{F}_A(\rho \otimes \psi) \leq \mathcal{F}_A(\rho) + \mathcal{F}_A(\psi)$; as shown in ~\cref{App: LinearRenyiSpread_properties}. 

As pointed out in~\cite{HaydenWinter2003}, any other difference between Rényi entropies is a measure of antiflatness of the ES. For the sake of concreteness, among all Rényi spreads, we choose a peculiar one
\begin{definition}
    The logarithmic antiflatness of $\rho$ is defined as
    \begin{equation}
        \log(\Lambda_\rho):=\log(\frac{tr(\rho_A^3)}{tr^2(\rho_A^2)})= 2(S_2(\rho_A)-S_3(\rho_A))\, .
    \end{equation}
\end{definition}
Also the logarithmic anti-flatness $\log(\Lambda_\rho)$ is a good quantifier, as shown in ~\cref{App: LogLambda_properties}, namely it satisfies (Faithfulness)  $\log(\Lambda_\rho)=0 $ if and only if $\rho_A$ is proportional to a projector ; (Positivity) $\log(\Lambda_\rho)\geq 0$; (Additivity) $\log(\Lambda_{\rho_1 \otimes \rho_2})=\log(\Lambda_{\rho_1})+\log(\Lambda_{\rho_2})$; (Uniform continuity) $|\log(\Lambda_\rho) - \log(\Lambda_{\rho'})| \leq 2 d (2 + 3 d) \epsilon $, with $\epsilon:= \frac{1}{2} \norm{\rho_A - \rho'_A}_1$.
Similar properties extend to all $\Delta_{\alpha \beta}(\rho_A)$ following from those of Rènyi entropies.
The above properties ensure that all our measures, $\mathcal{V}_A(\rho)$, $\mathcal{F}_A(\rho)$, $\log(\Lambda_\rho)$, and $\Delta_{\alpha \beta}(\rho_A)$ faithfully quantify deviations from flatness. 

The next step is therefore to characterize which spectra maximize these deviations. Since all quantities increase with the spread of the eigenvalues, one expects that the largest possible value should be attained by states whose spectrum is as far as possible from being uniform. 
In the finite-dimensional setting, it reduces to a discrete jump between the largest eigenvalue and the rest of the spectrum, as shown in Fig.~\ref{fig:spectra}, as it was first showed for $\mathcal{V}_A(\rho)$ in~\cite{reeb2015tight}.
\begin{theorem}[Universal Maximal Antiflatness]
\label{thm:Maximal_Antiflatness}
Let $\rho$ be a state on a $d$-dimensional Hilbert space with $d \geq 2$. The maximum deviations from flatness, quantified by the Capacity of Entanglement $\mathcal{V}_A(\rho)$, the Linear Rényi spread $\mathcal{F}_A(\rho)$ and the Logarithmic antiflatness $\log(\Lambda_\rho)$, are universally achieved by states possessing a jump spectrum of the form:
\begin{equation}
    \text{spec}( \rho_A) = \Big( 1- r_{\max}, \frac{r_{\max}}{d_A-1},\ldots,  \frac{r_{\max}}{d_A-1} \Big) \ ,
    \label{Eq : max spec}
\end{equation}
where the exact optimal weight $r_{\max} \in [0, 1/2]$ depends on the chosen quantifier.
Specifically, the measures are globally bounded by $0 \leq \mathcal{F}_A(\rho) \leq N_{\mathcal{F}_A}(d_A) < \frac{27}{256}$, $0 \leq \log(\Lambda_\rho) \leq N_{\Lambda}(d_A)< \log 2$ and $0 \leq \mathcal{V}_A(\rho) \leq N_{\mathcal{V}_A}(d_A)< 1/4 \log_2  (d_A- 1)^2 + 1/ \ln^2(2)$. In the limit of $d_A$, we have $r_{\max}^{\mathcal{F}_A} \approx \frac{1}{4}$ and $r_{\max}^{\log(\Lambda_\rho)} \approx r_{\max}^{\mathcal{V}_A} \approx \frac{1}{2}$.

\end{theorem}
The proofs for these bounds are deferred to~\cref{App : Maximal_Antiflatness}. They rely on the method of Lagrange multipliers (analogous to the approach in \cite{reeb2015tight}), demonstrating that the required spectral shape is universal, while the specific optimal weights account for the differences in the statistical scales probed by individual measures.

\subsection{Comments on the maxima}
\label{Sec : Pareto}

Since there are infinitely many ways to distort a spectrum away from the uniform distribution, the possibility of finding a distribution as the unique absolute maximum across all measures of antiflatness vanishes. These infinitely different ways of distorting a spectrum are what the parameter $r_{\max}$ in Eq.~\ref{Eq : max spec} controls.  Let $\rho_A(r)$ define the family of states possessing the jump spectrum $(1-r, \frac{r}{d-1}, \dots, \frac{r}{d-1})$, which characterize the extramal states, we can define
\begin{definition}[Pareto Optimal Antiflat$_A$ states]
A quantum state $\rho \in \mathcal{D}(\mathcal{H})$ is called Pareto optimal (or non-dominated) with respect to antiflatness, if there exists no other state $\sigma \in \mathcal{D}(\mathcal{H})$ that strictly antiflat majorizes it. Formally, $\rho$ is Pareto optimal if $ \nexists \; \; \sigma \in \mathcal{D}(\mathcal{H}) \quad \text{such that} \quad \rho \precneqq_{AF} \sigma $.
\end{definition}
The \textit{Pareto Frontier}, denoted as $\mathcal{P}_{AF}$, is the exact mathematical set of all such Pareto optimal states
\begin{equation}
    \mathcal{P}_{AF} = \Big\{ \rho \in \mathcal{D}(\mathcal{H}) \;\Big|\; \nexists \; \; \sigma \text{ s.t. } \rho \precneqq_{AF} \sigma \Big\} \ .
\end{equation}
This frontier is not a discrete set of isolated points, but a \textit{continuous} one-dimensional manifold.
For any two states $\rho_A(r_1), \rho_A(r_2) \in \mathcal{P}_{AF}$, moving from $r_1$ to $r_2$ might increase the spectral spread at the macroscopic scale (e.g., maximizing $\mathcal{F}_A(\rho)$), but it will need a decrease in the spread at a localized microscopic scale (e.g., the interval $\Delta_{2,3}$ captured by $\log(\Lambda_\rho)$).
 
\begin{proposition}[Global Bounding by the Pareto Frontier]
\label{prop:pareto_bounding}
For any arbitrary quantum state $\sigma \in \mathcal{D}(\mathcal{H})$ that is strictly sub-optimal (i.e., $\sigma \notin \mathcal{P}_{AF}$), there exists at least one extremal state $\rho \in \mathcal{P}_{AF}$ that strictly antiflat-majorizes it. Formally:
\begin{equation}
    \forall \, \sigma \notin \mathcal{P}_{AF}, \quad \exists \; \rho \in \mathcal{P}_{AF} \quad \text{such that} \quad \sigma \precneqq_{AF} \rho \ .
\end{equation}
\end{proposition}

\begin{proof}
The proof relies on the compactness of the state space and the transitivity of the partial order. Let $\sigma$ be an arbitrary state such that $\sigma \notin \mathcal{P}_{AF}$. By the definition of Pareto optimality, because $\sigma$ is not non-dominated, there must exist at least one state $\sigma_1$ that strictly antiflat-majorizes it: $\sigma \precneqq_{AF} \sigma_1$. 
If $\sigma_1 \in \mathcal{P}_{AF}$, the proposition is immediately satisfied by setting $\rho = \sigma_1$. If $\sigma_1 \notin \mathcal{P}_{AF}$, then by the same logic, there must exist another state $\sigma_2$ such that $\sigma_1 \precneqq_{AF} \sigma_2$. This establishes a strictly increasing chain under the partial order:
\begin{equation}
    \sigma \precneqq_{AF} \sigma_1 \precneqq_{AF} \sigma_2 \precneqq_{AF} \dots
\end{equation}
Because the probability simplex (and any corresponding iso-purity manifold $\mathcal{M}_P$) is a compact topological space, and the spectral fluctuations are strictly bounded from above by the absolute maxima of the continuous antiflatness monotones (e.g., $N(d_A)$ and $N_1(d_A)$), this sequence cannot grow indefinitely without bound. By Zorn's Lemma, every bounded chain in this partially ordered set must contain at least one maximal element.
Let this maximal element be $\rho$. Because $\rho$ is maximal, no other state can strictly majorize it, which dictates by definition that $\rho \in \mathcal{P}_{AF}$. Finally, because the antiflat majorization relation $\prec_{AF}$ is strictly transitive, the chain collapses to yield $\sigma \precneqq_{AF} \rho$.
\end{proof}

\subsection{Antiflatness of the gaps}
One can ask: Is it necessary to resolve the full spectrum, or can antiflatness be captured by the gaps between consecutive eigenvalues? This motivates the introduction of an alternative class of antiflatness measures based explicitly on spectral gaps.
We introduce here a measure of antiflatness defined in terms of the gaps between consecutive values in a probability distribution. Specifically, given a probability vector $\vec{\lambda}(\rho_A)^{\downarrow}$, with element in descending order, one can think at a measure of antiflatness as $\Gamma(\vec{\lambda}) = \frac{1}{d_A-1} \sum_{i = 1}^{d_A - 1} (\lambda_i - \lambda_{i+1})^2 = \frac{1}{d_A-1} ||B \vec{\lambda}||_2^2$. Here $B$ is the $(d_A-1) \times d_A$ forward difference matrix with entries $B_{i,i} = 1$ and $B_{i, i+1} = -1$, then $B \vec{\lambda}$ is the vector of adjacent gaps.
Despite its intuitive appeal, this quantity fails to be a faithful antiflatness measure. $\Gamma(\vec{\lambda})$ is not a valid quantifier since it is not faithful on the set flat$_A$ states (see~\cref{App : antiflatness of the gaps} for all the details). To restore faithfulness, one may consider a weighted version by $\vec{p}$ that probes all spectral gaps
$\widetilde{\Gamma}(\vec{\lambda})= \sum_{i<j} p_i p_j (\lambda_i - \lambda_j)^2 = \sum_i p_i \lambda_i^2
- \left( \sum_i p_i \lambda_i \right)^2 = \operatorname{Var}_{p}(\lambda)$.
While this modified measure is faithful, it coincides with one of the variance-based antiflatness previously introduced when $\vec{p}=\vec{\lambda}$, namely the Linear Rényi spread. This observation leads to the conclusion that any faithful gap-based antiflatness measure necessarily reconstructs a global spectral variance. Consequently, the variance-based and entropic quantifiers introduced earlier already provide a complete characterization of spectral antiflatness. Explicit gap-based measures do not yield genuinely new information, but rather re-express the same geometric content in a different form.

\subsection{Compatibility of the quantifiers with the antiflatness ordering}

The introduction of the antiflatness ordering allows us to reinterpret the relations between our established quantifiers. Specifically, we can analyze how the hierarchy established by $\prec_{AF}$ constrains the geometry of the state at both infinitesimal and macroscopic scales. 
Because the condition $\rho \prec_{AF} \sigma$ demands that the Rényi spread of $\rho$ is bounded by the spread of $\sigma$ across all valid intervals $(\alpha, \beta)$, this single constraint governs the behavior of our specific quantifiers. 

\begin{proposition}[Compatibility with the Antiflatness Ordering]
\label{prop:compatibility}
Let $\rho, \sigma \in \mathcal{D}(\mathcal{H})$ be two quantum states. If $\rho$ is antiflat-majorized by $\sigma$ ($\rho \prec_{AF} \sigma$), then the logarithmic antiflatness and the Capacity of Entanglement are strictly monotonically non-decreasing
\begin{align}
    \log(\Lambda_{\rho}) &\le \log(\Lambda_{\sigma}) \ , \label{eq:comp_log} \\
    \mathcal{V}_A(\rho) &\le \mathcal{V}_A(\sigma) \ . \label{eq:comp_capacity}
\end{align}
Furthermore, the Linear Rényi spread $\mathcal{F}_A(\rho)$ acts as a conditional monotone over iso-purity manifolds. If $\rho \prec_{AF} \sigma$ and the states additionally satisfy the purity constraint $\text{Tr}(\rho_A^2) \le \text{Tr}(\sigma_A^2)$ (equivalently, $S_2(\rho_A) \ge S_2(\sigma_A)$), then
\begin{equation}
    \mathcal{F}_A(\rho) \le \mathcal{F}_A(\sigma) \ . \label{eq:comp_linear}
\end{equation}
\end{proposition}

The unconditional relations in Eqs.~\eqref{eq:comp_log} and~\eqref{eq:comp_capacity} stem from the direct functional dependence of these measures on the Rényi spread. 
In contrast, the conditional nature of Eq.~\eqref{eq:comp_linear} arises because $\mathcal{F}_A(\rho)$ is not scale-invariant, as it measures antiflatness weighted by the overall purity of the state. Despite this restriction, the values of $\mathcal{F}_A(\rho)$ remain bounded
\begin{equation}
    \log(\Lambda_\rho)
    \;\le\;  d_A^2 \mathcal{F}_A(\rho) \;\le\; \left(\frac{e^{\Delta_{0\infty}(\rho_A)}}{2}\right)^2.
\end{equation}
Thus, the established antiflatness ordering classifies quantum states, seamlessly integrating all three quantifiers into its geometric framework at different probing scales. Detailed proofs and derivations of these compatibility relations are provided in Appendix \ref{app:compatibility_proofs}.

\section{Escort distributions formalism}
\label{Sec : Escort}
In the study of complex systems, generalized statistical mechanics, and quantum information theory, it is often necessary to analyze the structure of a probability distribution or a quantum state beyond its standard macroscopic observables \cite{li2008entanglement, tsallis1988possible, amari2016information}. A mathematical tool for this purpose is the \textit{escort distribution} (or escort state). Originally introduced in the context of multifractal analysis and later foundational to non-extensive statistical mechanics, escort distributions provide a framework to probe the different scales of a probability measure \cite{chhabra1989direct, beck1993thermodynamics, tsallis1998role}.

The escort distribution of order $q \in \mathbb{R}$, denoted as $P^{(q)}$, is defined by the mapping~\cite{chhabra1989direct, beck1993thermodynamics}
\begin{equation}
    P_i^{(q)} = \frac{p_i^q}{\sum_{j=1}^d p_j^q} \equiv \frac{p_i^q}{Z_q} \; \ ,
\end{equation}
where $Z_q = \sum_{j} p_j^q$ ensures normalization and $\{p_i\}_{i=1}^d$ represents the probability inside the simplex $\mathbf{\Delta}$ for a classical system. 
In the quantum regime, for a strictly positive density matrix in a finite-dimension Hilbert space; $\rho \in \mathcal{D}(\mathcal{H})$ (i.e., $\rho > 0$),
the escort state of order $q$ is defined via the functional calculus as~\cite{tsallis2009introduction} 
\begin{equation}
    \rho_{\text{escort}}^{(q)} = \frac{\rho^q}{\text{Tr}(\rho^q)} \; \ .
\end{equation}
The parameter $q$ acts as an analytical ``zoom lens" on the spectrum of $\rho$~\cite{chhabra1989direct}. 
For $q > 1$, the escort mapping amplifies the largest eigenvalues, suppressing the contributions from the tails of the distribution. While, for $q < 1$, it accentuates the smallest eigenvalues, giving weight to rare fluctuations. In the limit $q \to \infty$, the distribution collapses onto the maximum probability subspace, while $q \to 0$ yields the uniform distribution over the support of $\rho$. The mathematical structure of the escort distribution is connected to the foundations of statistical mechanics, specifically the continuous mapping between the microcanonical and canonical ensembles\cite{beck1993thermodynamics, tsallis1998role, abe2003geometry}. 
By defining an effective modular Hamiltonian for the state as $H = -\log \rho$~\cite{jaynes1957information, haag1992local}, the original density matrix can be exactly recast as a thermal Gibbs state at unit temperature
\begin{equation}
    \rho = \frac{e^{-H}}{\text{Tr}(e^{-H})} \; \ .
\end{equation}
When we apply the escort transformation of order $q$ to this state, the resulting density matrix takes the form
\begin{equation}
    \rho_{\text{escort}}^{(q)} = \frac{\rho^q}{\text{Tr}(\rho^q)} = \frac{e^{-qH}}{\text{Tr}(e^{-qH})} \; \ .
\end{equation}
In this physical picture, the escort parameter $q$ plays the exact role of an inverse temperature $\beta$. This reveals that cooling or heating the system is physically isomorphic to varying the escort parameter, recovering the microcanonical limit for $q \to 0$ ($\beta \to 0$). The system explores all accessible microstates with strictly equal probability, yielding a flat spectrum over the support of $\rho$. This flat state is the quantum mechanical definition of the microcanonical ensemble. 
When $q = 1$, we have the canonical state, operating as a standard canonical ensemble at the equilibrium inverse temperature $\beta = 1$. 
For $q \to \infty$, we have the ground state limit where thermal fluctuations vanish entirely, and the distribution collapses onto the ground state of $H$.

In this framework, the flat distribution plays the role of a microcanonical state in probability space, while deviations from flatness are resolved by canonical deformations along the escort parameter $q$, or physically by varying the temperature $\beta$.
\subsection{Relation among different measures}\label{sec:escort_derivations}
All the measures of antiflatness can be reformulated using the framework of escort distributions. 
Our aim is to show a unified interpretation in terms of the escort partition function
\(
\mathbb{Z}_{q} = \sum_i p_i^{q}
\),
which generates a one-parameter family of escort distributions probing the internal structure of a given spectrum.
When $\alpha \neq 1$, then
\begin{equation}
\mathbb{Z}_{\alpha}
:= \sum_{i} p_i^{\alpha}
= \exp\!\big( (1-\alpha)\, S_{\alpha}(\rho) \big) \, ,
\end{equation}
with Rényi entropies $S_{\alpha}(p)
= \frac{1}{1-\alpha}\,\log\!\left( \mathbb{Z}_{\alpha} \right)$.

Applying this identity, as shown in~\cref{App : relation quantifiers}, we can express the relevant quantities as follows
\begin{equation}
\label{eq:unification}
\begin{aligned}
&\mathcal{F}_A(\rho) = 1 - \mathbb{Z}_3 - \left(2 - \mathbb{Z}_2\right)^2= e^{-2\, S_3(\rho_A)} - e^{-2\, S_2(\rho_A)} \, , \\
&\log \Lambda_{\rho} = \log \mathbb{Z}_{3} - 2 \log \mathbb{Z}_{2} = 2\big( S_{2}(\rho_A) - S_{3}(\rho_A) \big) \, , \\
&\mathcal{V}_A(\rho) =  \left. \frac{d^{2}}{dq^{2}}
\Big[ (1-q)\, S_{q}(\rho_A) \Big] \right|_{q=1} \, .
\end{aligned}
\end{equation}  

Clearly, the logarithmic antiflatness of $\rho$ remains unchanged, since it is expressed directly in terms of Rényi entropies. Because all three quantifiers are fundamentally derived from the escort partition function $\mathbb{Z}_q$ 
, they are strictly faithful to the antiflatness geometry: they vanish if and only if $\rho \in \text{flat}_A$. As illustrated in Fig.~\ref{fig:relations} for the case of a reduced system composed of one qubit with Schmidt spectrum $(\lambda, 1-\lambda)$, all measures correctly identify the free flat$_A$ states at $\lambda = 0.5$ and $\lambda \in \{0, 1\}$ with exact zeroes. However, because each measure applies a distinct algebraic weighting to the escort distribution, they exhibit different sensitivities to spectral fluctuations. Taken together, these quantities do not represent independent notions of antiflatness, but rather complementary informations of the same underlying Rényi spread $\Delta_{\alpha\beta}$, each probing the spectrum at a different statistical scale. 
This differential sensitivity becomesapparent when analyzing their absolute maxima as showed in Sec.~\ref{Sec : Pareto} (see Fig.~\ref{fig:relations}).

\begin{figure}[t]
    \centering
    \includegraphics[width=0.6\linewidth]{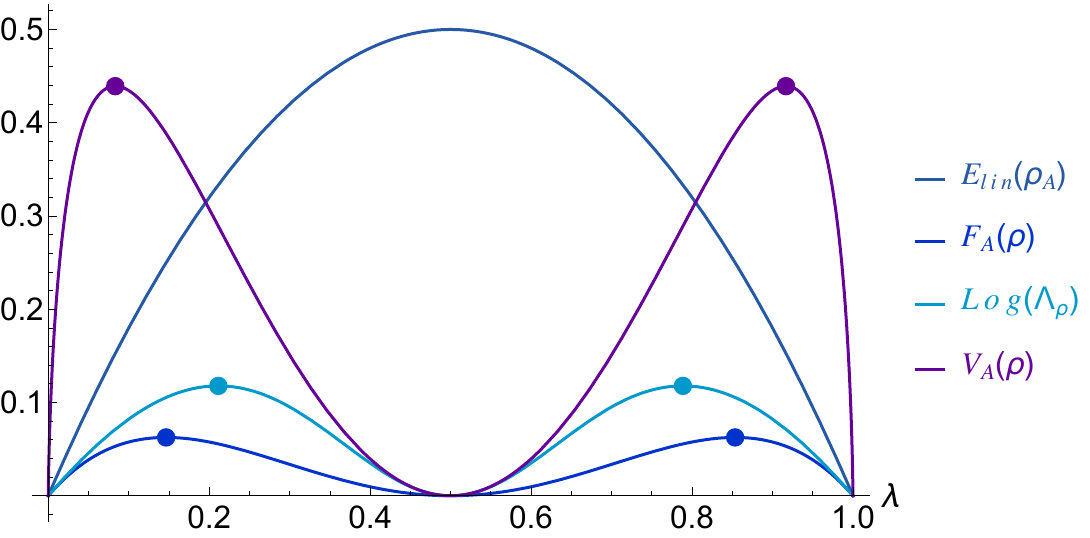}
    \caption{Different measures of antiflatness, $\mathcal{F}_A(\rho)$, $\log(\Lambda_\rho)$, $\mathcal{V}_A(\rho)$ and linear entanglement $E_{\text{lin}}(\rho_A):=1- \Tr(\rho_A^2)$ as a function of the Schimdt coefficient $\lambda$ for the case of pure state in an Hilbert space of dimension $d$ and $d_A=2$. The dots correspond to different maxima for each measure of antiflatness above: $\lambda_{\text{max}}^{\mathcal{F}}=\frac{1}{16} \left(8-4 \sqrt{2}\right)$, $\lambda_{\text{max}}^{\log(\Lambda)}=\frac{1}{6} \left(3-\sqrt{3}\right)$ and $\lambda_{\text{max}}^{\mathcal{V}_A}\approx 0.0832217$ and their reflection trough $\lambda=1/2$.}
    \label{fig:relations}
\end{figure}

\section{Geometric Formulation of Antiflatness via Bregman Divergences}
\label{Sec : Bregman}
Up to this point, our characterization of antiflatness has relied on the statistical properties of an ES, specifically through the lens of Rényi entropy spreads and escort distributions. However, these purely spectral quantities admit a unique geometric interpretation. When one wants to quantify a resource, it is always possible to think at the problem introducing a geometric distance to quantify how far a state is from the set of free states \cite{chitambar2019quantum, gour2024resources}. We can reframe our spectral fluctuations as generalized geometric distances on the probability simplex by employing the framework of Bregman divergences~\cite{bregman1967relaxation, amari2016information}.

Let $F(p)$ be a strictly convex and differentiable generating function defined over the probability simplex. The Bregman divergence between two probability distributions $p$ and $q$ is defined as
\begin{equation}
    D_F(p||q) = F(p) - F(q) - \sum_{i=1}^d \frac{\partial F(q)}{\partial q_i}(p_i - q_i) \; \ .
\end{equation}
While not a true metric (as it generally lacks symmetry and the triangle inequality), it establishes a statistical manifold that acts as a geometric susceptibility, allowing to evaluate the antiflatness of a state $\rho$ relative to a target state $\sigma$.

\subsection{Capacity of Entanglement and Kullback-Leibler Curvature}
Let us first consider the generating function $F(\rho) = \sum_{i} \lambda_{\rho_i} \log \lambda_{\rho_i} = -S_1(\rho)$, which corresponds to the negative Von Neumann entropy. In this case, the Bregman divergence uniquely reduces to the Quantum Relative Entropy, or Kullback-Leibler (KL) divergence
\begin{equation}
    D_{KL}(\rho||\sigma) = \sum_{i=1}^d \lambda_{\rho_i} \log\left(\frac{\lambda_{\rho_i}}{\lambda_{\sigma_i}}\right) \; \ ,
\end{equation}
where $\lambda_{\rho}$ and $\lambda_{\sigma_i}$ indicate the spectra of quantum states $\rho$ and $\sigma$, respectively. 
We evaluate the divergence between the state $\rho_A$, that now for simplicity we will call $p$, and its own generalized escort distribution of order $\alpha = 1+\epsilon$, defined by the eigenvalues $P^{(\alpha)}_i = p_i^\alpha / \sum_j p_j^\alpha$. As derived in ~\cref{App: DKL Derivation}, by expanding the KL divergence to second order under an infinitesimal escort perturbation $\alpha = 1+\epsilon$, the linear terms vanish, leaving a purely quadratic form. Isolating the variance of the surprisal operator, $\mathcal{V}_A(\rho) = \text{Var}_{\rho_A}(-\log \rho_A)$, yields the following identity
\begin{equation}
\begin{split}
\mathcal{V}_A(\rho) &= \lim_{\epsilon \to 0} \frac{2}{\epsilon^2} D_{KL}(\rho_A || \rho_{A_\text{escort}}^{(1+\epsilon)}) = \left. \frac{d^2}{d\alpha^2} D_{KL}(\rho_A || \rho_{A_\text{escort}}^{(\alpha)}) \right|_{\alpha=1} \ .
\end{split}
\label{eq:capDKL}
\end{equation}
Eq.~\eqref{eq:capDKL} reveals that $\mathcal{V}_A(\rho)$ is a differential, local property representing the Bregman curvature along the statistical trajectory defined by escort distributions. The more antiflat a state is, the faster it diverges geometrically when perturbed along the escort direction. If a distribution is perfectly flat, its escort distributions are identical to it ($D_{KL} = 0$), and thus the curvature vanishes. 

\subsubsection{Information-Geometric Interpretation via Quantum Fisher Information}
The identity established in Eq.~\eqref{eq:capDKL} reveals a connection to the framework of quantum information geometry. Quantum generalizations of the Fisher Information Matrix (QFIM) arise naturally as the Hessian matrix in the Taylor expansion of smooth divergences, such as the Rényi or Kullback-Leibler relative entropies, evaluated between infinitesimally close quantum states~\cite{wilde2025quantumfisher}. In our derivation, the continuous family of generalized escort distributions $\rho_{\text{escort}}^{(\alpha)}$ forms a one-dimensional statistical manifold parameterized by the escort index $\alpha$ (where $\alpha = 1+\epsilon$). By evaluating the second derivative of the Kullback-Leibler divergence with respect to this parameter evaluated at the physical state ($\alpha=1$), we compute the scalar Quantum Fisher Information along the escort trajectory. Therefore, we have
\begin{equation}
\label{eq:QFI}
    I_Q = \left. \frac{d^2}{d\alpha^2} D_{KL}(\rho || \rho_{\text{escort}}^{(\alpha)}) \right|_{\alpha=1} = \mathcal{V}_A(\rho) \; \ ,
\end{equation}
and notice that $\mathcal{V}_A(\rho)$ is not just a statistical variance, but a geometric metric, the Fisher distinguishability, that dictates how fast a quantum state separates from its own thermal-like escort deformations at the macroscopic scale. 

In classical statistics, Chentsov's Theorem guarantees that there is only one unique Riemannian metric on the probability simplex: the Fisher Information metric \cite{chentsov1982statistical, amari2016information}. In quantum information geometry, this uniqueness breaks down. According to Petz's Classification Theorem, there are infinitely many distinct quantum generalizations of the Fisher Information metric \cite{petz1996monotone, amari2016information}. Every single one of these metrics collapses back to the standard classical Fisher Information if the density matrices commute ($[\rho, \sigma] = 0$), but they behave differently otherwise. If you take a parameterized family of quantum states $\rho_\theta$ and compute the second derivative of the standard (Umegaki) quantum relative entropy $D(\rho_\theta || \rho_{\theta+d\theta})$ \cite{umegaki1962conditional}, you do not get the ``standard" Quantum Fisher Information that is most commonly used in quantum metrology (which is based on the Symmetric Logarithmic Derivative, or SLD) \cite{braunstein1994statistical}. Instead, the second derivative of the standard quantum relative entropy yields a specific quantum Fisher metric known as the Bogoliubov-Kubo-Mori (BKM) metric (also called the canonical inner product) \cite{petz1996monotone}. In our framework, we are evaluating the divergence between a generic state $\rho \in \mathcal{H_A}$ and its own escort distribution $\rho_{\text{escort}}^{(\alpha)}$. Because the escort distribution is defined entirely by taking powers of the eigenvalues of $\rho$, the state and its escort deformation strictly commute, $[\rho,\rho_{\text{escort}}^{(\alpha)}] = 0$. Consequently, Petz's theorem dictates that all the infinitely many quantum Fisher metrics collapse into a single, unique value.

In standard literature, the relative variance is typically defined with respect to an external reference state $\sigma$, measuring the fluctuations of the relative surprisal $\log \rho - \log \sigma$ \cite{boes2022variance}. Such formulations inherently rely on an \textit{extrinsic} geometry: the curvature is evaluated by perturbing the state $\rho$ against a fixed, external background, e.g. a thermal Gibbs state. While this is highly effective for standard thermodynamics, it is suboptimal for characterizing the internal spectral shape of a state. Conversely, our formulation defines the divergence $D_{KL}(\rho || \rho_{\text{escort}}^{(\alpha)})$ entirely through the generalized escort distributions. Because the escort state $\rho_{\text{escort}}^{(\alpha)}$ is intrinsically generated by the spectrum of $\rho$ itself, this approach relies purely on an \textit{intrinsic} geometry. We are not measuring the distance from $\rho$ to an arbitrary external environment; rather, we are measuring the distance from $\rho$ to its own deformed spectral deformations. Therefore, $\mathcal{V}_A(\rho)$ quantifies the susceptibility of the state. It defines the rate at which the state becomes statistically distinguishable from itself when perturbed along the flow of its escort distributions, bridging the thermodynamic concept of energy fluctuations with the geometric limits of quantum estimation theory.

It is important to stress that the escort parameter $\alpha$ is not, in general, a physical control parameter of the many-body Hamiltonian. 
In our setting, $\alpha$ should rather be understood as a coordinate that selects a direction in the space of entanglement spectra. The capacity of entanglement is therefore not primarily used here as a Cram\'er--Rao bound for estimating $\alpha$, although such an interpretation would be formally possible if the escort family were physically preparable. 
Its main role is geometric: $\mathcal{V}_A(\rho)$ is the Fisher norm of the tangent generated by an infinitesimal modular-temperature deformation of the ES.

This distinction is useful when the state depends on a genuine physical parameter $\lambda$.  In that case, the quantum Fisher information associated with estimating $\lambda$ is determined by the physical tangent $\partial_\lambda\rho_A(\lambda)$, and is generally different from $\mathcal{V}_A(\lambda)$. 
The latter should instead be interpreted as an intrinsic spectral susceptibility: it measures how sensitive the ES at fixed $\lambda$ is to an escort deformation. 
Thus, peaks or sharp variations of $\mathcal{V}_A(\lambda)$, or of its derivatives with respect to $\lambda$, can signal a reorganization of the entanglement spectrum even when standard entropic quantities vary smoothly. Flat spectra make this interpretation immediate. If $\rho_A$ is proportional to a projector on its support, the escort deformation leaves it invariant for all $\alpha$, and the Fisher norm vanishes. 
Conversely, a non-flat spectrum has nonzero modular-energy fluctuations and therefore a nonzero escort susceptibility.

\subsection{The Euclidean Divergence and $\mathcal{F}_A(\rho)$}
To establish the geometric nature of the Linear Rényi spread $\mathcal{F}_A(\rho)$, we select as generating function $F(p) = \mathbb{Z}_2 -1 $, the negative 2-Tsallis entropy, or equivalently $F(p)\equiv \mathbb{Z}_2$, by the invariance under affine transformations of Bregman divergences. The scalar Bregman divergence associated with this choice simplifies to the squared Euclidean distance
\begin{equation}
\begin{split}
    D_F(p||q) &= \sum_{i=1}^d p_i^2 - \sum_{i=1}^d q_i^2 - \sum_{i=1}^d (p_i - q_i)(2q_i)= \sum_{i=1}^d (p_i - q_i)^2 \; \ .
\end{split}    
\end{equation}
Being the spectrum of a state a probability distribution $p$ that takes the value $p_i$ with probability $p_i$, the expected value (the mean) of the spectrum is exactly the purity, i.e. $\mathbb{E}_p[p]= \mathbb{Z}_2$. The measure $\mathcal{F}_A(\rho)$ is the variance of this distribution, which can be expressed as the expected Bregman divergence of the individual eigenvalues from their mean purity
\begin{equation}
\begin{split}
    \mathbb{E}_p \big[ D_F(p_i || \mathbb{Z}_2) \big] &= \sum_{i=1}^{d_A} p_i (p_i - \mathbb{Z}_2)^2 = \sum_{i=1}^{d_A} p_i^3 - 2\mathbb{Z}_2 \sum_{i=1}^{d_A} p_i^2 + \mathbb{Z}_2^2 \sum_{i=1}^{d_A} p_i= \mathbb{Z}_3 - \mathbb{Z}_2^2 = \mathcal{F}_A(\rho) \ .
\end{split}
\end{equation}
$\mathcal{F}_A(\rho)$ equals the average Euclidean divergence of the ES from its own purity. Because this geometric distance is inherently anchored to $\mathbb{Z}_2$, it becomes evident why $\mathcal{F}_A(\rho)$ preserves the antiflatness order $\rho \prec_{AF} \sigma$ only when constrained within an iso-purity manifold, see Prop.~\ref{prop:compatibility}.
Hence $\mathcal{F}_A$ measures the average squared Euclidean separation between two eigenvalues independently sampled from the ES, a pairwise Euclidean dispersion of the spectral weights.

\subsection{Unification of Quantifiers via the Coefficient of Variation}
We can now notice that is possible to unify $\mathcal{F}_A(\rho)$ and $\log(\Lambda_{\rho})$, from the definition of $\Lambda_{\rho} = \frac{\mathbb{Z}_3}{\mathbb{Z}_2^2}$,
\begin{equation}
\label{unification}
    \log(\Lambda_{\rho}) = \log\left( 1 + \frac{\mathcal{F}_A(\rho)}{\mathbb{Z}_2^2} \right) = \log(1 + \frac{Var_p(p)}{\mathbb{E}_p[p]^2}).
\end{equation}
By noticing that $\frac{Var_p(p)}{\mathbb{E}_p[p]^2}$ is the squared coefficient of variation for the eigenvalues of $\rho_A$, in the limit of small fluctuations, we can apply the approximation $\log(1+x) \simeq x$, yielding $\log(\Lambda_{\rho}) \simeq \frac{\mathcal{F}_A(\rho)}{Z_2^2} $.
While $\log(\Lambda_{\rho})$ acts as an intensive, universal quantifier, $\mathcal{F}_A(\rho)$ acts as an extensive, local quantifier modulated by purity. 

The identity in Eq.~\eqref{unification} provides a practical advantage: it makes macroscopic antiflatness experimentally accessible without the need for full quantum state tomography, compared to Von Neumann entropy and the Capacity of Entanglement, which rely on the logarithmic spectrum ($-\log \rho$), a procedure that scales exponentially with system size~\cite{haeffner2005scalable}. In contrast, the Linear Rényi spread $\mathcal{F}_A(\rho) = \mathbb{Z}_3 - \mathbb{Z}_2^2$ depends strictly on the integer moments of the reduced density matrix, which can be efficiently extracted using multi-copy interferometry~\cite{ekert2002direct, daley2012measuring} or randomized measurement protocols like classical shadows~\cite{huang2020predicting, elben2023randomized}.
Consequently, evaluating the universal spectral fluctuations introduces zero additional experimental overhead compared to standard purity measurements.

Having established the maximal bounds for $\mathcal{F}_A(\rho)$ and $\log(\Lambda_\rho)$ in Theorem~\ref{thm:Maximal_Antiflatness}, we can highlight a relationship between $N_{\mathcal{F}_A}(d_A)$ and $N_{\Lambda}(d_A)$ by leveraging the identity derived in Eq.~\eqref{eq:unification}.
Because the purity is globally bounded by $\frac{1}{d_A} \le \mathbb{Z}_2 \le 1$, the inverse squared purity is strictly bounded by $1 \le \frac{1}{\mathbb{Z}_2^2} \le d_A^2$, then
\begin{equation}
    \log(1 + \mathcal{F}_A(\rho)) \le \log(\Lambda_\rho) \le \log(1 + d_A^2 \mathcal{F}_A(\rho)) \ .
\end{equation}
We can now evaluate these inequalities at the maximum configurations. Because $N_\Lambda(d_A)$ is the absolute maximum of $\log(\Lambda_\rho)$ across the entire simplex, it must be strictly greater than or equal to the logarithm evaluated at the state that maximizes $\mathcal{F}_A(\rho)$, yielding to
\begin{equation}
 N_{\mathcal{F}_A}(d_A)  \le \log(1 + N_{\mathcal{F}_A}(d_A)) \leq N_\Lambda(d_A) \leq \log(1 + d_A^2 N_{\mathcal{F}_A}(d_A))\ .
\end{equation}
since $N_{\mathcal{F}_A}(d_A)$ is strictly bounded from above by $27/256 \approx 0.105$.

\section{Ensembles of states}
\label{Sec : Averages}
The previous sections characterized antiflatness at the level of individual spectra. 
We now turn to a complementary question: what is the typical amount of antiflatness generated by natural ensembles of quantum states?
The ensembles considered below should therefore be understood as probes of typical spectral fluctuations associated with different notions of randomness.
For Haar-random pure states, concentration of measure implies that, in large bipartite Hilbert spaces, the reduced state of a sufficiently small subsystem is overwhelmingly close to maximally mixed~\cite{Popescu_2006, Hayden_2007}. 
Equivalently, the ES is typically close to flat. 
One therefore expects antiflatness to be a subleading fluctuation rather than an extensive feature of a typical Haar state. 
The purpose of this section is to make this statement quantitative, and to compare it with the corresponding behavior in the Bures--Hall and Clifford ensembles.

\subsubsection{Haar measure}
We first consider the unitarily invariant ensemble of pure states. Because normalization and global phase are physically irrelevant, the space of pure states on $\mathcal{H}$ is the complex projective space
$\mathbb{C}P^{d-1}$~\cite{nielsen_chuang_2010}. 
On this manifold, of real dimension $2d-2$, there exists a unique, unitarily invariant measure $d\psi$~\cite{zyczkowski2006GeometryQuantumStates,wootters1990random}. This measure is induced by the Haar measure $dU$ on the unitary group $U(d)$, acting on a fixed fiducial state $|\psi_0\rangle\in\mathcal{H}$, namely  
\begin{equation}
\int_{\mathbb{C}P^{d-1}} d\psi\, g(|\psi\rangle)
=\int_{U(d)} dU\, g(U|\psi_0\rangle)\,,
\label{eq:Haar_projective}
\end{equation}
for any integrable function $g$. Equipping the set of pure states with this uniform measure turns any real-valued function $g(|\psi\rangle)$ into a random variable, whose expectation value is computed as 
$\mathbb{E}_{\psi}[g]=\int d\psi\, g(|\psi\rangle)=\int dU\, g(U|\psi_0\rangle)=\mathbb{E}_U[g(\psi)]$~\cite{Mele_2024,collins2006integration,watrous2018theory,Collins_2015}.  
In practice, it is often convenient to express averages over pure states in terms of the eigenvalues of the reduced density matrix of a subsystem. Consider a bipartition of 
$\mathcal{H}=\mathcal{H}_A\otimes \mathcal{H}_B$ with dimensions $\text{dim}(\mathcal{H})=d_A d_B=d$. For a pure state $|\psi\rangle$, let 
$\rho_A=\mathrm{Tr}_B|\psi\rangle\langle\psi|$ denote the reduced density matrix of subsystem $A$. The unitarily invariant measure on $\mathbb{C}P^{d-1}$ induces a measure on the eigenvalues $\lambda_1,\dots,\lambda_{d_A}$ of $\rho_A$~\cite{lloyd_pagels_1988}, as showed in \eqref{eq:Haar measure}.
This formulation allows one to compute expectation values of functions $g(\rho)$ as 
\begin{equation}
\mathbb{E}_\psi[g] = \int d\mu(\lambda_1,\dots,\lambda_{d_A})\, g(\{\lambda_i\})\,.
\end{equation}

Under this measure, the average value of the antiflatness $\mathcal{F}_A(\rho_A)$ for a reduced density matrix $\rho_A$ of a pure state $\rho$ is given
\footnote{Remember that in the case of the same dimensions for the subsystems one has the ensemble induced by the Hilbert-Schmidt measure.} by ~\cref{app:Haar average}
\begin{equation}
\mathbb{E}_U[\mathcal{F}_A(\rho)] 
= \frac{\left(d_A^2-1\right) \left(d_B^2-1\right)}{(d+1) (d+2) (d+3)}
= \Theta(d^{-1})\,,
\end{equation}
which can be computed using the tools developed in~\cite{liu2018entanglement,wei2019exact}.
For the case of two-qubit states $\rho$, one finds $\mathbb{E}[\mathcal{F}_A(\rho)] = 3/70$, and can explicitly compute the full probability density function
\begin{equation}
\begin{split}
    P_{\mathcal{F}_A(\rho)}^{U}(f)&= 3 \sqrt{2} \left( \frac{\sqrt{\left(\sqrt{1-16 f}-1\right) (16 f-1)}}{1-16 f}+\sqrt{\frac{1}{\sqrt{1-16 f}}+\frac{1}{1-16 f}} \right)\,,
\end{split}
\end{equation}
see Fig.~\ref{fig:pdf_2_FA} and ~\cref{app:Haar average} for further details.

It is important to notice that a square root singularity appears in the distribution for two qubits, reminiscent of what happens with the logarithmic singularity in \cite{iannotti2025_IVH} for SREs.

Another interesting regime for the Haar average is when $d_B \gg d_A \gg 1$
\begin{equation}
    \mathds{E}_U[\mathcal{F}_A(\rho)] \simeq\frac{1}{d_A^2}\, 
\end{equation}
By estimating the variance, see ~\cref{app:Haar average}, of $\mathcal{F}_A(\rho)$ one can use the Chebyshev's inequality, i.e.
\begin{equation}
    \Delta^2 \mathcal{F}_A(\rho)= \mathds{E}_U[\mathcal{F}_A^2(\rho)]-\mathds{E}_U[\mathcal{F}_A(\rho)]^2 =\Theta\left(\frac{1}{d^3}\right)\,,
\end{equation}
one can argue the concentration of $\mathcal{F}_A(\rho)$ around its mean value for large dimensions.
Due to the Lipschitzianity of the function, the Lévy Lemma \cite{watrous2018theory,zyczkowski2006GeometryQuantumStates,milman1986asymptotic} for the concentration of measure also applies.

It is straightforward to see that the relative growth of the average anti-flatness compared to the average linear entanglement entropy, i.e. $E_{\rm lin}[\rho] = 1 - \mathrm{Tr}[\rho^2]$,
which is closely related to their logarithmic versions, the capacity of entanglement and the von Neumann entanglement entropy, and has important consequences in quantum field theory~\cite{boes2022variance}, vanishes in the limit of large subsystem dimensions as well~\cite{wei2023average}:
\begin{equation}
    \lim_{d_A,d_B \to \infty} \frac{\mathbb{E}_U[\mathcal{F}_A(\rho)]}{\mathbb{E}_U[E_{\rm lin}(\rho)]} = 0 \,.
    \label{eq:tendtozero}
\end{equation}
In the case of the capacity of entanglement, the Haar average was computed using different techniques in~\cite{okuyama2021capacity}.

\subsubsection{Bures-Hall measure}
\begin{figure}
    \centering
    \includegraphics[width=0.45\linewidth]{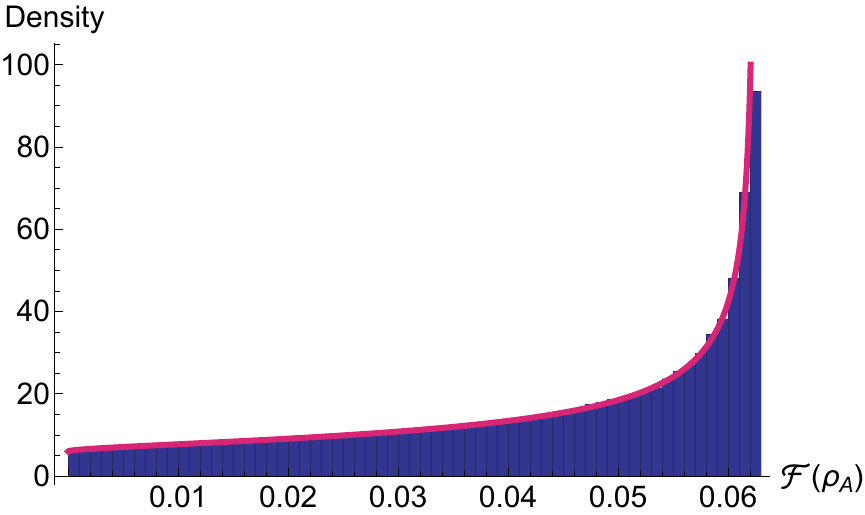}
    \includegraphics[width=0.45\linewidth]{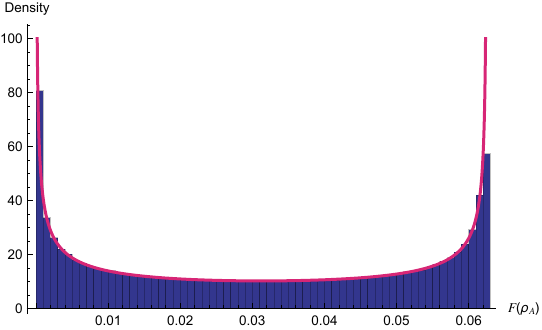}
    \caption{\emph{Left.} Probability density function of $\mathcal{F}_A(\rho)$ for 2 qubits random pure states induced by the Haar measure, with the violet line representing the analytical result and the histogram depicting the numerical computation reasalized using $N_{samples}=2 \times 10^6$. \emph{Right.} Probability density function of $\mathcal{F}_A(\rho)$ for 2 qubits random pure states induced by the Bures probability distribution, with the violet line representing the analytical result and the histogram depicting the numerical computation reasalized using $N_{samples}=2 \times 10^6$.}
    \label{fig:pdf_2_FA}
\end{figure}
In parallel with the Haar distribution over pure states, one can define the Bures--Hall ensemble of mixed states via a purification plus ``twist'' construction introduced in~\cite{hall1998random,zyczkowski2001induced,sommers2003bures} which we enforce on the reduced density state.
The Bures metric, which induces the above ensemble of states~\footnote{Similarly, the Hilbert-Schmidt metric does for Haar random pure states~\cite{zyczkowski2006GeometryQuantumStates}}, is a distinguished information-geometric metric on the space of quantum states, and it is closely tied to quantum state distinguishability. In the standard finite-dimensional setting, the Bures metric coincides with the symmetric-logarithmic-derivative quantum Fisher information metric up to the usual normalization convention~\cite{zyczkowski2006GeometryQuantumStates}. 
For this reason, we study its implications for various antiflatness measures.

A random pure state $\ket{\psi} \in \mathcal{H}_A \otimes \mathcal{H}_B$ is first formed using complex Gaussian coefficients, namely a random pure state, and then a biased superposition is made via a random unitary $U$ on subsystem $A$. Concretely, $|\phi\rangle = |\psi\rangle + (U_A\otimes \mathbb{I}_B)\,|\psi\rangle\ $,
where $U\in U(d_A)$ is drawn from a unitary measure weighted by $\det(\mathbb{I}+U)^{2\alpha+1}$ \cite{Sarkar_2019}.  
The density matrix of the composite system is then $\rho = |\phi\rangle\langle \phi|$ upon normalization, and the reduced state on subsystem $A$ is $\rho_A = \mathrm{Tr}_B\,\rho$.  
Under this construction, the eigenvalue distribution of $\rho_A$ is exactly the Bures--Hall distribution. The resulting joint eigenvalue probability density \cite{zyczkowski2006GeometryQuantumStates} is given by

\begin{equation}
\begin{split}
d\mu_{\mathrm{BH}}\left(\vec{\lambda}\right)&\propto\delta (1-\sum_{i=1}^{d_A} \lambda_i )\prod_{1 \leq i<j \leq d_A} \frac{\left(\lambda_i-\lambda_j\right)^2}{\lambda_i+\lambda_j} \prod_{i=1}^{d_A} \lambda_i^\alpha\,,
\end{split}
\label{eq:Bures_Hall_pdf}
\end{equation}
where $\alpha = d_B - d_A - \tfrac{1}{2}$.  
Equipping the space of reduced density matrices $\mathcal{D}(\mathcal{H}_A)$ with this measure defines the Bures--Hall ensemble.  

In the case of Bures-Hall metric, one can extract the average value of the antiflatness using \cite{al2010random,li2021moments}
\begin{equation}
    \mathbb{E}_{\mathrm{BH}}[\mathcal{F}_A(\rho)]= \frac{(d-1) (7 d-8)}{4 (d+2) (d+4) (d+6)}= \Theta(d^{-1})\,,
\end{equation}
in the case $d_A=d_B=\sqrt{d}$.
Also, in this case, the relative growth of
the average anti-flatness compared to the average linear
entanglement entropy \eqref{eq:tendtozero} tends to zero for large Hilbert spaces.
In the case of 2 qubits, $d_A=d_B=2$ we can also give an explicit formula for the probability density function as in the case for the Haar average
\begin{equation}
    P_{\mathcal{F}_A(\rho)}^{BH}(f)=\frac{4}{\pi  \sqrt{f(1-16 f)} }\,,
\end{equation}
see Fig. \ref{fig:pdf_2_FA}.
Instead, a full derivation of the entanglement capacity average over the Bures-Hall ensemble can be found in~\cite{wei2023average}.

\subsubsection{Clifford ensembles}
Finally, we discuss the Clifford ensembles.
The Clifford group plays a central role in the study of quantum randomness and quantum resource theories, providing the backbone of stabilizer formalism and efficient classical simulation. 
Formally, the \emph{$n$-qubit Clifford group} $\mathcal{C}_n$ is defined as the normalizer of the Pauli group $\mathcal{P}_n$ within the unitary group $U(2^n)$, namely $\mathcal{C}_n = \{\, U \in U(2^n) \;|\; U \mathcal{P}_n U^\dagger = \mathcal{P}_n \,\}$.
Elements of $\mathcal{C}_n$ form a unitary 3-design~\cite{zhu2016CliffordGroupFails}, meaning that averages of polynomial functions of degree up to three over $\mathcal{C}_n$ coincide with the corresponding Haar averages on $U(2^n)$. 
However, since the Clifford group is a finite subgroup of $U(2^n)$, it cannot reproduce higher-order moments of the Haar measure, which are crucial to capture genuine non-Clifford correlations.

To interpolate between stabilizer-preserving Clifford unitaries and fully random (Haar-distributed) unitaries, one introduces the \emph{$k$-doped Clifford ensemble}~\cite{leone2025noncliffordcostrandomunitaries}. 
This ensemble is constructed by composing a Clifford unitary $C \in \mathcal{C}_n$ with the application of a fixed non-Clifford single-qubit gate $K$ on a subset of $k$ qubits. 
Formally, the ensemble of $k$-doped Clifford circuits is defined as
\begin{equation}
\begin{split}
    C_k = \bigl\{\, U &= (\prod_{j=1}^k C_j\, K^{(S)}) \,C_0 \;\big|\; C_0,C_j \in \mathcal{C}_n,\; S \subseteq \{1,\dots,n\},\; |S| = k \,\bigr\},
\end{split}
\end{equation}
where $K^{(S)} = \bigotimes_{i=1}^n K_i$ acts as $K$ on qubits $i \in S$ and as the identity elsewhere. 
A typical choice of the non-Clifford gate is a diagonal phase gate of the form
\begin{equation}
    K = \ket{0}\!\bra{0} + e^{i\theta}\ket{1}\!\bra{1}, 
    \quad \theta \in [0,2\pi],
\end{equation}
which breaks the stabilizer structure while preserving a controllable amount of Clifford symmetry. 

The $k$-doped Clifford ensemble thus defines a tunable family of unitary distributions interpolating between the Clifford group ($k=0$) and the full Haar measure ($k=n$). 
By varying $k$, one can probe how non-Cliffordness, entanglement generation, and statistical properties of random circuits evolve from efficiently simulable regimes to fully generic quantum behavior.

This ensemble is naturally equipped with a probability measure. 
Let $d\mu_{\mathrm{Cl}}$ be the uniform measure on $\mathcal{C}_n$ and $d\mu_1$ a measure over $U(1)$ for the single-qubit gate $K$. 
The measure over $C_k$ is the convolution of Clifford and non-Clifford layers,
\begin{equation}
d\mu_{C_k}(U) = d\mu_{\mathrm{Cl}}(C_k) \star d\mu_1(K^{(S_k)}) \star \cdots \star d\mu_{\mathrm{Cl}}(C_0),
\end{equation}
where $\star$ denotes convolution. For any integrable $f(U)$, the ensemble average reads
\begin{equation}
\begin{split}
    \mathbb{E}_{C_k}[f] &= \int_{C_k} d\mu_{C_k}(U)\, f(U)= \int_{\mathcal{C}_n} dC_k \int_{U(1)} d\mu_1(K^{(S_k)}) \cdots \int_{\mathcal{C}_n} dC_0\, f\Big(\prod_{j=1}^k C_j K^{(S_j)} C_0\Big).
\end{split}
\end{equation}

As antiflatness captures some features of both entanglement and magic, we use the average over the Clifford group as an important tool to investigate this relationship through the eyes of antiflatness.

Considering the k-doped Clifford average \cite{leone2025noncliffordcostrandomunitaries,leone2021quantum} we have for $\ket{\psi}=\ket{0}^{\otimes n}$, $d_A=d_B=\sqrt{d}$ and $K=\ket{0}\bra{0}+e^{i \theta} \ket{1}\bra{1}$
\begin{equation}
\begin{aligned}
    \mathds{E}_{C_k}[\mathcal{F}_A(\rho)] &= \frac{5d+1}{(d+1)(d+2)}- \frac{2(2 d^2 +9 d+1)+(d^2-2d+1)(f_\theta^-)^k}{(d+1)(d+2)(d+3)}
\end{aligned}
\end{equation}
with $f_\theta^-= \frac{7d^2-3d+d(d+3)\cos(4 \theta) -8}{8(d^2-1)}$ for $\theta \in [0,2 \pi]  \backslash \{\pm \frac{\pi}{2}\}$ otherwise is equal to one and $\rho_A=tr_B \ket{\psi}\bra{\psi}$.
In the case the $k=0$ one can write a close formula for the average value that depends on the stabilizer purities \cite{tirrito2024quantifying}.

\section{Bound on rate of linear entanglement}
\label{sec:Rate_entanglement}
Finally, we comment on the use of antiflatness in dynamical scenarios.
The linear entanglement rate, a quantity used to characterize entanglement growth in both many-body dynamics and quantum circuits, can be upper-bounded in terms of two physically meaningful quantities: the speed of quantum evolution and the square root of the Linear Rènyi spread.
This result is the linearized version of the Capacity of Entanglement and Von Neumann entropy in \cite{Shrimali_2022}.Thus, Linear Rènyi spread provides a direct way to control how rapidly linear entanglement can be generated during dynamical evolutions.
\begin{proposition}
    Given an initial state $\ket{\psi(0)} \in \mathcal{H}_A \otimes \mathcal{H}_B$ and the dynamics generated by a non-local Hamiltonian $H_{AB}$, with the time evolved state at later time $\ket{\psi(t)}=e^{- i H_{AB} t} \ket{\psi(0)}$, then the linear entanglement rate is upper bounded by
    \begin{equation}
        \abs{\frac{d \,E_{lin}(\ket{\psi(t)})}{dt}} \leq 4 \sqrt{\mathcal{F}_A(\psi(t))} \sqrt{Var_{\psi} (H_{AB})}
        \label{eq:E_lin_bound}
    \end{equation}
    with 
    $$Var_{\psi} (H_{AB}) = \sqrt{\bra{\psi(0)}H_{AB}^2\ket{\psi(0)}-\bra{\psi(0)}H_{AB}\ket{\psi(0)}^2}$$
    and $E_{lin}(\ket{\psi(t)})=1- tr( \psi^2_A(t))$ for $\psi_A(t)\equiv tr_B(\ket{\psi(t)} \bra{\psi(t)})$.
\end{proposition}

\begin{proof}
    From the definition of linear entanglement $E_{\mathrm{lin}}(t)
=1-\operatorname{Tr}_A\!\left[\psi_A(t)^2\right]$ we have
\begin{equation}
        \begin{aligned}
            \frac{dE_{\mathrm{lin}}(t)}{dt}
&=-2\,\operatorname{Tr}_A\!\left[\psi_A(t)\dot{\psi}_A(t)\right]=2i\,\operatorname{Tr}_A\!\left[
\psi_A(t)\operatorname{Tr}_B\!\left[H_{AB},\psi_{AB}(t)\right]
\right]
\\
&=
2i\,\operatorname{Tr}_{AB}\!\left[
\left(\psi_A(t)\otimes \mathbb{I}_B\right)
\left[H_{AB},\psi_{AB}(t)\right]
\right]
=2i\,\operatorname{Tr}_{AB}\!\left[
\psi_{AB}(t)
\left[\psi_A(t)\otimes \mathbb{I}_B,H_{AB}\right]
\right]
\\
&=
2i\,\bra{\psi(t)}
\left[\psi_A(t)\otimes\mathbb{I}_B,H_{AB}\right]
\ket{\psi(t)} .
    \end{aligned}
\end{equation}
By taking the absolute value and apply the Cauchy--Schwarz (or equivalently, the Heisenberg--Robertson \cite{robertson1929uncertainty}) uncertainty relation for two noncommuting observables
    $A=\psi_A(t)\otimes\mathbb{I}_B$ and $B=H_{AB}$:
    \begin{equation}
        \tfrac{1}{2}\big|\bra{\psi(t)}[A,B]\ket{\psi(t)}\big|
        \le
        \sqrt{\operatorname{Var}_{\psi(t)}(A)}\,\sqrt{\operatorname{Var}_{\psi(t)}(B)}.
    \end{equation}
    Identifying $\operatorname{Var}_{\psi_A(t)}(\psi_A(t))=\mathcal{F}_A(\psi(t))$
    and $\operatorname{Var}_{\psi(t)}(B)=\operatorname{Var}_{\psi(0)}(H_{AB})$ (the latter being time-invariant under unitary evolution), we obtain inequality~\eqref{eq:E_lin_bound}.
\end{proof}

The interpretation of $2 \sqrt{Var_{\psi} (H_{AB})}$ as the speed limit evolution can be done using the Fubini-Study metric on the projective Hilbert space of the composite system as in \cite{Shrimali_2022}.
Since we have an upper bound on this antiflatness measure we can state the following loose upper bound on the entanglement rate
\begin{equation}
    \abs{\frac{d \,E_{lin}(\ket{\psi(t)})}{dt}} < \frac{3 \sqrt{3}}{8} \sqrt{Var_{\psi} (H_{AB})}.
\end{equation}

\section{Conclusions}
\label{Sec : Conclusions}
In this work, we established that standard majorization theory is incompatible with the concept of antiflatness, as it is blind to the geometry of spectral fluctuations. To overcome this limitation, we introduced a partial ordering framework that ranks quantum states according to their Rényi entropy spread. This construction highlights a structural distinction: while conventional entropic quantities quantify the amount of disorder, the Rényi spread captures the \textit{inhomogeneity} of that disorder. As such, it provides a natural language for describing antiflatness, shifting the physical focus from the mere concentration of a state to the variance of its ES.

By reframing these spectral fluctuations through escort distributions and Bregman divergences, we revealed that quantities like the Capacity of Entanglement possess information-geometric significance. While previous works have recognized the formal relationship between the Capacity of Entanglement and the Quantum Fisher Information~\cite{de2019aspects}, our escort-based framework unveils a geometric simplification. Because a state and its escort distributions commute, the quantum Fisher geometry becomes unique. In this setting, the Capacity of Entanglement measures the sensitivity to escort deformations, implying that antiflatness is the spectral susceptibility for probing the shape of the ES. Furthermore, we demonstrated that states maximizing these geometric distortions shatter the notion of a single universal maximum, instead forming a continuous Pareto frontier of non-dominated extremal states. From a practical perspective, our unification of all the presented antiflatness measures reveals a operational advantage: by expressing the quantifiers entirely in terms of integer moments of the density matrix, macroscopic spectral fluctuations become experimentally accessible in near-term quantum devices through efficient randomized measurement protocols, entirely bypassing the exponential overhead of full state tomography.

We define the state convertibility of antiflat$_A$ states by defining Flatness-Preserving Operations (FPOs) as the maximal set of completely positive trace-preserving (CPTP) maps that leave the Rényi spread invariant. This methodological distinction is crucial when contextualizing the open problem of operational sufficiency. As explicitly demonstrated in the resource theory of non-stabilizerness~\cite{heimendahl2021axiomatic}, the set of operations defined axiomatically can be strictly larger than the set of operations defined operationally. Mathematical permission does not automatically guarantee physical realizability with a restricted gate set. A parallel dichotomy is highly probable in the resource theory of antiflatness. While $\prec_{AF}$ successfully dictates the ultimate mathematical limits---revealing that FPOs operate beyond standard majorization---discovering the exact \textit{operational} equivalents remains an open challenge. Identifying the specific physical quantum circuits required to execute these compressions deterministically will be a compelling direction for future research, setting the stage for operational discoveries regarding the interplay of entanglement and magic (possibly through the lens of non-local magic), and the thermodynamics of spectral shape.

The statistical analysis of random ensembles further clarifies the role of antiflatness as a fluctuation diagnostic. For Haar and Bures--Hall ensembles, the average Linear Rényi Spread vanishes in the large-dimensional limit, showing that typical reduced states become not only highly entangled but also increasingly flat in their spectral structure.

Finally, the dynamical bound derived for the linear entanglement rate shows that antiflatness is not only a static property of spectra, but it controls how efficiently entanglement can be generated under non-local dynamics, together with the energy variance of the driving Hamiltonian.

Overall, antiflatness provides a complementary way of organizing quantum correlations. By focusing on the internal fluctuations of the entanglement spectrum, it exposes a layer of quantum correlations that is naturally connected to magic entanglement, random-state typicality and information geometry.

\section{Acknowledgements}
B.J. would like to thank L. Leone, S.F.E. Oliviero and P. Faist for fruitful discussions on the matter. A.H. acknowledge support from  PNRR MUR project PE0000023-NQSTI and the PNRR MUR project CN 00000013-ICSC.
The authors acknowledge the novella \textit{Flatland: A Romance of Many Dimensions} by Edwin Abbott Abbott~\cite{abbott2023flatland}, which inspired the title of this paper.

\bibliographystyle{quantum}
\bibliography{reference}

\appendix

\section{Failing of monotonicity under majorization}
\label{App : Failing}
A state $\rho \in \mathcal{D}(\mathcal{H}_A)$ is defined as flat if it is maximally uniform on its support. If we focus on the $\mathcal{H}_A$ Hilbert space, since the condition is given on states living in that space, we can deduce the geometric structure of the sets of states of interest.
The defining relation implies that, for $r := \operatorname{rank}(\rho)$,
\begin{equation}
\rho = \frac{1}{r} \Pi, \quad \Pi^2 = \Pi = \Pi^\dagger, \quad \operatorname{Tr}(\Pi) = r.
\end{equation}
Thus, flat$_A$ states of fixed rank $r$ correspond bijectively to orthogonal projectors $\Pi$ of rank $r$.

Such projectors are parameterized by the complex Grassmannian manifold $\mathrm{Gr}(r,d)$, which is the space of all $r$-dimensional subspaces of $\mathbb{C}^d$. This manifold is smooth and has the following well-known real dimension~\cite{zyczkowski2006GeometryQuantumStates}:
\begin{equation}
\dim_{\mathbb{R}} \mathrm{Gr}(r,d) = 2r(d - r).
\end{equation}
Accordingly, we denote by $\mathrm{FLAT}_r$ the subset of flat$_A$ states of rank $r$ and have the isomorphism of smooth manifolds
\begin{equation}
\mathrm{FLAT}_r \cong \mathrm{Gr}(r,d), \quad \text{with} \quad \dim_{\mathbb{R}}(\mathrm{FLAT}_r) = 2r(d - r).
\end{equation}

Since the total set of flat$_A$ states decomposes into a disjoint union over all possible ranks, $\mathrm{FLAT} = \bigsqcup_{r=1}^d \mathrm{FLAT}_r$ , 
its maximal dimension is given by $\dim_{\mathbb{R}}(\mathrm{FLAT}) = \max_{1 \leq r \leq d} 2r(d - r)$.
The maximum is attained at $r = \left\lfloor \frac{d}{2} \right\rfloor$ or $r = \left\lceil \frac{d}{2} \right\rceil$, yielding
\begin{equation}
\dim_{\mathbb{R}}(\mathrm{FLAT}) =
\begin{cases}
\displaystyle \frac{d^2}{2}, & \text{if } d \text{ is even}, \\[1.5ex]
\displaystyle \frac{d^2 - 1}{2}, & \text{if } d \text{ is odd}.
\end{cases}
\end{equation}
In particular, for pure states ($r=1$), which are a subset of flat$_A$ states, the dimension reduces to $\dim_{\mathbb{R}}(\mathrm{FLAT}_1) = 2(d - 1)$,
matching the well-known dimension of the complex projective space of pure states~\cite{zyczkowski2006GeometryQuantumStates}.
The use of the maximal dimension among the strata $\mathrm{FLAT}_r$ is motivated by the stratified nature of the set $\mathrm{FLAT}$ itself. Since $\mathrm{FLAT}$ is a disjoint union of manifolds with varying dimensions, it is not a smooth manifold globally but rather a union of smooth manifolds of different sizes. The maximal dimension corresponds to the largest, most generic component of $\mathrm{FLAT}$, capturing the effective degrees of freedom and the complexity of the set of flat$_A$ states. From both geometric and practical perspectives, the top-dimensional strata dominate. Hence, the maximal dimension provides a meaningful characterisation of the size and structure of the free states set.
Additionally, note that neither $ \mathrm{FLAT}$ nor $ \mathrm{FLAT}_r$ are convex sets, it is sometimes an odd feature in resource theory, as is the case for stabilizer states. Hence, neither $ \mathrm{FLAT}_A$ is convex.

It is worth noticing that all the above measure are neither Shur-convex nor Shur-concave. In quantum information, Schur-convex and Schur-concave functions are used to quantify resources (like entanglement or coherence), where the function's behavior with respect to majorization reflects the relative "spread" or "concentration" of the resource.
As far as the Capacity of entanglement, this was discussed in \cite{boes2022variance}, whereas for the logarithmic anti-flatness we can show a concrete example (the same holds for the Linear Rényi spread).
Suppose $\rho \prec \sigma$  and that their spectrum is $\lambda^\rho=\left(\frac{9}{10},\frac{1}{10}\right)$ and $\lambda^\sigma=\left(\frac{8}{10},\frac{2}{10}\right)$ then $\Delta^{23}(\rho)\simeq 0.0593 < \Delta^{23}(\sigma)\simeq 0.0847  $.
Now one can consider a different spectrum with the same majorization criterion  $\lambda^\rho=\left(\frac{9}{10},\frac{1}{10}\right)$ and $\lambda^\sigma=\left(\frac{6}{10},\frac{4}{10}\right)$ and it happens that $\Delta^{23}(\rho)\simeq 0.0593 > \Delta^{23}(\sigma)\simeq 0.0252  $.

In general, if a faithful function in a resource theory is neither Shur-concave nor Schur-concave, it means that the quantification of the resource does not follow conventional expectations about mixing, combining, or spreading resources. 

Finding free operation for the resource theory of flatness is in general not an easy task except for trivial example. In ~\cref{App: Hardeness of finding FPOs}, we highlight some problems.
A trivial set of free operations is given by $U_A \otimes U_B$ with $U_A,U_B$ local unitaries.

In the context of quantum resource theories, any resource measure that is compatible with the majorization ordering must mathematically be a Schur-convex function. However, the usual concept of ordering by majorization fails in our context, because it inherently orders states by purity rather than by flatness. 

\begin{theorem}[No-Go of Majorization Ordering for Antiflatness] \label{thm:No-go theorem}
None of the antiflatness measures, $\mathcal{F}_A(\rho), \; \log(\Lambda_{\rho}), \; \mathcal{V}_A(\rho)$ is neither strictly Schur-convex nor Schur-concave. Consequently, the majorization relation $\rho \prec \sigma$ is not sufficient to establish an ordering of antiflat$_A$ states, i.e., $\rho \prec \sigma \nRightarrow \mathcal{V}_A(\rho) \lesseqqgtr \mathcal{V}_A(\sigma)$.
\end{theorem}

\begin{proof}
Let us consider the set of free states. Both the pure state, defined by the spectrum $\rho_{pure} = (1, 0, \dots, 0)$, and the completely mixed state, defined by $\rho_{mixed} = (\frac{1}{d}, \frac{1}{d}, \dots, \frac{1}{d})$, belong to the set of free states. Given the faithfulness of the three antiflatness measures, we know that they must evaluate to exactly $0$ for both of these free states. Because these functions are strictly positive for intermediate states but return to $0$ at both the maximum and the minimum of the majorization ordering , by Rolle's Theorem, they cannot possibly be monotone along the majorization chain. Therefore, states  $\notin \rm FLAT$ can not be ordered via majorization since it orders according to the purity of the state and not according its antiflatness. 
\end{proof}

Notice that the argument presented in the proof can be extended to reveal an even more profound structural incompatibility between majorization and antiflatness by considering intermediate flat$_A$ states. A generic flat state of rank $k$ (where $1 \le k \le d$) is proportional to a rank-$k$ projector, possessing the spectrum $\rho_k = (\frac{1}{k}, \dots, \frac{1}{k}, 0, \dots, 0)$. Under the majorization preorder, these flat$_A$ states of varying ranks form a strict, totally ordered chain defined entirely by their purity
\begin{equation}
   \rho_d \prec \rho_{d-1} \prec \dots \prec \rho_k \prec \dots \prec \rho_2 \prec \rho_1 \; \ ,
\end{equation}
where $\rho_d$ is the completely mixed state and $\rho_1$ is the pure state. Because all such states belong to the set of free states ($\rho_k \in \text{FLAT}$), our antiflatness measures evaluate to exactly zero for every single one of them $\mathcal{V}(\rho_k) = \mathcal{F}_A(\rho_k) = \log(\Lambda_3(\rho_k)) = 0 \quad \forall k \in \{1, \dots, d\}$. This implies that along the majorization hierarchy, the antiflatness functionals do not merely start at zero and end at zero; they periodically drop back to zero at multiple discrete ``nodes" (every time the state reaches a perfectly flat configuration of any rank). Between any two such nodes, $\rho_k$ and $\rho_{k-1}$, there exist non-flat intermediate states with strictly positive antiflatness. Consequently, antiflatness measures exhibit multiple local minima (all evaluating to zero) along continuous majorization paths. A functional that repeatedly oscillates between zero and positive values along a strict majorization chain completely shatters any possibility of Schur-convexity or Schur-concavity. This thoroughly solidifies the conclusion that majorization (which is essentially a hierarchy of purity) is fundamentally blind to the geometry of antiflatness.

\section{Discrimination via Antiflatness}
\label{App : Discrimination via Antiflatness}

In the main text, we introduced antiflatness as a way to probe the structure of the entanglement spectrum beyond scalar entropic quantities. The purpose of this appendix is to make this distinction explicit through two complementary observations. First, we show that the same reduced spectrum may admit purifications with different stabilizer properties. Second, we construct two states with the same von Neumann entropy but different R\'enyi-spread structure, showing why the ordering $\prec_{AF}$ contains information that is invisible to standard entanglement entropy.

Let us start from the spectra displayed in Fig.~\ref{fig: REcurves}. 
Denote by $\vec{\lambda}^{\,X} = (\lambda_0^X,\lambda_1^X,\lambda_2^X,\lambda_3^X), X\in\{F,A,B,C\}$, one of the four reduced spectra, where $F$ labels the flat spectrum. A canonical Schmidt purification realizing this spectrum is $|\Psi_X\rangle = \sum_{j=0}^{3}  \sqrt{\lambda_j^X}\, |j\rangle_A |j\rangle_B$. By construction, $\Tr_B |\Psi_X\rangle\langle\Psi_X|$ has spectrum $\vec{\lambda}^{\,X}$. Therefore these purifications provide explicit pure-state representatives of the R\'enyi curves used in the main text.

For the flat spectrum $\vec{\lambda}^{\,F} =\left(\frac14,\frac14,\frac14,\frac14\right)$, the canonical purification is $|\Psi_F\rangle = \frac{1}{2} \sum_{j=0}^{3} |j\rangle_A |j\rangle_B$. If $A=A_1A_2$ and $B=B_1B_2$, this state factorizes into two Bell pairs, $|\Psi_F\rangle = |\Phi^+\rangle_{A_1B_1} \otimes |\Phi^+\rangle_{A_2B_2}$, where $|\Phi^+\rangle = \frac{|00\rangle+|11\rangle}{\sqrt{2}}$. Thus, the flat spectrum can be realized by a maximally entangled stabilizer state. Although this state has maximal entanglement entropy across the bipartition $A|B$, its entanglement spectrum is perfectly flat. 
Consequently,
\begin{equation}
    \Delta_{\alpha\beta}(\rho_A)=0,
    \qquad
    \forall \alpha<\beta,
\end{equation}
and all antiflatness measures vanish. This illustrates the distinction between the amount of entanglement and the internal structure of the entanglement spectrum.

However, the same flat reduced spectrum can also arise from purifications with different stabilizer structure. 
For example, consider a locally dressed state $|\Psi_F^{T}\rangle = (U_T\otimes I_B)|\Psi_F\rangle$, where $U_T$ is a local non-Clifford unitary acting on subsystem $A$, such as
\begin{equation}
    U_T=T\otimes I,
    \qquad
    T=\mathrm{diag}(1,e^{i\pi/4}).
\end{equation}
The reduced state transforms as $\rho_A \longmapsto U_T\rho_A U_T^\dagger = \frac{I_4}{4}$. Hence the reduced spectrum remains exactly flat, and all antiflatness measures still vanish. 
Nevertheless, the global pure state need not remain a stabilizer state, since the purification has been locally dressed by a non-Clifford operation. This example shows that antiflatness does not quantify arbitrary local magic. Rather, it captures the component of quantum complexity encoded in the entanglement spectrum itself.

The non-flat spectra $A$, $B$, and $C$ instead define purifications whose entanglement spectra have nontrivial internal fluctuations. In particular, the pair $(A,C)$ provides a useful example of the partial nature of the antiflat ordering. 
Their R\'enyi entropy curves cross, and their R\'enyi spreads reverse ordering on different intervals. 
Therefore, although both states are non-flat, neither spectrum globally dominates the other under $\prec_{AF}$. 
This is precisely the sense in which antiflatness is not captured by a single scalar steepness criterion, but by the full family of R\'enyi-spread inequalities.

\begin{example}[Iso-entropic states and antiflatness discrimination]
\label{ex:iso_entropic}

We now give a simple example showing that two states may have the same von Neumann entropy while having different antiflatness. 
Consider two reduced density matrices on a three-dimensional subsystem $A$, with spectra
\begin{align}
    \lambda(\rho_A)
    =
    \left(
    \frac12,\frac12,0
    \right), \quad
    \lambda(\sigma_A)
    \approx
    \left(
    0.773,0.1135,0.1135
    \right).
\end{align}
Using logarithms in base two, their von Neumann entropies are
\begin{equation}
    S_1(\rho_A)
    =
    -\frac12\log_2\frac12
    -
    \frac12\log_2\frac12
    =
    1,
\end{equation}
and
\begin{equation}
    S_1(\sigma_A)
    =
    -0.773\log_2(0.773)
    -
    2(0.1135)\log_2(0.1135)
    \approx
    1 .
\end{equation}
Thus, from the perspective of the von Neumann entropy, the two states contain the same amount of bipartite entanglement.

Their R\'enyi-spread structure, however, is completely different. 
The state $\rho_A$ is flat on its support: it is proportional to a rank-two projector. 
Therefore all its nonzero eigenvalues are equal, all its R\'enyi entropies coincide, and
\begin{equation}
    \Delta_{\alpha\beta}(\rho_A)=0,
    \qquad
    \forall \alpha<\beta .
\end{equation}
In other words, $\rho\in\mathrm{FLAT}_A$ and it is a free state for the antiflatness resource theory.

By contrast, $\sigma_A$ is full-rank and strongly non-uniform. 
Although its von Neumann entropy equals that of $\rho_A$, its R\'enyi entropies depend nontrivially on $\alpha$. 
Consequently, there exist pairs $\alpha<\beta$ for which
\begin{equation}
    \Delta_{\alpha\beta}(\sigma_A)>0.
\end{equation}
Thus $\sigma_A$ is antiflat, while $\rho_A$ is flat on its support.

This example shows why the R\'enyi-spread ordering introduced in the main text is not redundant with standard entanglement entropy. The von Neumann entropy measures the mean surprisal, $S_1(\rho_A) = \mathbb{E}_{\rho_A}[-\log \rho_A]$, whereas antiflatness probes the variation of the R\'enyi curve, i.e. the fluctuations and scale-dependence of the entanglement spectrum. 
Two states may therefore have the same average entanglement while belonging to different antiflatness classes. 
This is the operational role of $\prec_{AF}$: it distinguishes spectra that standard scalar entanglement measures identify.
\end{example}

\section{Antiflatness allows operations forbidden by standard majorization}
\label{App : Maj freezed}
We present an explicit counterexample on a four-dimensional Hilbert space $\mathcal{H}_A$ ($d_A=4$) that illustrates the structure of the iso-purity manifold. We construct an initial state $\sigma$ and a target state $\rho$ with identical purity, such that $\rho$ is strictly antiflat-majorized by $\sigma$, while standard majorization forbids the transformation.
Define the ordered spectra
\begin{equation}
\mathrm{spec}(\sigma_A)^\downarrow = \left( \frac{22}{40}, \frac{9+\sqrt{77}}{40}, \frac{9-\sqrt{77}}{40}, 0 \right) \approx ( 0.5500, , 0.4444, , 0.0056, , 0 ),,
\end{equation}
\begin{equation}
\mathrm{spec}(\rho_A)^\downarrow = \left( \frac{1}{2}, \frac{1}{2}, 0, 0 \right),.
\end{equation}
Both states lie on the manifold $\mathcal{M}_{1/2}$ with purity $S_2(\rho_A)=S_2(\sigma_A)=\log(2)$.
A necessary condition for a flatness-preserving operation is that the Rényi spread does not increase, i.e. $\Delta_{\alpha\beta}(\rho_A)\le \Delta_{\alpha\beta}(\sigma_A)$ for all parameters. Since $\rho$ is flat, this inequality is strict, implying $\rho \prec_{AF} \sigma$. Consequently, an FPO can implement the transformation $\sigma \to \rho$.

\subsection{Proof of Standard Majorization Failure}
We now examine standard majorization. Nielsen’s criterion requires that the cumulative sums of the ordered eigenvalues of $\rho$ do not exceed those of $\sigma$. At $k=1$, the condition holds: $\rho_A^{(1)}= 0.5000 \le \sigma_A^{(1)} = 0.5500$. At $k=2$, however,
\begin{equation}
    \rho_A^{(1)} + \rho_A^{(2)} = 1.0000 
    \qquad \text{while} \qquad 
    \sigma_A^{(1)} + \sigma_A^{(2)} \approx 0.9944 \,,
\end{equation}
which violates the required inequality. The Lorenz curves therefore intersect, and neither $\rho \prec \sigma$ nor $\sigma \prec \rho$ holds, then the transformation is thus forbidden under LOCC or unital maps.

This example clarifies the operational distinction. The map $\sigma \to \rho$ simultaneously reduces the leading eigenvalue and removes the tail, redistributing the weight into the intermediate component to produce a perfectly flat spectrum. Standard majorization doesn't allow such a process: decreasing the largest eigenvalue necessarily spreads probability mass outward, typically increasing rank and raising the Shannon entropy. In contrast, flatness-preserving operations act as spectral compressors, concentrating weight inward while maintaining fixed purity.

On the iso-purity manifold, both states satisfy $S_2(\rho)=S_2(\sigma)$, so their Rényi entropy curves coincide at $\alpha=2$. Away from this point, their behavior is governed by the difference in spectral flatness. The state $\rho$, having a more uniform spectrum, produces a Rényi entropy curve with weaker dependence on $\alpha$, whereas the more uneven spectrum of $\sigma$ leads to a stronger variation.
This difference implies a systematic ordering on either side of $\alpha=2$. For $\alpha<2$, where Rényi entropies are more sensitive to the contribution of smaller eigenvalues, the broader support of $\sigma$ yields $S_\alpha(\sigma) > S_\alpha(\rho)$. For $\alpha>2$, where the dominant eigenvalues control the entropy, the reduced peak of $\rho$ gives $S_\alpha(\rho) > S_\alpha(\sigma)$. Hence the curves cross exactly at $\alpha=2$.
This crossing should not be interpreted as incomparability. The relevant quantities for antiflat-majorization are the Rényi spreads, not the absolute entropy values. Despite the intersection, the spreads remain ordered as $\Delta_{\alpha\beta}(\rho) \le \Delta_{\alpha\beta}(\sigma)$ for all parameters. The crossing is therefore a direct consequence of the shared purity constraint, which enforces a single intersection point while preserving a strict antiflat ordering.

\section{Necessary condition for pure-state convertibility under axiomatic FPOs}
\label{App: prop necessary condition}
\begin{proposition}[Necessary condition for pure-state convertibility under FPOs]
Let $|\psi\rangle,|\phi\rangle \in \mathcal{H}_A\otimes\mathcal{H}_B$ be bipartite pure states, and let
$\rho^\psi,\rho^\phi$ denote the associated density matrices.
Assume that FPOs are defined axiomatically as the set of CPTP maps $\Phi$ such that
\begin{equation}
\Delta_{\alpha\beta}((\Phi(\rho))_A)\le \Delta_{\alpha\beta}(\rho_A)
\qquad \forall \rho,\ \forall \alpha<\beta .
\end{equation}
Then
\begin{equation}
|\psi\rangle \xrightarrow{\mathrm{FPO}} |\phi\rangle
\quad \Longrightarrow \quad
\rho^\phi \prec_{AF} \rho^\psi .
\end{equation}
\end{proposition}
\begin{proof}
Assume that there exists an FPO channel $\Phi$ such that $\Phi\!\left(|\psi\rangle\langle\psi|\right)=|\phi\rangle\langle\phi| .$
By the defining axiom of FPOs, for every pair $\alpha<\beta$ one has
\begin{equation}
\Delta_{\alpha\beta}\!\left((\Phi(|\psi\rangle\langle\psi|))_A\right) = \Delta_{\alpha\beta}(\rho^\phi_A)
\le
\Delta_{\alpha\beta}\!\left((|\psi\rangle\langle\psi|)_A\right)= \Delta_{\alpha\beta}(\rho^\psi_A).
\end{equation}
By Definition~\ref{def:Antiflatness Ordering}, this is exactly $\lambda_\phi \prec_{AF} \lambda_\psi .$
\end{proof}

\section{Hardness of finding FPOs: Failure of standard operations}
\label{App: Hardeness of finding FPOs}

In this section, we explicitly prove the non-triviality of finding Flatness-Preserving Operations (FPOs) by showing that several common classes of quantum dynamics fail to be resource-free in the context of the antiflatness resource theory. Specifically, we will first show that Local Operations and Classical Communication (LOCC) can strictly generate antiflatness. Then, we will prove that local projective measurements (dephasing) and POVMs also fail to preserve the set of free states.

\subsection{Failure of LOCC}

Let us consider standard state convertibility for bipartite pure states under LOCC. By Nielsen's theorem, a global pure state $\ket{\Psi_\rho}$ can be deterministically transformed into $\ket{\Psi_\sigma}$ via LOCC if and only if the spectrum of the initial reduced state is majorized by the spectrum of the target reduced state, i.e., $\text{spec}(\rho_A) \prec \text{spec}(\sigma_A)$. 

Consider an initial state $\ket{\Psi_\rho}$ whose reduced density matrix $\rho_A$ possesses a perfectly flat spectrum, $\text{spec}(\rho_A) = (1/3, 1/3, 1/3)$, implying that its Rényi spread is exactly zero ($\Delta_{\alpha\beta}(\rho_A) = 0$). We wish to transform this into a target state $\ket{\Psi_\sigma}$ with a fluctuating reduced spectrum, for instance $\text{spec}(\sigma_A) = (2/3, 1/6, 1/6)$, which exhibits a strictly positive spread ($\Delta_{\alpha\beta}(\sigma_A) > 0$). Because the cumulative sums of the ordered eigenvalues satisfy the standard majorization condition ($1/3 \le 2/3$, $2/3 \le 5/6$, and $1=1$), we have $\rho_A \prec \sigma_A$. Consequently, LOCC protocols are perfectly capable of executing the global transformation $\ket{\Psi_\rho} \xrightarrow{\text{LOCC}} \ket{\Psi_\sigma}$. However, in the context of our resource theory, the initial reduced state is resource-free, while the final reduced state possesses macroscopic spectral fluctuations. The LOCC protocol has strictly increased the antiflatness of the reduced spectrum. Therefore, LOCC acts as a resource-generating operation for antiflatness, proving that standard LOCC protocols cannot be classified as FPOs.

\subsection{Failure of Local Projective Measurements and Dephasing}

Similarly, one might expect that local dephasing channels would smear quantum information and reduce resource content. We now show that projective local measurements (and by extension POVMs) are not resource-free operations, as they can cause a state to depart from the set of free flat$_A$ states. After a local projective measurement on subsystem $A$, the state evolves according to the dephasing superoperator
\begin{equation}
    D_{\mathcal{B}_A}(\rho) = \sum_{k=1}^{d_A} \left( \omega_k \otimes \id_B \right) \rho \left( \omega_k \otimes \id_B \right) \,,
\end{equation}
where $\omega_k = \ket{k}\bra{k}$ and $\mathcal{B}_A = \{ \ket{k} \}_{k=1}^{d_A}$ forms a complete orthonormal basis on $\mathcal{H}_A$. 

Suppose $\rho \in \mathcal{D}(\mathcal{H}_A \otimes \mathcal{H}_{B})$ is a pure state, $\rho = \ket{\rho}\bra{\rho}$, which can be written in its Schmidt decomposition as
\begin{equation}
    \ket{\rho} = \sum_{i=1}^{d_A} \sqrt{\lambda_i} \ket{\phi_i \chi_i} \in \mathcal{H}_A \otimes \mathcal{H}_{B} \,,
\end{equation}
where $\lambda_i \geq 0$, $\sum_i \lambda_i = 1$, and we assume, without loss of generality, that $d_A \leq d_B$.
After applying the dephasing superoperator, the state assumes the form:
\begin{equation}
    D_{\mathcal{B}_A}(\ket{\rho}\bra{\rho}) = \sum_k \sum_{i,j} \sqrt{\lambda_i \lambda_j} \,\omega_k \ket{\phi_i} \bra{\phi_j} \omega_k \otimes \ket{\chi_i} \bra{\chi_j} \,.
\end{equation}
After partial tracing over subsystem $B$, one gets the new reduced density matrix:
\begin{equation}
    \rho_A' \equiv \Tr_B( D_{\mathcal{B}_A}(\ket{\rho}\bra{\rho})) = \sum_{i,k} \lambda_i \,\omega_k \ket{\phi_i} \bra{\phi_i} \omega_k = \sum_k \omega_k \rho_A \omega_k = \sum_k \Tr(\rho_A \omega_k)\omega_k \,,
\end{equation}
since $\rho_A = \sum_i \lambda_i \ket{\phi_i} \bra{\phi_i}$.
To evaluate whether this resulting state remains flat, we must compute $\rho_A'^2$ and verify if it is proportional to a projector
\begin{equation}
    \rho_A'^2 = \sum_{k,l} \omega_k \rho_A \omega_k \omega_l \rho_A \omega_l = \sum_k \omega_k \rho_A \omega_k \rho_A \omega_k = \sum_k \bra{k}\rho_A \ket{k}^2 \ket{k}\bra{k} = \sum_k \Tr^2(\rho_A \omega_k) \omega_k \,.
\end{equation}
This implies that $\rho_A'$ is flat if and only if $\Tr(\rho_A \omega_k) = \frac{1}{\text{rank}(\rho_A')}$ for all $k$ in its support, which is generally not true.

Let us provide a concrete example demonstrating that a free state can be mapped to a resourceful state. Given an initially flat reduced state $\rho_A = \ket{00}\bra{00}$ (a rank-1 pure state in a $d_A=4$ dimensional space), we choose the following complete orthonormal basis:
$$ \mathcal{B}_A = \left\{ \frac{\ket{00}+\sqrt{2} \ket{01}}{\sqrt{3}}, \, \frac{(2-\sqrt{2})\ket{00}+(1-\sqrt{2}) \ket{01}}{\sqrt{9-6 \sqrt{2}}}, \, \frac{\ket{10}+ \sqrt{2}\ket{11}}{\sqrt{3}}, \, \frac{\sqrt{2}\ket{10}-\ket{11}}{\sqrt{3}} \right\} \equiv \{ \ket{k_i} \}_{i=1}^4 \,. $$
Computing the dephased reduced state yields
\begin{equation}
    \rho_A' = \frac{1}{3} \ket{k_1}\bra{k_1} + \frac{2}{3}\ket{k_2}\bra{k_2} \,,
\end{equation}
which is decidedly not flat, since $\rho_A'^2 \neq \frac{\rho_A'}{\text{rank}(\rho_A')}$. Consequently, projective local measurements are not resource-free operations. Because total projective measurements are a subset of POVMs, POVMs in general also fail to be FPOs.

It is worth noting that a local measurement on the \textit{complementary} subsystem behaves differently. If we apply $D_{\mathcal{B}_B}(\rho) = \sum_{k=1}^{d_B} ( \id_A \otimes \theta_k ) \rho ( \id_A \otimes \theta_k )$ with $\mathcal{B}_B = \{ \theta_k \}_{k=1}^{d_B}$ being an orthonormal basis on $\mathcal{H}_B$, tracing out $B$ yields
\begin{equation}
    \Tr_B( D_{\mathcal{B}_B}(\ket{\rho}\bra{\rho})) = \sum_{i,j,k} \sqrt{\lambda_i \lambda_j} \, \bra{\chi_i}\theta_k\ket{\chi_j}\,\ket{\phi_i} \bra{\phi_j} = \sum_i \lambda_i \ket{\phi_i} \bra{\phi_i} = \rho_A \,.
\end{equation}
The reduced density matrix is perfectly unaffected; if the state was initially flat, it remains trivially flat.

\subsection{Failure within Stabilizer Bases}
Finally, the problem persists even when restricting the operations to stabilizer formalisms. Suppose we have a reduced state over four qubits ($d_A=16$) defined as
\begin{equation}
    \rho_A = \frac{1}{10} \sum_{i=1}^{10} \ket{i}\bra{i} \,,
\end{equation}
which is a flat rank-10 projector. Now, define the dephasing basis as a stabilizer basis, $\mathcal{B}_A = \{ C\ket{i}\bra{i} C^\dagger \}_{i=1}^{d_A}$, choosing the local Clifford operator $C = H \otimes \id_2^{\otimes 3}$.
Computing the dephased state yields:
\begin{equation}
\rho_A' = \text{diag}\left( \frac{1}{10},\frac{1}{10},\frac{1}{20},\dots,\frac{1}{20},\frac{1}{10},\frac{1}{10},\frac{1}{20},\dots,\frac{1}{20} \right) = \frac{\Pi_1}{10} + \frac{\Pi_2}{20} \,,
\end{equation}
where $\Pi_1$ is the projector onto the subspace spanned by $\{\ket{0000},\ket{0001},\ket{1000},\ket{1001}\}$ and $\Pi_2$ acts on its complement. This resulting state is clearly not flat. 

Hence, we conclude that not even dephasing within a stabilizer basis is guaranteed to be a flatness-preserving operation.

\section{Differential Geometry of the Iso-Purity Manifold}
\label{App : Iso-Purity Manifold}
In this section, we formally prove that the iso-purity set of quantum states is not merely a set, but possesses the rigorous structure of a differentiable manifold. We employ the Regular Level Set Theorem to separate the algebraic purity constraint from the positivity constraint, ultimately characterizing the state space as a stratified manifold.

Let $\mathcal{H} \cong \mathbb{C}^d$ be a finite-dimensional Hilbert space. The iso-purity set $\mathcal{M}_P$ for a given purity $P \in [1/d, 1]$ is defined as the set of all density matrices sharing that exact purity:
\begin{equation}
    \mathcal{M}_P = \left\{ \rho \in \text{Herm}(\mathcal{H}) \mid \rho \ge 0, \; \Tr(\rho) = 1, \; \Tr(\rho^2) = P \right\} \ .
\end{equation}
To prove that $\mathcal{M}_P$ is a manifold, we must separate the equality constraints from the inequality (positivity) constraint.

Let $\text{Herm}(\mathcal{H})$ be the real vector space of Hermitian matrices acting on $\mathcal{H}$, which has a real dimension of $d^2$. We endow this space with the standard Hilbert-Schmidt inner product $\langle X, Y \rangle = \Tr(XY)$.
The set of all trace-one Hermitian matrices forms an affine subspace $\mathcal{A}$:
\begin{equation}
    \mathcal{A} = \{ X \in \text{Herm}(\mathcal{H}) \mid \Tr(X) = 1 \} \ .
\end{equation}
Because it is defined by a single linear constraint, $\mathcal{A}$ is a smooth, flat manifold of dimension $d^2 - 1$. The tangent space to $\mathcal{A}$ at any point $X$ is simply the vector space of traceless Hermitian matrices:
\begin{equation}
    T_X \mathcal{A} = \{ Y \in \text{Herm}(\mathcal{H}) \mid \Tr(Y) = 0 \} \ .
\end{equation}

We define the purity function $f: \mathcal{A} \to \mathbb{R}$ as $ f(X) = \Tr(X^2)$. In order to apply the regular level set theorem (Preimage Theorem), we must compute the Fr\'echet derivative of $f$ at a point $X \in \mathcal{A}$ in the direction of an arbitrary tangent vector $Y \in T_X \mathcal{A}$:
\begin{equation}
    df|_X(Y) = \left. \frac{d}{dt} f(X + tY) \right|_{t=0}= \left. \frac{d}{dt} \Tr((X + tY)^2) \right|_{t=0} = \left. \frac{d}{dt} \big( \Tr(X^2) + 2t\Tr(XY) + t^2\Tr(Y^2) \big) \right|_{t=0}  = 2\Tr(XY) \ .
\end{equation}

The Regular Level Set Theorem dictates that the preimage $f^{-1}(P)$ is a smooth embedded submanifold if the differential $df|_X$ is surjective (non-zero) for \textit{every} $X \in f^{-1}(P)$. 

We must identify the critical points where $df|_X$ fails to be surjective, which occurs when $df|_X(Y) = 0$ for all $Y \in T_X \mathcal{A}$:
\begin{equation}
    2\Tr(XY) = 0 \quad \forall \, Y \text{ such that } \Tr(Y) = 0 \ .
\end{equation}

By the Riesz representation theorem for the Hilbert-Schmidt inner product, if a matrix $X$ is orthogonal to all traceless matrices, $X$ must be strictly proportional to the identity matrix, i.e. $X = c \mathbb{I}$.
Because $X$ resides in the affine space $\mathcal{A}$, it must satisfy $\Tr(X) = 1$. This fixes $c = 1/d$, yielding the maximally mixed state, i.e. $X = \frac{\mathbb{I}}{d} $.

The purity evaluated at this unique critical point is $1/d$.
Therefore, for any purity strictly greater than that of the maximally mixed state, i.e., $P \in (1/d, 1]$, there are no critical points on the level set. The differential $df|_X$ is everywhere full-rank. Consequently, the set 
\begin{equation}
    \tilde{\mathcal{M}}_P = \{ X \in \mathcal{A} \mid \Tr(X^2) = P \}
\end{equation}
is a smooth, embedded submanifold of $\mathcal{A}$ of dimension $(d^2 - 1) - 1 = d^2 - 2$. Geometrically, this represents a hypersphere $S^{d^2-2}$.
For any $P \in (1/d, 1)$, $\mathcal{M}_P$ constitutes a compact, stratified submanifold of dimension $d^2 - 2$. Its top-dimensional stratum comprises a smooth, open Riemannian manifold of full-rank states, while its boundaries consist of smooth manifolds of strictly lower dimensions corresponding to states of reduced rank.

The true physical iso-purity manifold requires the density matrices to be positive semi-definite ($\rho \ge 0$):
\begin{equation}
    \mathcal{M}_P = \tilde{\mathcal{M}}_P \cap \{ X \in \text{Herm}(\mathcal{H}) \mid X \ge 0 \} \ .
\end{equation}
The set of positive semi-definite matrices forms a closed convex cone. To classify the geometric structure of $\mathcal{M}_P$ rigorously, we consider its interior and boundary.

The interior of the density matrix space consists of strictly positive, full-rank states ($\rho > 0$). The intersection of a smooth manifold $\tilde{\mathcal{M}}_P$ with the open set $\{\rho > 0\}$ is, by definition, an open submanifold. Thus, the full-rank portion of $\mathcal{M}_P$ forms a perfectly smooth Riemannian manifold of dimension $d^2 - 2$.

When a state has one or more eigenvalues exactly equal zero, it lies on the boundary of the positive cone. The intersection of a smooth manifold (the hypersphere $\tilde{\mathcal{M}}_P$) with the faces of a convex cone produces a stratified manifold (or a manifold with corners). Each stratum of this manifold corresponds to states of a specific rank $r < d$. For example, the extreme boundary occurring at $P=1$ (pure states) is the stratum of rank-1 states, which natively forms the well-known smooth manifold of the Complex Projective Space, $\mathbb{CP}^{d-1}$.

\section{One-parameter reformulation of antiflatness on iso-purity manifolds}
\label{app:one_parameter_af_proofs}

In this Appendix we collect the proofs of the structural results used in Sec.~\ref{Sec : Ordering} to analyze antiflat convertibility on iso-purity manifolds. The purpose of these results is twofold. First, they show that the definition of antiflat majorization, although formulated as an infinite two-parameter family of Rényi-spread inequalities, can be reformulated exactly as a one-parameter monotonicity condition. Second, once this reformulation is restricted to the iso-purity manifold, it yields a sharp sign structure around the anchoring point $\alpha=2$ and leads to a rigidity statement: any non-trivial antiflatly ordered pair on a fixed iso-purity manifold must differ in rank. These results clarify why antiflat dynamics on $\mathcal{M}_P$ are fundamentally distinct from standard majorization flows.

\subsection{Reduction of the two-parameter constraints to a one-parameter monotonicity test}

\begin{proof}[Proof of Proposition~\ref{prop:one_parameter_AF}]
By definition of antiflat majorization, $\rho \prec_{AF} \sigma \iff \Delta_{\alpha\beta}(\rho_A)\le \Delta_{\alpha\beta}(\sigma_A) \qquad \forall \alpha<\beta$. Using $\Delta_{\alpha\beta}(\tau_A)=S_\alpha(\tau_A)-S_\beta(\tau_A)$,
this is equivalent to
\begin{equation}
S_\alpha(\rho_A)-S_\beta(\rho_A)
\le
S_\alpha(\sigma_A)-S_\beta(\sigma_A)
\qquad \forall \alpha<\beta .
\end{equation}
Defining $G_{\sigma|\rho}(\alpha):=S_\alpha(\sigma_A)-S_\alpha(\rho_A)$,
we obtain
\begin{equation}
G_{\sigma|\rho}(\alpha)\ge G_{\sigma|\rho}(\beta)
\qquad \forall \alpha<\beta .
\end{equation}
This is precisely the statement that the function $\alpha\mapsto G_{\sigma|\rho}(\alpha)$
is non-increasing on $(0,\infty)$. Reversing the same algebraic steps proves the converse implication.
\end{proof}

Now we present an explicit example that clearly shows the decreasing in the computational cost by the reformulation of the ordering condition using the one-parameter family of functions. In numerical applications, the Rényi parameter is typically sampled on a finite ordered grid $\alpha_1<\alpha_2<\cdots<\alpha_N$.
A direct test of the original two-parameter condition requires checking all pairs
\begin{equation}
    \Delta_{\alpha_i\alpha_j}(\rho_A)
    \leq
    \Delta_{\alpha_i\alpha_j}(\sigma_A),
    \qquad
    i<j.
    \label{eq:pairwise}
\end{equation}
The number of inequalities is therefore $\frac{N(N-1)}{2}$.
By contrast, the one-parameter formulation requires checking the monotonicity of $G_i:=G_{\sigma|\rho}(\alpha_i)$.
On the grid, this is equivalent to the nearest-neighbor conditions
\begin{equation}
    G_i\geq G_{i+1},
    \qquad
    i=1,\dots,N-1.
\end{equation}
Indeed, if these $N-1$ inequalities hold, then for every $i<j$ one has
\begin{equation}
    G_i\geq G_{i+1}\geq\cdots\geq G_j,
\end{equation}
and hence $G_i\geq G_j$.
Thus all pairwise spread inequalities follow.
Therefore, after the Rényi entropies have been evaluated, the number of ordering inequalities to be checked is reduced from $O(N^2)$ to $O(N)$. This reduction concerns the comparison step. 
It does not remove the spectral cost associated with diagonalizing the reduced density matrices or evaluating their Rényi entropies.

As a simple illustration, consider the two spectra
\begin{equation}
    \lambda(\rho_A)=(0.51,0.49),
    \qquad
    \lambda(\sigma_A)=(0.71,0.29).
\end{equation}
Using natural logarithms and sampling $\alpha\in\{0.5,1,2,4,8,\infty\}$
one obtains
\begin{equation}
\begin{array}{c|cccccc}
\alpha 
& 0.5 & 1 & 2 & 4 & 8 & \infty \\
\hline
G_{\sigma|\rho}(\alpha)
& -0.0472 & -0.0908 & -0.1621 & -0.2448 & -0.3002 & -0.3309
\end{array}
\end{equation}
which is monotonically non-increasing. 
Thus, on this grid, all pairwise inequalities \eqref{eq:pairwise} follow from the five nearest-neighbour checks on $G_{\sigma|\rho}$, instead of the fifteen direct checks on the Rényi spreads.

Finally, we stress that a grid-based test verifies the condition only on the chosen discrete set of Rényi indices. 
A proof of continuous antiflat comparability requires either an analytic proof of monotonicity,
\begin{equation}
    \partial_\alpha G_{\sigma|\rho}(\alpha)\leq0,
    \qquad
    \forall \alpha>0,
\end{equation}
or an argument controlling the behavior of $G_{\sigma|\rho}$ between grid points.

\subsection{The iso-purity sign structure}

The previous proposition becomes particularly transparent on the iso-purity manifold, where the second Rényi entropy is fixed and therefore acts as a distinguished anchor.

\begin{proof}
Let $\rho,\sigma\in\mathcal{M}_P$. By definition of the iso-purity manifold, $\Tr(\rho_A^2)=\Tr(\sigma_A^2)=P$,
hence $S_2(\rho_A)=S_2(\sigma_A)$.
Therefore $G_{\sigma|\rho}(2)=S_2(\sigma_A)-S_2(\rho_A)=0$. If moreover $\rho\prec_{AF}\sigma$, then by Proposition~\ref{prop:one_parameter_AF} the function $G_{\sigma|\rho}$ is non-increasing. Hence, for every $\alpha<2$ one has $G_{\sigma|\rho}(\alpha)\ge G_{\sigma|\rho}(2)=0$, while for every $\alpha>2$ one has $G_{\sigma|\rho}(\alpha)\le G_{\sigma|\rho}(2)=0$. Rewriting these inequalities in terms of Rényi entropies yields $S_\alpha(\sigma_A)\ge S_\alpha(\rho_A)
\qquad \forall \alpha<2$, and $S_\alpha(\sigma_A)\le S_\alpha(\rho_A)
\qquad \forall \alpha>2$.
\end{proof}

\subsection{Rank obstruction on a fixed iso-purity manifold}

Here we show that, once the purity is fixed, same-rank comparability becomes trivial.

\begin{proof}[Proof of Proposition~\ref{thm:freezing}]
Let $G_{\sigma|\rho}(\alpha):=S_\alpha(\sigma_A)-S_\alpha(\rho_A)$. Assume first that $\rho\prec_{AF}\sigma$. By Proposition~\ref{prop:one_parameter_AF}, the function $G_{\sigma|\rho}$ is non-increasing on $(0,\infty)$. Assume now that $\rank(\rho_A)=\rank(\sigma_A)$. Using the standard limit of the Rényi entropy at $\alpha\to 0^+$, one has
\begin{equation}
\lim_{\alpha\to 0^+} S_\alpha(\tau_A)=S_0(\tau_A)=\log\rank(\tau_A)
\end{equation}
for any finite-dimensional state $\tau_A$. Therefore
\begin{equation}
\lim_{\alpha\to 0^+}G_{\sigma|\rho}(\alpha) = \log\rank(\sigma_A)-\log\rank(\rho_A)=0.
\end{equation} 
Since $\rho,\sigma\in\mathcal{M}_P$, then $G_{\sigma|\rho}(2)=0.$

We now show that this forces
\begin{equation}
G_{\sigma|\rho}(\alpha)=0
\qquad \forall \alpha\in(0,2].
\end{equation}
Suppose by contradiction that there exists $\alpha_0\in(0,2)$ such that $G_{\sigma|\rho}(\alpha_0)>0$.
Since $G_{\sigma|\rho}$ is non-increasing, for every $0<\alpha<\alpha_0$ one would have
\begin{equation}
G_{\sigma|\rho}(\alpha)\ge G_{\sigma|\rho}(\alpha_0)>0,
\end{equation}
which contradicts the limit
\begin{equation}
\lim_{\alpha\to 0^+}G_{\sigma|\rho}(\alpha)=0.
\end{equation}
Conversely, if there existed $\alpha_0\in(0,2)$ such that $G_{\sigma|\rho}(\alpha_0)<0$,
then monotonicity would imply
\begin{equation}
G_{\sigma|\rho}(2)\le G_{\sigma|\rho}(\alpha_0)<0,
\end{equation}
contradicting $G_{\sigma|\rho}(2)=0$. Hence $G_{\sigma|\rho}(\alpha)=0$, $\forall \alpha\in(0,2]$, then $S_\alpha(\sigma_A)=S_\alpha(\rho_A)
\qquad \forall \alpha\in(0,2]$.
For a finite spectrum, the function $\alpha\mapsto \Tr(\tau_A^\alpha)=e^{(1-\alpha)S_\alpha(\tau_A)}$
is real-analytic on $(0,\infty)$. Therefore equality on the open interval $(0,2)$ implies $\Tr(\sigma_A^\alpha)=\Tr(\rho_A^\alpha)$, $ \forall \alpha>0$.
Thus all power sums of the non-zero eigenvalues coincide. Since a finite spectrum is uniquely determined by its full set of power sums, the ordered spectra must coincide:
\begin{equation}
\lambda(\rho_A)^\downarrow=\lambda(\sigma_A)^\downarrow .
\end{equation}
\end{proof}

\subsection{Interpretation and relation with the freezing of standard majorization}

The previous proposition should be interpreted as a rigidity statement intrinsic to the antiflat order on $\mathcal{M}_P$. In the main text, Proposition~\ref{thm:freezing} shows that standard majorization is frozen on the iso-purity manifold: if $\rho\prec \sigma$ and $\rho,\sigma\in\mathcal{M}_P$, then equality of purity forces equality of the ordered spectra, so no non-trivial majorization flow can occur.

The result proven here is conceptually parallel, but stronger in a different sense. Proposition~\ref{thm:freezing} shows that even for the antiflat ordering, comparability on $\mathcal{M}_P$ becomes trivial as soon as the support size is fixed. Therefore, any non-trivial FPO flow on $\mathcal{M}_P$ must necessarily involve a change in rank. In this way, the endpoint inequalities
\begin{equation}
S_1(\rho_A)\le S_1(\sigma_A),
\qquad
S_\infty(\rho_A)\ge S_\infty(\sigma_A),
\end{equation}
which follow immediately from evaluating the spread constraints at $(\alpha,\beta)=(1,2)$ and $(2,\infty)$, are seen to be the first shadows of a much sharper structural fact: non-trivial antiflat dynamics on an iso-purity manifold cannot occur without a genuine reshaping of the support.

Thus, the freezing of standard majorization and the rank obstruction for antiflat majorization fit together into a single geometric picture. Standard majorization is spectrally rigid on $\mathcal{M}_P$ because it is governed by Schur-concavity at fixed purity, whereas antiflat convertibility probes a different geometry, controlled by the monotonic decrease of Rényi spread. This is precisely why non-trivial FPO flows, when they exist, must lie outside the usual majorization-driven sector.

\section{Probability of Antiflatness Majorization}
\label{App: Probability of Antiflatness}

In this section, we rigorously characterize the geometry of the antiflatness majorization set and explicitly compute the restricted integral Eq.~\eqref{eq: iso_purity_integral} for low-dimensional systems.

\subsection{Proof of Non-Linearity: Why $\mathcal{A}_{AF}$ is not a Polytope}
In standard majorization, the set $M_a(\rho)$ is a polytope because it is defined by a finite number of linear inequalities, namely the $d-1$ constraints on cumulative sums of ordered eigenvalues, $\sum_{i=1}^k x_i \ge \sum_{i=1}^k \rho_i$. A polytope is, by definition, a bounded region obtained as the intersection of finitely many half-spaces, and is therefore enclosed by flat hyperplanes.
The accessible antiflatness set $\mathcal{A}_{AF}(\rho)$ does not satisfy these structural properties. Its boundary is determined by conditions of the form $\Delta_{\alpha\beta}(\vec{x}) = C$, where the Rényi spread depends on the Rényi entropies $S_\alpha(\vec{x}) = \frac{1}{1-\alpha}\log\!\left(\sum_i x_i^\alpha\right)$. For $\alpha \neq 1$, these expressions involve logarithms of polynomial functions, which generate intrinsically curved hypersurfaces in the probability simplex $\mathbf{\Delta_{d-1}}$ rather than affine hyperplanes.
In addition, the defining condition $\Delta_{\alpha\beta}(\vec{x}) \ge \Delta_{\alpha\beta}(\rho)$ must hold for all $\alpha < \beta$. Since the parameter space $(0,\infty)$ is continuous, this amounts to imposing an uncountable family of constraints. Consequently, $\mathcal{A}_{AF}(\rho)$ arises as the intersection of infinitely many non-linear regions.
It follows that $\mathcal{A}_{AF}(\rho)$ is a non-linear geometric object with curved boundaries, and therefore cannot be a polytope.

\subsection{Analytical Calculation of the Accessible Antiflatness Volume for low dimensions}
\label{app:AF_probability_dA2}

In this appendix we compute explicitly the accessible antiflatness probability for $d_A=2$.

\subsection{Determination of the accessible antiflatness set}

Let the target spectrum be $\vec r=(r,1-r), r\in[1/2,1]$, and let the random spectrum be $\vec x=(x,1-x), x\in[1/2,1]$. The accessible antiflatness set is $I_{AF}(r) = \left\{ x\in[1/2,1] \ \middle|\ \Delta_{\alpha\beta}(x) \geq \Delta_{\alpha\beta}(r), \quad \forall\alpha<\beta \right\}$. Equivalently, defining
\begin{equation}
    G_{x|r}(\alpha)
    :=
    S_\alpha(x,1-x)-S_\alpha(r,1-r),
\end{equation}
one has
\begin{equation}
    x\in I_{AF}(r)
    \qquad
    \Longleftrightarrow
    \qquad
    G_{x|r}(\alpha)
    \text{ is non-increasing in } \alpha .
\end{equation}

If the target is flat on its support, then either $r=\frac12$ or r=1. In both cases the Rényi entropy is independent of $\alpha$, and therefore $\Delta_{\alpha\beta}(r)=0, \forall \alpha<\beta$. Since Rényi entropies are non-increasing functions of the Rényi index,
\begin{equation}
    S_\alpha(x,1-x)\geq S_\beta(x,1-x),
    \qquad
    \alpha<\beta,
\end{equation}
we have $\Delta_{\alpha\beta}(x)\geq0 = \Delta_{\alpha\beta}(r)$ for every $x\in[1/2,1]$. Hence $I_{AF}(r)=[1/2,1], r\in\left\{\frac12,1\right\}$.
Now consider a non-flat target, $r\in(1/2,1)$. Suppose that $x\in I_{AF}(r)$. Since both $(x,1-x)$ and $(r,1-r)$ are full-rank binary spectra, one has
\begin{equation}
    S_0(x,1-x)=S_0(r,1-r)=\log 2.
\end{equation}
Thus $ G_{x|r}(0)=0$,
where the value at $\alpha=0$ is understood as the limiting Rényi entropy. If $G_{x|r}$ is non-increasing, then for every $\alpha>0$, $G_{x|r}(\alpha)\leq0$, or equivalently $S_\alpha(x,1-x) \leq S_\alpha(r,1-r)$. For binary ordered spectra, $S_\alpha(x,1-x)$ is strictly decreasing in $x$ on $[1/2,1]$ for every $\alpha>0$. Therefore $x\geq r$.
It remains to show that $x>r$ is impossible. For large $\alpha$, the binary Rényi entropy has the asymptotic expansion
\begin{equation}
    S_\alpha(x,1-x) = -\log x - \frac{\log x}{\alpha-1} + o\!\left(\frac{1}{\alpha}\right).
\end{equation}
Hence
\begin{equation}
    G_{x|r}(\alpha) = \log\frac{r}{x} + \frac{\log(r/x)}{\alpha-1} + o\!\left(\frac{1}{\alpha}\right).
\end{equation}
If $x>r$, then $\log\frac{r}{x}<0$. Consequently, for sufficiently large $\alpha$,
\begin{equation}
    \frac{d}{d\alpha}G_{x|r}(\alpha)>0,
\end{equation}
which contradicts the requirement that $G_{x|r}$ be non-increasing. 
Thus $x>r$ is not allowed, and the only possibility is $x=r$. Therefore, for every non-flat target,
\begin{equation}
    I_{AF}(r)=\{r\},
    \qquad
    r\in(1/2,1).
\end{equation}

\subsection{Evaluation of the probability}

The accessible antiflatness probability is
\begin{equation}
    \mathbb{P}(\sigma_A\succ_{AF}\rho_A)
    =
    \int_{I_{AF}(r)}p_K(x)\,dx,
\end{equation}
where
\begin{equation}
    p_K(x) = \frac{2(2K-1)}{B(K-1,K-1)} (2x-1)^2 [x(1-x)]^{K-2}.
\end{equation}
For a flat target, $r=1/2$ or $r=1$, the accessible set is the whole ordered simplex:
\begin{equation}
    I_{AF}(r)=[1/2,1].
\end{equation}
Therefore
\begin{equation}
    \mathbb{P}(\sigma_A\succ_{AF}\rho_A)
    =
    \int_{1/2}^{1}p_K(x)\,dx
    =
    1.
\end{equation}
For a non-flat target, $r\in(1/2,1)$, the accessible set is a single point:
\begin{equation}
    I_{AF}(r)=\{r\}.
\end{equation}
Since $p_K(x)$ is a continuous density on $[1/2,1]$, the measure of a single point vanishes:
\begin{equation}
    \mathbb{P}(\sigma_A\succ_{AF}\rho_A)
    =
    \int_{\{r\}}p_K(x)\,dx
    =
    0.
\end{equation}
Thus, for $d_A=2$ and arbitrary $d_B=K\geq2$,
\begin{equation}
\boxed{
\mathbb{P}(\sigma_A\succ_{AF}\rho_A)
=
\begin{cases}
1,
&
\rho_A \text{ is flat on its support},
\\[1mm]
0,
&
\rho_A \text{ is full-rank and non-flat}.
\end{cases}
}
\end{equation}
Although the full antiflat order gives only the trivial probabilities above for $d_A=2$, it is useful to record the probability mass of a general interval. 
This becomes relevant if one studies a scalar antiflatness threshold or a finite-grid approximation to the full order. Let $I=[x_-,x_+]\subset[1/2,1]$. Then
\begin{equation}
    \mathbb{P}(x\in I)
    =
    \int_{x_-}^{x_+}p_K(x)\,dx.
\end{equation}
Using the substitution $t=4x(1-x)$, one obtains
\begin{equation}
    \int_{x_-}^{x_+}p_K(x)\,dx
    =
    I_{4x_-(1-x_-)}
    \left(K-1,\frac32\right)
    -
    I_{4x_+(1-x_+)}
    \left(K-1,\frac32\right),
\end{equation}
where $I_z(a,b)$ is the regularized incomplete beta function.
In particular, for an interval of the form $[\lambda,1]$,
\begin{equation}
    \int_{\lambda}^{1}p_K(x)\,dx
    =
    I_{4\lambda(1-\lambda)}
    \left(K-1,\frac32\right).
\end{equation}
For $K=2$, this reduces to
\begin{equation}
    \int_{\lambda}^{1}6(2x-1)^2\,dx
    =
    1-(2\lambda-1)^3.
\end{equation}

\section{Properties of Linear Renyi spread and logarithmic antiflatness}
\label{App: LinearRenyiSpread_properties}
\label{App: LogLambda_properties}
In this section, we explicitly prove the subadditivity of the linear Rényi spread and some properties of the logarithmic antiflatness from those of Rènyi entropies.
\begin{proposition}
Let $\ket{\rho} \in \mathcal{H}_A \otimes \mathcal{H}_B$ and $\ket{\psi} \in \mathcal{H}_{A'} \otimes \mathcal{H}_{B'}$ be quantum states then the following holds
\begin{equation}
    \mathcal{F}_A(\rho) \mathcal{F}_A(\psi)\leq \mathcal{F}_A(\rho \otimes \psi) \leq \mathcal{F}_A(\rho) + \mathcal{F}_A(\psi)\,.
\end{equation}
\end{proposition}

\begin{proof}
    Since $\Tr_{BB'} \big[ \rho \otimes \psi \big] = \rho_A \otimes \psi_{A'}$, we have that
\begin{equation}
\begin{split}
    \mathcal{F}_A(\rho \otimes \psi)
= \Tr[\rho_A^3]\Tr[\psi_{A'}^3]
- \Tr^2[\rho_A^2] \Tr^2[\psi_{A'}^2]
= \mathcal{F}_A(\rho)\Tr[\psi_{A'}^3]
+ \Tr^2[\rho_A^2] \mathcal{F}_A(\psi)
\leq \mathcal{F}_A(\rho) + \mathcal{F}_A(\psi)\,,
\end{split}
\end{equation}
since $\Tr[\psi_{A'}^3]$ and $\Tr^2[\rho_A^2]$ are both between zero and one.
For the lower bound we have that 
\begin{equation}
\mathcal{F}_A(\rho \otimes \psi) - \mathcal{F}_A(\rho)\mathcal{F}_A(\psi)
= \mathcal{F}_A(\rho)\Tr^2[\psi_{A'}^2]
+ \Tr^2[\rho_A^2] \mathcal{F}_A(\psi) \geq 0.
\end{equation}
\end{proof}


\begin{proposition}
    \label{Prop:faith_pos}
    Given $\rho, \rho' \in \mathcal{D}(\mathcal{H})$ on a $d$-dimensional Hilbert space, then the following holds
    \begin{itemize}
    \item (Faithful)  $\log(\Lambda_\rho)=0 $ if and only if $\rho$ is proportional to a projector ;
    \item (Positive) $\log(\Lambda_\rho)\geq 0$;
    \item (Additive) $\log(\Lambda_{\rho \otimes \rho'})=\log(\Lambda_{\rho})+\log(\Lambda_{\rho'})$; 
    \item (Uniform continuity) $|\log(\Lambda_\rho) - \log(\Lambda_{\rho'})| \leq 2 d (2 + 3 d) \epsilon $, with $\epsilon:= \frac{1}{2} \norm{\rho - \rho'}_1$.
\end{itemize}
\end{proposition}
\begin{proof}
    If $\rho \in \text{FLAT}$, by definition we have that $\log(\Lambda_\rho)=0 $.
    Furthermore, If $\log(\Lambda_\rho)=0 $ this means that $\log \Tr(\rho^3)- \log \Tr^2(\rho^2)=0$, since the logarithm is a monotonic function $\Tr(\rho^3)= \Tr^2(\rho^2)$, hence $\rho$ is flat.
    Positivity follows from the properties of the Renyi entropies $S_\alpha \geq S_\beta$ with $\alpha \leq \beta$.
    Whereas additivity follows from the additivity of the Renyi entropies.
    Finally, uniform continuity follows from the uniform continuity of the Renyi entropies \cite{Chen2017}.
\end{proof}
The above result naturally holds also for $\rho \to \rho_A$.

\section{Proof of Theorem \ref{thm:Maximal_Antiflatness} on Universal Maximal Antiflatness}
\label{App : Maximal_Antiflatness}

In this section, we are going to prove the theorem on Universal maximal antiflatness.
The proof for the varentropy, i.e. capacity of entanglement, was done in~\cite{reeb2015tight}.
Here we show, using the same argument of Lagrangian multipliers, that a similar result, states with a jump spectrum, maximizes also the Linear Renyi spread and the Logarithmic antiflatness. 

We start with the Linear Rényi spread.
Given a fixed $d \geq 2$, one maximizes the expression of $\mathcal{F}_A(\rho)$ over all probability distributions $\{ p_i \}$, i.e. spectra of $\rho$, which leads to the Lagrange function
\begin{equation}
    L( \{p_i \}, v ):=\sum_{i=1}^d p_i^3-\left(\sum_{i=1}^d p_i^2 \right)^2+\nu \left(\sum_{i=1}^d p_i -1 \right) \; ,
\end{equation}
with Lagrange multipliers $\nu$ to ensure the normalization to one.
Suppose $\{ \hat{p}_i \}$ attains the maximum of $\mathcal{F}_A(\rho)$ over all probability distributions $\{ p_i \}$. Indeed, since we are considering a compact set onto which a continuous function is defined and for the extreme value theorem we are guaranteed that the maximum is attained.
We can consider $L$ as a function of those variables $p_i$ for which $\hat{p}_i >0$ and the rest are all zeros.
The Lagrangian multipliers method guarantees the existence of $\hat{\nu} \in \mathds{R} : \frac{d L}{d p_j} \Big |_{\{ \hat{p}_i \}}$ for every $j$  by the extremality of $\{ \hat{p}_i \}$and that all the components are in the interior of the Lagrangian domain.
As a result, we obtain that
\begin{equation}
    \hat{p}_j= \frac{1}{3} \Big ( 2 Pur(\hat{\rho}) \pm \sqrt{4 Pur(\hat{\rho})^2 - 3 \hat{\nu}} \Big ) \; ,
\end{equation}
where $Pur(\hat{\rho})=tr(\hat{\rho}^2)$ denotes the purity of the states that attain the maximum.
We need to find out how many $\hat{p}_j$ have one value and how many the other.
Hence, leaving off hats again, an optimal $\hat{\rho}= \rho$ has the form
\begin{equation}
    \rho = \text{diag} \Big ( \frac{1-r}{m}, \ldots , \frac{1-r}{m},\frac{r}{n}, \ldots, \frac{r}{n}\Big) \; ,
\end{equation}
with $r \in [0,1]$ and $m,n \geq 1$ such that $m+n \leq d$.
Since we can assume that $r\leq \frac{1}{2}$ without loss of generality due to the permutation invariance of the entries of $\rho$, we have that
\begin{equation}
    \mathcal{F}_A(\rho)= \frac{r(1-r)(n(r-1)+mr)^2}{m^2 n^2} \geq 0 \; ,
\end{equation}
for every choice of $m,n,r$ as stated above.
By maximizing over $m,n$ for fixed $r$ one arrives at the conclusion that $m=1$ and $n=d-1$.
In order to find the absolute maximum, we extremize over $r$ and find that 
\begin{equation}
    r_{max}^{\mathcal{F}_A}= \frac{5d-2- \sqrt{4-4d+9d^2}}{8d} \; .
\end{equation}
Thus the maximum value of the antiflatness, defined as $N(d)$, is

\begin{equation}
   N_{\mathcal{F}_A}(d)= \frac{\left(5 d-\sqrt{d (9 d-4)+4}-2\right) \left(3 d+\sqrt{d (9 d-4)+4}-6\right)^2 \left(3 d+\sqrt{d (9 d-4)+4}+2\right)}{4096 (d-1)^2 d^2}\; .
\end{equation}

Additionally, a quite tight upper bound to $N_{\mathcal{F}_A}(d)$ can be found considering that the function is monotonic increasing in $d \geq 2$ and taking the limit for $d \to \infty$
. Namely
\begin{equation}
    \lim_{d \to \infty} N_{\mathcal{F}_A}(d)= \frac{27}{256} \;.
\end{equation}

The same procedure can be followed for the Logarithmic antiflatness.
Given the Lagrange function
\begin{equation}
\begin{split}
    L( \{p_i \}, \nu ):=\log\left(\sum_{i=1}^d p_i^3\right)-\log\left(\left(\sum_{i=1}^d p_i^2 \right)^2\right)+\nu \left(\sum_{i=1}^d p_i -1 \right) \; .
    \end{split}
\end{equation}.
Suppose $\{ \hat{p}_i \}$ attains the maximum of $\log(\Lambda_\rho)$ over all probability distributions $\{ p_i \}$. Indeed, since we are considering a compact set onto which a continuous function is defined and for the extreme value theorem we are guaranteed that the maximum is attained, invoking the same reasoning as for the linear version and using that the logarithm is a monotonic function.
We can consider L as a function of those variables $p_i$ for which $\hat{p}_i >0$ and the rest are all zeros.
The Lagrangian multipliers method guarantees the existence of $\hat{\nu} \in \mathds{R} : \frac{d L}{d p_j} \Big |_{\{ \hat{p}_i \}}$ for every $j$  by the extremality of $\{ \hat{p}_i \}$and that all the components are in the interior of the Lagrangian domain.
As a result, we obtain that
\begin{equation}
\begin{split}
    \hat{p}_j= \frac{1}{3 (\mathcal{F}_A(\hat{\rho})-Pur(\hat{\rho})^2)} \left ( \frac{2}{ Pur(\hat{\rho})} \pm \sqrt{4 Pur(\hat{\rho})^2 - \frac{3 \hat{\nu}}{\mathcal{F}_A(\hat{\rho})-Pur(\hat{\rho})^2}} \right ) \; ,
    \end{split}
\end{equation}
where $Pur(\hat{\rho})$ denotes the purity and $\mathcal{F}_A(\hat{\rho})$ the Linear Rènyi spread of the states that attain the maximum.
We need to find out how many $\hat{p}_j$ have one value and how many the other.
Hence, leaving off hats again, an optimal $\hat{\rho}= \rho$ has the form
\begin{equation}
    \rho = \text{diag} \Big ( \frac{1-r}{m}, \ldots , \frac{1-r}{m},\frac{r}{n}, \ldots, \frac{r}{n}\Big) \; ,
\end{equation}
with $r \in [0,1]$ and $m,n \geq 1$ such that $m+n \leq d$.
Since we can assume that $r\leq \frac{1}{2}$ without loss of generality due to the permutation invariance of the entries of $\rho$, we have that
\begin{equation}
\begin{split}
    \log(\Lambda_\rho)= \log\left( \frac{(1-r)^3}{m^2} + \frac{r^3}{n^2}\right)- 2 \log\left( \frac{(1-r)^2}{m} + \frac{r^2}{n}\right),
    \end{split}
\end{equation}
for every choice of $m,n,r$ as stated above.
By maximizing over $m,n$ for fixed $r$ one arrives at the conclusion that $m=1$ and $n=d-1$.
We are left with determing
\begin{equation}
    N^{\Lambda}(d)= \max_{r \in [0, \frac{1}{2}]} \log(\frac{-d^2 (r-1)^3+2 d (r-1)^3+3 r (r-1)+1}{\left(d (r-1)^2+2 r-1\right)^2})
\end{equation}
which leads to
\begin{equation}
r_{\max}^{\Lambda}=
    \begin{cases}
        \frac{1}{6} \left(3-\sqrt{3}\right) &\quad d=2\\
        2-\sqrt{6} \sin \left(\frac{1}{3} \tan ^{-1}\left(\frac{1}{2 \sqrt{2}}\right)\right)+\sqrt{2} \left(-\cos \left(\frac{1}{3} \tan ^{-1}\left(\frac{1}{2 \sqrt{2}}\right)\right)\right) &\quad d=3\\
       \frac{3}{2}-\frac{1}{2} 3 \sin \left(\frac{\pi }{18}\right)-\frac{1}{2} \sqrt{3} \cos \left(\frac{\pi }{18}\right)& \quad d=4\\
        \frac{4}{3}-\frac{2 \sin \left(\frac{1}{3} \tan ^{-1}\left(\frac{3}{4}\right)\right)}{\sqrt{3}}-\frac{1}{3} 2 \cos \left(\frac{1}{3} \tan ^{-1}\left(\frac{3}{4}\right)\right)&\quad d=5\\
        \frac{5}{4}-\frac{1}{4} \sqrt{15} \sin \left(\frac{1}{3} \tan ^{-1}\left(\frac{2}{\sqrt{5}}\right)\right)-\frac{1}{4} \sqrt{5} \cos \left(\frac{1}{3} \tan ^{-1}\left(\frac{2}{\sqrt{5}}\right)\right)& \quad d=6\\
        \frac{1}{2} &\quad d \geq 7
    \end{cases}\,.
\end{equation}
We can derive an upper bound on $N^{\Lambda}(d)$, namely
\begin{equation}
    \lim_{d \to \infty} N^{\Lambda}(d)= \log 2\, .
\end{equation}

\section{Compatibility of the partial order with antiflatness quantifiers}
\label{app:compatibility_proofs}
In this section we verify the compatibility of the proposed antiflatness ordering described in the main text with the established quantifiers.
\begin{itemize}
    \item \textbf{Compatibility with $\log(\Lambda_{\rho})$:}
    Since $\log(\Lambda_{\rho}) = 2(S_2 - S_3) = 2\Delta_{2,3}(\rho)$, by the definition of the antiflatness ordering, if a state $\rho$ is antiflat majorized by $\sigma$, it strictly follows that $\Delta_{2,3}(\rho) \le \Delta_{2,3}(\sigma)$. Therefore, the ordering is naturally preserved.
    \item \textbf{Compatibility with $\mathcal{V}_A(\rho)$:} 
    The Capacity of Entanglement usually is defined as the second derivative of Renyi entropy in $\alpha$, see \cite{cao2024gravitational}. Here we first show that it can be derived directly from the slope of the Rényi entropy evaluated at $\alpha = 1$,
    \begin{equation}
        \begin{split}
            \left. \frac{\partial S_\alpha(\rho)}{\partial \alpha} \right|_{\alpha=1} &= -\frac{1}{2} \left. \frac{d^2}{d\alpha^2} \ln \text{Tr}(\rho^\alpha) \right|_{\alpha=1} = -\frac{1}{2} \Big( \text{Tr}(\rho \ln^2 \rho) - [\text{Tr}(\rho \ln \rho)]^2 \Big) = -\frac{1}{2} \text{Var}_\rho(-\log \rho) = -\frac{1}{2} \mathcal{V}_A(\rho) \ .
        \end{split}
    \end{equation}
    Because the Rényi entropy is monotonically decreasing, a larger spread $\Delta_{\alpha\beta}$ implies a steeper (more negative) slope at $\alpha=1$. Consequently, an ordering in the spread dictates an ordering in the absolute value of the derivative. Thus, if $\rho \prec_{AF} \sigma$, the absolute value of the derivative for $\rho$ is bounded by that of $\sigma$
    \begin{equation}
        \left| \left. \frac{\partial S_\alpha(\rho)}{\partial \alpha} \right|_{\alpha=1} \right| \le \left| \left. \frac{\partial S_\alpha(\sigma)}{\partial \alpha} \right|_{\alpha=1} \right| \implies \frac{1}{2}\mathcal{V}_A(\rho) \le \frac{1}{2}\mathcal{V}_A(\sigma) \; \ .
    \end{equation}
    Therefore, also the Capacity of Entanglement perfectly respects the antiflatness majorization.

    \item \textbf{Conditional Monotonicity of $\mathcal{F}_A(\rho)$:} 
    The linear Rényi spread $\mathcal{F}_A(\rho)$ requires careful treatment. It measures the antiflatness weighted by the purity of the state. We can rewrite this purely in terms of Rényi entropies
    \begin{equation}
        \mathcal{F}_A(\rho) = e^{-2S_2(\rho)}(e^{2\Delta_{2,3}(\rho)} - 1) \; \ .
    \end{equation}
    In this form, the term $(e^{2\Delta_{2,3}(\rho)} - 1)$ represents the \textit{shape factor} , while the term $e^{-2S_2(\rho)}$ acts as a \textit{scale factor} that depends strictly on the purity. Because of this purity factor, the ordering works only if we compare states possessing the same purity, thereby restricting the valid comparisons to iso-purity manifolds. We formalize this via the following Lemma. 

\end{itemize}

\begin{lemma}[Conditional Monotonicity of $\mathcal{F}_A(\rho)$] \label{lemma: cond monotonicity}
    Let $\rho, \sigma \in \mathcal{D}(\mathcal{H})$. If $\rho \prec_{AF} \sigma$ and $S_2(\rho) \ge S_2(\sigma)$ (which is equivalent to $\text{Tr}(\rho^2) \le \text{Tr}(\sigma^2)$ ), then $\mathcal{F}_A(\rho) \le \mathcal{F}_A(\sigma)$.
    \end{lemma}
    \begin{proof} 
    We analyze each factor of our function separately.
    Because $\rho \prec_{AF} \sigma$, we know that $\Delta_{2,3}(\rho) \le \Delta_{2,3}(\sigma)$. Since the exponential is a strictly increasing function, it follows that $\left(e^{2\Delta_{2,3}(\rho)} - 1\right) \le \left(e^{2\Delta_{2,3}(\sigma)} - 1\right)$. Imposing the condition of the lemma, $S_2(\rho) \ge S_2(\sigma)$, mathematically implies that $e^{-2S_2(\rho)} \le e^{-2S_2(\sigma)}$.
    By multiplying these non-negative factors, we yield the final inequality $\rho \prec_{AF} \sigma \land \text{Tr}(\rho^2) \le \text{Tr}(\sigma^2) \implies \mathcal{F}_A(\rho) \le \mathcal{F}_A(\sigma) $.
    \end{proof}

\subsubsection{Geometric Interpretation of Iso-Purity Manifolds}\label{sec:iso_purity}
To visualize the state space geometrically can be useful to understand the role of iso-purity manifolds. Fixing the purity to a constant value $P$, such that $\text{Tr}(\rho_A^2) = \sum_{i=1}^{d_A} \lambda_i^2 = P$, is geometrically equivalent to constraining the eigenvalue vector $\vec{\lambda}(\rho_A)^{\downarrow}$ to lie on the surface of a hypersphere of radius $R = \sqrt{P}$, centered at the origin $(0, \dots, 0)$. Naturally, the global bounds on purity restrict this radius to $\frac{1}{\sqrt{d_A}} \le R \le 1$. For $\rho_A$ to represent a valid physical state, it must also belong to the probability simplex, satisfying the normalization constraint $\sum_{i=1}^{d_A} \lambda_i = 1$. In $\mathbb{R}^{d_A}$, this equation defines a $(d_A-1)$-dimensional hyperplane. Therefore, the set of states possessing the exact same purity is given by the intersection between the probability hyperplane and the hypersphere. This geometric intersection generates a $(d_A-2)$-dimensional sphere, which defines the \textit{iso-purity manifolds}. When we evaluate the shape factor $(e^{2\Delta_{2,3}(\rho_A)} - 1)$ from the Rényi spread, we are essentially measuring the antiflatness strictly within one of these $(d_A-2)$-dimensional spherical slices. The function $\mathcal{F}_A(\rho)$ scales this internal geometric property by the squared radius of the hypersphere ($P^2$). Comparing the antiflatness $\mathcal{F}_A$ of two states located on spheres of wildly different radii leads to mathematically ambiguous results.

With this in mind, we can show a counterexample to the inverse logical implication ($\impliedby$) of Lemma~\ref{lemma: cond monotonicity}. We aim to prove that the conditions $\mathcal{F}_A(\rho) \le \mathcal{F}_A(\sigma)$ and $\text{Tr}(\rho_A^2) \le \text{Tr}(\sigma_A^2)$ are strictly not sufficient to conclude that $\rho \prec_{AF} \sigma$. We can prove this by constructing a $3$-dimensional counterexample where a state has a lower $\mathcal{F}_A$ simply because its purity is severely suppressed, despite having a more antiflat spectral shape. Consider two density matrices, $\rho_A$ and $\sigma_A$, $\mathcal{H_A}$ be a 3-dimensional Hilbert space, with the following sorted eigenvalue spectra $\text{spec}(\rho_A) = (0.6, 0.3, 0.1) \ , \; \text{spec}(\sigma_A) = (0.8, 0.2, 0.0)$.
We can immediately see that $\text{Tr}(\rho_A^2) \le \text{Tr}(\sigma_A^2)$ (since $0.46 \le 0.68$), and $\mathcal{F}_A(\rho) \le \mathcal{F}_A(\sigma)$ (since $0.0324 \le 0.0576$). The Rényi spread factor dictates the true antiflatness ordering. We extract it using the identity $(e^{2\Delta_{2,3}} - 1) = \mathcal{F}_A / \text{Tr}(\rho_A^2)^2$. So, $(e^{2\Delta_{2,3}(\rho_A)} - 1) = \frac{0.0324}{0.2116} \approx 0.1531 \ge (e^{2\Delta_{2,3}(\sigma_A)} - 1) = \frac{0.0576}{0.4624} \approx 0.1245$, which implies that $\Delta_{2,3}(\rho_A) > \Delta_{2,3}(\sigma_A)$. However, because its purity ($0.46$) is significantly lower than $\sigma$'s purity ($0.68$), the global measure $\mathcal{F}_A(\rho)$ is artificially suppressed below $\mathcal{F}_A(\sigma)$. Knowing $\mathcal{F}_A(\rho) \le \mathcal{F}_A(\sigma)$ and $\text{Tr}(\rho_A^2) \le \text{Tr}(\sigma_A^2)$ is not sufficient to claim $\rho \prec_{AF} \sigma$. Because of the definition of antiflat ordering ($\rho \prec_{AF} \sigma$) strictly requires $\Delta_{\alpha\beta}(\rho_A) \le \Delta_{\alpha\beta}(\sigma_A)$ to hold for all $\alpha < \beta$, finding a single violation at $(\alpha=2, \beta=3)$ is sufficient to prove that $\rho \not\prec_{AF} \sigma$.

\section{Relation between different quantifiers}
\label{App : relation quantifiers}
Since we have been studying each antiflatness quantifier independently, it is also useful to show what the relations among them are.
\begin{proposition}[Local expansion of the entanglement entropy spread]
Let $\rho$ be a full-rank density operator on a finite-dimensional Hilbert space $\mathcal{H}$, that is, 
$\rho > 0$ and $\mathrm{tr}\,\rho = 1$. 
Consequently, for any $0 \leq \alpha < \beta \leq \infty$ such that $\alpha, \beta \to 1$,
\begin{equation}
    \Delta^{\alpha \beta}(\rho)= \frac{\beta - \alpha}{2}\,
    \mathcal{V}_A(\rho)
    + o(|\beta - \alpha|).
\end{equation}
\end{proposition}

\begin{proof}
    Then the Rényi-$\alpha$ entropy is real-analytic in $\alpha$ around $\alpha = 1$, and the following expansion holds:
\begin{equation}
    S_{\alpha}(\rho)
    = S(\rho)
    - \frac{\alpha - 1}{2}\,\mathcal{V}_A(\rho)
    + o(\alpha - 1),
    \qquad (\alpha \to 1),
\end{equation}
where $S(\rho) = -\mathrm{tr}(\rho \log \rho)$.
\end{proof}

\begin{proof}
Let $\rho$ be a full-rank density operator on a finite-dimensional Hilbert space $\mathcal{H}$. 
Write the spectral decomposition $\rho = \sum_i p_i \, |i\rangle\langle i|$, with $p_i > 0$ and $\sum_i p_i = 1$.
Define
\begin{equation}
    \mathbb{Z}(\alpha) := \mathrm{tr}\,\rho^\alpha = \sum_i p_i^\alpha,
    \qquad
    f(\alpha) := \log \mathbb{Z}_\alpha.
\end{equation}
Then the Rényi-$\alpha$ entropy can be written as
\begin{equation}
    S_\alpha(\rho) = \frac{f(\alpha)}{1-\alpha}.
\end{equation}
Since all $p_i > 0$, $f(\alpha)$ is real-analytic around $\alpha = 1$, and admits a Taylor expansion
\begin{equation}
    f(\alpha) = f(1) + f'(1)(\alpha-1) + \frac{1}{2} f''(1) (\alpha-1)^2 + o((\alpha-1)^2),
\end{equation}
with
\begin{equation}
\begin{split}
    f(1) &= \log \sum_i p_i = 0, 
    \quad
    f'(1) = \sum_i p_i \log p_i = -S(\rho), \\
    f''(1) &= \sum_i p_i (\log p_i)^2 - (\sum_i p_i \log p_i)^2 = \mathcal{V}_A(\rho).
    \end{split}
\end{equation}
Substituting into $S_\alpha(\rho)$, we get
\begin{equation}
\begin{split}
    S_\alpha(\rho) 
    &= \frac{f(\alpha)}{1-\alpha} = \frac{-S(\rho)(\alpha-1) + \frac{1}{2}\mathcal{V}_A(\rho)(\alpha-1)^2 + o((\alpha-1)^2)}{-(\alpha-1)} \\
    &= S(\rho) - \frac{\alpha-1}{2}\,\mathcal{V}_A(\rho) + o(\alpha-1),
\end{split}
\end{equation}
which is the desired expansion.
Finally, for any $\alpha, \beta \to 1$, subtracting the expansions yields
\begin{equation}
    \Delta^{\alpha \beta}(\rho) := S_\alpha(\rho) - S_\beta(\rho) 
    = \frac{\beta - \alpha}{2}\,\mathcal{V}_A(\rho) + o(|\beta - \alpha|).
\end{equation}
\end{proof}

\begin{proposition}[Relation between different measures]
Let $\rho$ be a state on
a $d$-dimensional system. Then
\begin{equation}
\begin{split}
     \log(\Lambda_\rho)
    &\;\le\;  d^2 \mathcal{F}_A(\rho) \;\le\; \left(\frac{e^{\Delta^{0,\infty}(\rho)}}{2}\right)^2.
    \end{split}
\end{equation}
\end{proposition}

\begin{proof}
By definition, the logarithmic anti-flatness is
\begin{equation}
    \log(\Lambda_\rho) = \log\Big( 1 + \frac{\mathcal{F}_A(\rho)}{\mathrm{tr}^2(\rho^2)} \Big),
\end{equation}
where $\mathcal{F}_A(\rho) \ge 0$.
Since $\log(1+x) \le x$ for all $x \ge 0$, we have
\begin{equation}
    \log(\Lambda_\rho) \le  \frac{\mathcal{F}_A(\rho)}{\mathrm{tr}^2(\rho^2)}.
\end{equation}
Finally, using the general lower bound on the purity,
\begin{equation}
    \mathrm{tr}(\rho^2) \ge \frac{1}{d} \quad \implies \quad \frac{1}{\mathrm{tr}^2(\rho^2)} \le d^2,
\end{equation}
we obtain the fully universal upper bound
\begin{equation}
    \Delta^{2,3}(\rho) \le d^2 \, \mathcal{F}_A(\rho).
\end{equation}
To bound $\mathcal{F}_A(\rho)$ in terms of the entanglement spread $\Delta^{0,\infty}(\rho)$, note that the largest eigenvalue $p_{\max}$ of $\rho$ satisfies
\begin{equation}
    p_{\max} = e^{-S_\infty(\rho)} = e^{-S_0(\rho) +\Delta^{0,\infty}(\rho)} \le \frac{e^{\Delta^{0,\infty}(\rho)}}{d}.
\end{equation}
Then, since 
\begin{equation}
    \begin{split}
    \mathcal{F}_A(\rho) &= \sum_i p_i^3 - (\sum_i p_i^2)^2\leq \sum_i p_i^2 (p_{max} - \sum_i p_i^2) \leq p_{\max}^2 / 4 
    \end{split}
\end{equation}
maximum of $x(p_{max}-x)$ occurs at $x=p_{max}/2$, we get
\begin{equation}
    \mathcal{F}_A(\rho) \le \frac{p_{\max}^2}{4} \le \left(\frac{e^{\Delta^{0,\infty}(\rho)}}{2 d}\right)^2.
\end{equation}
\end{proof}

\begin{proposition}[Dipendence on escort distribution]
    Let $p = \{p_i\}_{i=1}^n$ be a probability distribution, i.e. $p_i \ge 0 \; \; \; \forall i$ and $\sum_{i=1}^n p_i = 1$. So, for a real parameter $q$ the escort distribution is defined as 
\begin{equation}
   P^{(q)} = \frac{p_i^q}{\mathbb{Z}_q} \; \; \; \ , \; \; \; \mathbb{Z}_q = \sum_{j=1}^n p_j^q \; \; . 
\end{equation}
The escort distribution has the same formal structure as the generalized canonical distribution of thermodynamics.
For $q \ne 1$ we can rewrite both the Rényi entropies and the Tsallis entropies as functions of the escort normalization constant $\mathbb{Z}_q$
\begin{equation}\begin{split}
    S_{\alpha}(p) &= \frac{1}{1-\alpha} \log(\sum_{i=1}^n p_i^{\alpha}) =  \frac{1}{1-\alpha} \log(\mathbb{Z}_{\alpha}) \; \; ,\\
    T_{\alpha}(p) &= \frac{1}{1-\alpha}(1 - \sum_{i=1}^n p_i^{\alpha}) = \frac{1 - \mathbb{Z}_{\alpha}}{1-\alpha} \; \; .
\end{split}\end{equation}
From which we can rewrite the escort normalization constant $\mathbb{Z}_q$ as
\begin{equation}
\mathbb{Z}_{\alpha}
:= \sum_{i} p_i^{\alpha}
= \exp\!\big( (1-\alpha)\, S_{\alpha}(\rho) \big) \, ,
\end{equation}
and so we have
\begin{equation}
T_{\alpha}(\rho)
= \frac{1-\mathbb{Z}_{\alpha}}{1-\alpha}
= \frac{1-\exp\!\big((1-\alpha)\, S_{\alpha}(\rho)\big)}{1-\alpha} \, .
\end{equation}
In particular, for $\alpha = 3$, we have
\begin{equation}
T_3(\rho)
= \frac{1-\exp\!\big(-2\, S_3(\rho)\big)}{-2}
= \frac{\exp\!\big(-2\, S_3(\rho)\big)-1}{2} \, ,
\end{equation}
and for $\alpha = 2$, the Tsallis entropy reads
\begin{equation}
T_2(\rho)
= 1 - \exp\!\big(- S_2(\rho)\big) \, .
\end{equation}
So, the antiflatness of $\rho$ can be rewritten as
\begin{equation}
\mathcal{F}_A(\rho) = 1 - \mathbb{Z}_3 - \left(2 - \mathbb{Z}_2\right)^2
= \exp\!\big(-2\, S_3(\rho)\big)
- \exp\!\big(-2\, S_2(\rho)\big) \; \ .
\end{equation}
While the capacity of entanglement reads
\begin{equation}
\begin{aligned}
\mathcal{V}_A(\rho)
&:= \mathrm{Var}_{\rho}(-\log \rho)= \mathrm{tr}\!\left( \rho \, \log^{2}\rho \right)
   - \mathrm{tr}^{2}\!\left( \rho \log \rho \right)
= \left. \frac{d^{2}}{dq^{2}} \log \mathbb{Z}_{q} \right|_{q=1}= \left. \frac{d^{2}}{dq^{2}}
\Big[ (1-q)\, S_{q}(\rho) \Big] \right|_{q=1} \, .
\end{aligned}
\end{equation}
Clearly, the logarithmic antiflatness of $\rho$ remains unchanged, since it is expressed directly in terms of Rényi entropies
\begin{equation}
\begin{aligned}
\log \Lambda_{\rho}
&:= \log\!\left(
\frac{\mathrm{tr}(\rho^{3})}{\mathrm{tr}^{2}(\rho^{2})}
\right)
= \log \mathbb{Z}_{3} - 2 \log \mathbb{Z}_{2}= (1-3)\, S_{3}(\rho) - 2(1-2)\, S_{2}(\rho)
= 2\big( S_{2}(\rho) - S_{3}(\rho) \big) \, .
\end{aligned}
\end{equation}
\end{proposition}

\section{Derivation of the Capacity of Entanglement as KL Curvature}\label{App: DKL Derivation}
In this section, we explicitly derive the connection between the Kullback-Leibler divergence, the escort distributions, and the Capacity of Entanglement shown in Eq.~\eqref{eq:capDKL}. Let $p$ be the eigenvalue spectrum of the state $\rho$. We evaluate the divergence between $p$ and its generalized escort distribution of order $\alpha = 1+\epsilon$, defined as 
\begin{equation}
    P^{(1+\epsilon)}_i = p_i^{1+\epsilon} / \sum_j p_j^{1+\epsilon} \; \;  .
\end{equation}
To first order in $\epsilon$, the perturbation $\delta p_i = P^{(1+\epsilon)}_i - p_i$ is given by

\begin{equation}
\delta p_i \simeq \epsilon p_i (\log p_i - \langle \log p \rangle_p) = \epsilon p_i (\log p_i + S_1(p)) \ .
\end{equation}
We expand the KL divergence $D_{KL}(p || p + \delta p)$ using the Taylor series $\log(1+x) \simeq x - \frac{x^2}{2}$
\begin{equation}
D_{KL}(p || p+\delta p) \simeq \sum_{i=1}^d p_i \left( -\log\left(1 + \frac{\delta p_i}{p_i}\right) \right) \simeq \sum_{i=1}^d p_i \left( -\frac{\delta p_i}{p_i} + \frac{1}{2}\left(\frac{\delta p_i}{p_i}\right)^2 \right) \ .
\end{equation}
Because total probability is conserved, the linear perturbation terms sum to zero ($\sum_i \delta p_i = 0$). The linear term thus vanishes, leaving a purely quadratic form defined by the Hessian
\begin{equation}
D_{KL}(p || p+\delta p) \simeq \frac{1}{2} \sum_{i=1}^d \frac{(\delta p_i)^2}{p_i} \ .\end{equation}
Substituting our specific $\delta p_i$ driven by the escort perturbation, we obtain
\begin{equation}
D_{KL}(p || P^{(1+\epsilon)}) \simeq \frac{1}{2} \sum_{i=1}^d \frac{1}{p_i} \left[ \epsilon p_i (\log p_i + S_1(p)) \right]^2 = \frac{\epsilon^2}{2} \sum_{i=1}^d p_i (\log p_i + S_1(p))^2 \ .
\end{equation}
We immediately recognize the remaining sum as the exact definition of the variance of the surprisal operator, $\text{Var}_p(-\log p)$, which is the Capacity of Entanglement $\mathcal{V}_A(\rho)$. Rearranging the terms directly yields the main text result
\begin{equation}
\mathcal{V}_A(\rho) = \lim_{\epsilon \to 0} \frac{2}{\epsilon^2} D_{KL}(\rho || \rho_{\text{escort}}^{(1+\epsilon)}) \ .
\end{equation}

\section{Haar random states ensemble and antiflatness}\label{app:Haar average}
In this section, we are going to compute the average and variance of Linear Renyi spread according to the Haar average
\begin{equation}
    \mathds{E}_{U}[\mathcal{F}_A(\rho_A)]=\mathds{E}_{U}[\Tr(\rho_A^3)]-\mathds{E}_{U}[\Tr^2(\rho_A^2)]\, .
    \label{eq:Haar_Flat}
\end{equation}
A well-known result \cite{zyczkowski2006GeometryQuantumStates}
\begin{equation}
\label{eq:tr_3}
    \mathds{E}_{U}[\Tr(\rho_A^3)]=\frac{d_A^2+d_B^2+3d_A d_B+1}{(d_A d_B +1)(d_A d_B +2)}\, , \quad \mathds{E}_{U}[\Tr(\rho_A^2)]=\frac{d_A+d_B}{d_A d_B +1}\,.
\end{equation}
In order to compute the second element we can use another statement in \cite{liu2018entanglement}, namely
\begin{equation}
    \int \mathrm{d} U\left(\Tr[\rho_{A C}^\alpha]\right)^s =\frac{1}{d^{s \alpha}} \sum_{\sigma, \gamma \in S_{s \alpha}} d_A^{\xi\left(\sigma \tau_{\alpha, s}\right)} d_B^{\xi(\sigma)} d_C^{\xi\left(\gamma \tau_{\alpha, s}\right)} d_D^{\xi(\gamma)} \mathrm{Wg}\left(d, \sigma \gamma^{-1}\right)
\end{equation}
where $\tau_{\alpha, s}:=\prod_{r=0}^{s-1}(\alpha r+1 \; \alpha r+2 \cdots \alpha(r+1))$ is the product of canonical full cycles on each of the $s$ blocks with $\alpha$ symbols. 
Here $\mathrm{Wg}$ are the Weingarten functions of $U(d)$.
In our case, this reads as

\begin{equation}
    \mathds{E}_{U}[\Tr^2(\rho_A^2)]=\frac{d_A^3 d_B+2 d_A^2 \left(d_B^2+2\right)+d_A d_B \left(d_B^2+10\right)+4 d_B^2+2}{(d_A d_B+1) (d_A d_B+2) (d_A d_B+3)}\, ,
\end{equation}

so that the final expression for \eqref{eq:Haar_Flat} is
\begin{equation}
    \mathds{E}_{U}[\mathcal{F}_A(\rho_A)]= \frac{\left(d_A^2-1\right) \left(d_B^2-1\right)}{(d_A d_B+1) (d_A d_B+2) (d_A d_B+3)} \, .
\end{equation}

The variance of this quantity, namely $\Delta^2 \mathcal{F}_A(\rho_A)=\mathds{E}_{U}[\mathcal{F}_A^2(\rho_A)]-\mathds{E}_{U}[\mathcal{F}_A(\rho_A)]^2$, can be computed in a similar manner
\begin{align}
   \Delta^2 \mathcal{F}_A(\rho_A) &=\frac{1}{(d+1)^2 (d+2)^2 (d+3)^2 (d+4) (d+5) (d+6) (d+7)}(\left(d_A^2-1\right) 
   \left(d_B^2-1\right) \\ 
   & \times  ( 2 d_A^6 d_B^4+7 d_A^5 d_B^3 \left(d_B^2-2\right)+d_A^4 d_B^2 \left(2 d_B^4-5 d_B^2-62\right)-7 d_A^3 d_B \left(2 d_B^4+11 d_B^2-26\right)\\ &+d_A^2 \left(-62 d_B^4+413 d_B^2+588\right)+2 d_A d_B \left(91 d_B^2+683\right)+84 \left(7 d_B^2+2\right) ) ) \,,
\end{align}
with $d=d_A d_B$.

Finally, the quantity that summarizes all this computation in the simple case of $d_A=2\leq d_B$ is the probability density function of the Linear Rényi spread.
It is known that for a Haar random state with $d_A=2$ and a generic $d_B$
\begin{equation}
    \ket{\rho}=\sqrt{\lambda} \ket{\phi_A \phi_B}+ \sqrt{1- \lambda} \ket{\chi_A \chi_B}\, ,
\end{equation}
the associated probability distribution is given by $P(\lambda)=3(1-2\lambda)^2$ \cite{zyczkowski2001induced}.
Since the antiflatness for this system is given by
\begin{equation}
    \mathcal{F}_A(\rho_A)=-(1-2 \lambda)^2 (\lambda-1) \lambda\,.
\end{equation}
Our goal is to determine the probability distribution of $f \equiv \mathcal{F}_A(\lambda)$, where $\lambda$ is a random variable following the distribution $P^{U}(\lambda)$. The probability distribution of $\lambda$ is then given by:
\begin{equation}
    P^{U}(f)=\sum_{\lambda: \mathcal{F}_A(\lambda)=f} \frac{P^{U}(\lambda)}{\abs{\mathcal{F}_A'(\lambda)}}\,,
\end{equation}
where $\mathcal{F}_A'(\lambda)$ is the derivative of $\mathcal{F}_A(\lambda)$ with respect to $\lambda$, which can be explicitly calculated as:
\begin{equation}
    \mathcal{F}_A'(\lambda)=(1-2 \lambda) (8 (\lambda-1) \lambda+1)\,.
\end{equation}
Furthermore, the solution to the equation $f = \mathcal{F}_A(\lambda)$ is:
\begin{equation}
    \lambda=\frac{1}{2} \left(1\pm\frac{\sqrt{1\pm\sqrt{1-16 f}}}{\sqrt{2}}\right)\,.
\end{equation}
Finally the probability distribution reads:

\begin{equation}
    P^{U}_{\mathcal{F}_A(\rho_A)}(f)=
    3 \sqrt{2} \left(\frac{\sqrt{\left(\sqrt{1-16 f}-1\right) (16 f-1)}}{1-16 f}+\sqrt{\frac{1}{\sqrt{1-16 f}}+\frac{1}{1-16 f}}\right)\, .
\end{equation}

Analogously to the anti-flatness, one can derive the probability distribution according to the haar average of the logarithmic anti-flatness for 2 qubits, which reads

\begin{equation}
\begin{split}
    P^{U}_{ \log(\Lambda_{\rho_A})}(x)&=  \Bigg(2 \left(\sqrt{9-2^{x+3}}-2^{x+1}+3\right) \left| \frac{\sqrt{\sqrt{9-2^{x+3}}-2^{x+1}+3} \left(2^{x+2} \sqrt{9-2^{x+3}}-3 \sqrt{9-2^{x+3}}+2^{x+3}-9\right)}{18-17\ 2^{x+1}+4^{x+2}}\right|\\ &-2 \left(\sqrt{9-2^{x+3}}+2^{x+1}-3\right) \left| \frac{\left(\sqrt{9-2^{x+3}}-3\right) \left(3 \sqrt{9-2^{x+3}}+2^{x+2}-9\right)}{\sqrt{\sqrt{9-2^{x+3}}+2^{x+1}-3} \left(3 \sqrt{9-2^{x+3}}+2^{x+3}-9\right)}\right| \Bigg) 2^{\frac{1}{2} (-3 x-7)}\log (8)\, .
\end{split}
\end{equation}

\section{Bures-Hall ensemble and antiflatness}

We compute the probability density function of the Linear Rènyi spread according to the Bures-Hall ensemble for the simple case of $d_A=2\leq d_B$.
It is known that for a Bures-Hall random state with $d_A=2$ and a generic $d_B$
\begin{equation}
    \ket{\rho}=\sqrt{\lambda} \ket{\phi_A \phi_B}+ \sqrt{1- \lambda} \ket{\chi_A \chi_B}\, ,
\end{equation}
the associated probability distribution is given by $P^{BH}(\lambda)=\frac{2 (1-2 \lambda)^2}{\pi  \sqrt{(1-\lambda) \lambda}}$ \cite{zyczkowski2006GeometryQuantumStates}.
Since the antiflatness for this system is given by
\begin{equation}
    \mathcal{F}_A(\rho_A)=-(1-2 \lambda)^2 (\lambda-1) \lambda\,.
\end{equation}
Our goal is to determine the probability distribution of $f \equiv \mathcal{F}_A(\lambda)$, where $\lambda$ is a random variable following the distribution $P^{BH}(\lambda)$. The probability distribution of $\lambda$ is then given by:
\begin{equation}
    P^{BH}(f)=\sum_{\lambda: \mathcal{F}_A(\lambda)=f} \frac{P^{BH}(\lambda)}{\abs{\mathcal{F}_A'(\lambda)}}\,,
\end{equation}
where $\mathcal{F}_A'(\lambda)$ is the derivative of $\mathcal{F}_A(\lambda)$ with respect to $\lambda$, which can be explicitly calculated as:
\begin{equation}
    \mathcal{F}_A'(\lambda)=(1-2 \lambda) (8 (\lambda-1) \lambda+1)\,.
\end{equation}
Furthermore, the solution to the equation $f = \mathcal{F}_A(\lambda)$ is:
\begin{equation}
    \lambda=\frac{1}{2} \left(1\pm\frac{\sqrt{1\pm\sqrt{1-16 f}}}{\sqrt{2}}\right)\,.
\end{equation}
Finally the probability distribution reads:

\begin{equation}
    P^{BH}_{\mathcal{F}_A(\rho_A)}(f)= \frac{4}{\pi  \sqrt{f(1-16 f)} }.
\end{equation}

\section{Properties of the antiflatness of the gaps}
\label{App : antiflatness of the gaps}
\begin{itemize}
    \item (Faithfulness)  If the probability vector $\vec{\lambda}$ is uniform, then $\Gamma(\vec{\lambda}) = 0$. In other words,  $\Gamma(\vec{\lambda})=0 $ if and only if $\vec{\lambda} \in \text{FLAT}$;
    \item (Positivity) $\Gamma(\vec{\lambda})\geq 0$;
    \item (Subadditivity) $\Gamma(\vec{\lambda}\otimes\vec{\lambda'}) \le \Gamma(\vec{\lambda}) + \Gamma(\vec{\lambda}') \; ; $
    \item (Convexity)
    Let $p = t p^{(1)} + (1-t)p^{(2)}$, with $t\in[0,1].$
    $\Gamma(tp^{(1)} + (1-t)p^{(2)}) \le t \Gamma(p^{(1)}) + (1-t)\Gamma(p^{(2)}).$
    Thus, \(\Gamma\) is convex.
\end{itemize}

\begin{proof}
    \textbf{Faithfulness}
    \newline
    If the probability vector is uniform, \(\lambda_i = 1/d_A\) for all \(i\), then $\Gamma = \frac{1}{N-1}\sum_{i=1}^{d_A}(0)^2 = 0$.
    Conversely, if \(\Gamma=0\), then all differences vanish,$ \lambda_1 = \lambda_2 = \cdots = \lambda_{d_A},$ so the distribution is uniform.  Thus, \(\Gamma=0\) iff the state is free.  
    Hence, \(\Gamma\) is faithful.
    \newline
    \textbf{Monotonicity under free (doubly-stochastic) operations}
    \newline
    In many resource theories (purity, thermodynamics, majorization-based theories), free operations are exactly the noisy, mixing or ``forgetful'' maps, described by doubly stochastic matrices \(D\), $p' = D p.$ A key result (Hardy–Littlewood–Pólya) states: $Dp \prec p$, i.e.\ \(p'\) is more uniform than \(p\). Since the squared finite-difference operator $\Delta_i(p) = \lambda_i - \lambda_{i+1}$ is Schur-convex, we have $\sum_{i} (\Delta_i(p'))^2 \le \sum_{i} (\Delta_i(p))^2.$ Hence $\Gamma(p') \le \Gamma(p)$. Thus, \(\Gamma\) is monotonic under all free (stochastic or doubly-stochastic) operations.
    \newline
    \textbf{Subadditivity}
    \newline
    For tensor-product spectra, the discrete difference is not additive but becomes smaller due to increasing uniformity; specifically, $p(\lambda\otimes\mu) \prec p(\lambda)\otimes p(\mu).$ So we have, $\Gamma(p_{\lambda\otimes\mu}) \le \Gamma(p_\lambda) + \Gamma(p_\mu).$
    \newline
    \textbf{Convexity}
    \newline
    Let $p = t p^{(1)} + (1-t)p^{(2)}$, with $t\in[0,1].$ Finite differences are linear; $\Delta_i(p) = t\,\Delta_i(p^{(1)}) + (1-t)\,\Delta_i(p^{(2)}).$ Since \(x\mapsto x^2\) is convex, $(\Delta_i(p))^2 \le t(\Delta_i(p^{(1)}))^2 + (1-t)(\Delta_i(p^{(2)}))^2.$ 
    \newline
    Summing and multiplying by \(\frac{1}{N-1}\), $\Gamma(tp^{(1)} + (1-t)p^{(2)}) \le t \Gamma(p^{(1)}) + (1-t)\Gamma(p^{(2)}).$
    Thus, \(\Gamma\) is convex.

    Another way to see it is by rewriting $\Gamma$ as a quadratic form. We consider the forward difference matrix $B \in \mathbb{R}^{(d_A-1) \times d_A}$, defined as $(B \lambda)_{i} = \lambda_i - \lambda_{i+1}$ for all $i = 1, \dots, d_A -1$ . To get a quadratic form, notice that $L = B^TB$ is nothing but the discrete laplacian. In particular, we have 
    \begin{equation}
        \sum_{i=1}^{d_A - 1}(\lambda_i - \lambda_{i+1})^2 = \lambda_1^2 + 2 \sum_{i=2}^{d_A - 1} \lambda_i^2 + \lambda_{d_A}^2 - 2 \sum_{i=1}^{d_A - 1} \lambda_i \lambda_{i+1} \; \ ,
    \end{equation}
    so al we need to fin is the matrix L s.t. $\Gamma = \frac{1}{d_A - 1} \lambda^T L \lambda = \frac{1}{d_A - 1} \sum_{ij} L_{ij} \lambda_i \lambda_j$ . The matrix we are looking for is build in such a way, $H_{11} =1 \; ; \; H_{ii} = 2 \; \; \rm{for} \; 2 \le i \le d_A-1 \; ; \; H_{d_A, d_A} = 1$ while the off-diagonal elements are $H_{i, i+1} = H_{i+1, i} = -1$ for $i = 1, \dots , d_A-1$. So the Hessian of $\Gamma$ you can immediately see is 2L which is positive semidefinite, and this concludes the proof.
\end{proof}

\end{document}